\documentclass[sn-mathphys]{sn-jnl}%
\usepackage{hyperref}
\hypersetup{
    colorlinks=true,
    linkcolor=blue,
    citecolor=violet,
    filecolor=magenta,      
    urlcolor=cyan,
}
\usepackage{adjustbox}
\usepackage{bm}
\usepackage{amssymb,amsmath,amsthm,bbm}

\usepackage{wasysym}
\usepackage{enumitem}
\usepackage[capitalize]{cleveref}

\usepackage{graphicx}
\usepackage{tikz-cd}
\tikzcdset{scale cd/.style={every label/.append style={scale=#1},
    cells={nodes={scale=#1}}}}
\usepackage{program}
\catcode`\|=12\relax

\numberwithin{equation}{section}

\newtheorem{theorem}[equation]{Theorem}
\newtheorem{corollary}[equation]{Corollary}
\newtheorem{lemma}[equation]{Lemma}
\newtheorem{remark}[equation]{Remark}

\newtheorem{proposition}[equation]{Proposition}

\newtheorem{definition}[equation]{Definition}
\newtheorem{property}[equation]{Property}

\newtheorem{conjecture}[equation]{Conjecture}

\numberwithin{equation}{section}
\theoremstyle{definition}
\newtheorem{example}[equation]{Example}

\newcommand\josh[1]{\textbf{\textcolor[rgb]{0,.5,1}{[Josh: #1]}}}
\newcommand\truncatedOptionalImportant[1]{}
\newcommand\truncatedOptional[1]{}
\newcommand\optional[1]{}

\renewcommand{\R}{\mathbb{R}}

\newcommand{\X}{{\mathbb X}}
\newcommand{\Y}{{\mathbb Y}}

\newcommand{\CC}{\mathcal{C}}

\newcommand{\MM}{\mathcal{M}}

\newcommand{\RR}{\mathcal{R}}
\renewcommand{\SS}{\mathcal{S}}
\newcommand{\TT}{\mathcal{T}}
\newcommand{\UU}{\mathcal{U}}

\newcommand{\F}{\mathsf{F}}
\newcommand{\G}{\mathsf{G}}

\newcommand{\Rtop}{\ensuremath{\mathbb{R}\text{-}\mathbf{Top}}}
\newcommand{\Reeb}{\ensuremath{\mathbf{Reeb}}}
\newcommand{\Pre}{\ensuremath{\mathbf{Pre}}}
\newcommand{\Csh}{\ensuremath{\mathbf{Csh}}}
\newcommand{\Int}{\ensuremath{\mathbf{Int}}}
\newcommand{\IntCl}{\ensuremath{\overline{\Int}}}
\renewcommand{\Vec}{\ensuremath{\mathbf{Vec}}}
\newcommand{\Set}{\ensuremath{\mathbf{Set}}}

\newcommand{\Ext}{\mathrm{Ext}}
\newcommand{\Rel}{\mathrm{Rel}}
\newcommand{\Ord}{\mathrm{Ord}}
\newcommand{\ExDgm}{\mathrm{ExDgm}}
\newcommand{\Dgm}{\mathrm{Dgm}}
\newcommand{\type}{\mathrm{type}}

\renewcommand{\e}{\varepsilon}
\renewcommand{\phi}{\varphi}
\newcommand{\inv}{^{-1}}

\renewcommand{\Im}{\mathsf{Im}}

\newcommand{\dB}[1]{d_{B}^{#1}}
\newcommand{\db}[1]{d_{b}^{#1}}
\newcommand{\dbc}[2]{d_{b}^{\textnormal{#2}_{#1}}}
\newcommand{\du}{\ensuremath{\delta_E}}
\newcommand{\dfd}{\ensuremath{d_{FD}}}
\newcommand{\di}{\ensuremath{d_{I}}}

\newcommand{\plS}{\ensuremath{\mathsf{PL}}}

\newcommand{\tS}{\ensuremath{\mathsf{T}}\hspace{0.2em}}

\newcommand{\cS}{\ensuremath{\mathsf{C}}\hspace{0.2em}}

\newcommand{\mS}{\ensuremath{\mathsf{M}}\hspace{0.2em}}
\newcommand{\mSB}{\ensuremath{\mathsf{M}^0}\hspace{0.2em}}
\newcommand{\mR}{\ensuremath{{\RR}(\mathsf{M})}\hspace{0.2em}}
\newcommand{\mRB}{\ensuremath{{\RR}(\mathsf{M}^0)}\hspace{0.2em}}
\newcommand{\mC}{\ensuremath{{\CC}(\mathsf{M})}\hspace{0.2em}}

\newcommand{\sX}{\ensuremath{\X}}

\newcommand{\rX}{\ensuremath{\X_f}}

\renewcommand{\sf}{\ensuremath{f}}
\newcommand{\sg}{\ensuremath{g}}
\newcommand{\rf}{\ensuremath{\tilde{f}}}
\newcommand{\rg}{\ensuremath{\tilde{g}}}
\newcommand{\uf}{\ensuremath{\hat{f}}}

\newcommand{\RRf}{\ensuremath{\RR_{\sf}}}
\newcommand{\RRg}{\ensuremath{\RR_{\sg}}}

\newcommand{\limit}{\ensuremath{\mathcal{L}}}

\begin{document}

\title[Reeb Graph Metrics from the Ground Up]{Reeb Graph Metrics from the Ground Up}

\author*[1]{\fnm{Brian} \sur{Bollen}}\email{bbollen23@math.arizona.edu}

\author[2]{\fnm{Erin} \sur{Chambers}}\email{echambe5@slu.edu}

\author[3]{\fnm{Joshua} A. \sur{Levine}}\email{josh@arizona.edu}

\author[4]{\fnm{Elizabeth} \sur{Munch}}\email{muncheli@msu.edu}

\affil[1]{Dept.~of Mathematics, University of Arizona}
\affil[2]{Dept.~of Computer Science, University of Arizona}
\affil[3]{Dept.~of Computer Science, St.~Louis University}
\affil[4]{Dept.~of Computational Mathematics, Science and Engineering, Michigan State University}
\affil[5]{Dept.~of Mathematics, Michigan State University}

\abstract{The Reeb graph has been utilized in various applications including the analysis of scalar fields. Recently, research has been focused on using topological signatures such as the Reeb graph to compare multiple scalar fields by defining distance metrics on the topological signatures themselves. Here we survey five existing metrics that have been defined on Reeb graphs: the bottleneck distance, the interleaving distance, functional distortion distance, the Reeb graph edit distance, and the universal edit distance. Our goal is to (1) provide definitions and concrete examples of these distances in order to develop the intuition of the reader, (2) visit previously proven results of stability, universality, and discriminativity, (3) identify and complete any remaining properties which have only been proven (or disproven) for a subset of these metrics, (4) expand the taxonomy of the bottleneck distance to better distinguish between variations which have been commonly miscited, and (5) reconcile the various definitions and requirements on the underlying spaces for these metrics to be defined and properties to be proven.}


\maketitle

\section{Introduction}

In numerous application fields, there is an increasing need to analyze topological and geometric information about shapes.
Given a real-valued function on a topological space, a commonly used object for such analysis is the Reeb graph, which encodes the changing component structure of the level sets of the object. Given specific restrictions on the function and topological space, the Reeb graph is a 1-dimensional regular CW-complex (i.e. a graph) with a function inherited from the input data. 
Reeb graphs are utilized in a variety of computational topology and topological data analysis (TDA) applications in order to obtain a lower dimensional representation of a structure which maintains topological properties of the original data, such as shape analysis \cite{Escolano2013,Hilaga2001}, data skeletonization \cite{Chazal2014, Ge2012}, and surface denoising \cite{Wood2004}; see the recent survey on scalar field analysis through distances on topological descriptors~\cite{Yan2021} for a more exhaustive list of applications.

The Reeb graph is constructed on input data known as an $\R$-space, which is an assignment of scalar data to each point of a topological space. More formally, we say an $\R$-space is a pair $(\X,f)$, where $\X$ is a topological space and $f:\X \to \R$ is a continuous, scalar-valued function. We then construct the Reeb graph by first defining an equivalence relation on $\X$ by letting (two) elements be equivalent if they have the same function value and lie in the same path connected component of their levelset. We denote the Reeb graph defined on $(\X,f)$ as $\RR_f$.

In data analysis and other applied settings, $\R$-spaces are more commonly referred to as \emph{scalar fields}. While the definitions of these two objects are identical, it is common to think of scalar fields having some additional restrictions on the space, such as requiring a simply connected domain. Many common physical phenomena, such as temperature of a surface or distribution of pressure in a liquid, can be described using scalar fields. 

Increases in computational power and the availability of large data sets has lead to increased interest in comparing such scalar fields. Since Reeb graphs have already been utilized for analyzing single scalar fields, there has been recent interest in defining distances on Reeb graphs so that we can compare multiple scalar fields to one another via this topological summary.

This survey will focus on five distances which have been defined on Reeb graphs: 
the \textbf{bottleneck distance}, which can be graded \cite{Steiner2009} or ungraded \cite{Bjerkevik2016a} \footnote{We note that these two notions of the bottleneck distance are closely related but seem to be often conflated or confused in the literature, which leads to incorrect bounds in some prior work on this topic.  Graded bottleneck distance is the older of the two concepts and arises from early work on comparing persistence diagrams, where dimensions are handled separately when comparing the diagrams.  In contrast, ungraded bottleneck is a more recent contribution from the algebraic side and to the best of our knowledge first arose from considering interlevel set persistence. We discuss these in more depth in Section~\ref{sec:bottleneckDist}, and adopt the terminology of graded versus ungraded in order to clarify the distinction.}, 
the \textbf{interleaving distance} \cite{deSilva2016}, 
the \textbf{functional distortion distance} \cite{Bauer2014}, 
the \textbf{Reeb graph edit distance} \cite{DiFabio2016}, 
and the \textbf{universal distance} \cite{Bauer2020}. 
These distances vary in their goals and properties: bottleneck distance finds matchings between points of persistence diagrams associated to the Reeb graphs; interleaving distance finds approximately height-preserving isomorphisms between the graphs; functional distortion distance measures the amount of distortion to continuously map one Reeb graph into another; the Reeb graph edit distance deforms Reeb graphs into one another through a sequence of steps and assigns a cost to this sequence; the universal distance constructs a diagram of topological spaces which transform one Reeb graph into another and assigns a cost to this diagram.
Each distance has been inspired by different mathematical disciplines such as Banach and metric spaces, category theory, sequence and string matching, and graph theory \cite{DiFabio2016,Bauer2016,Bauer2020}. While the bottleneck distance is usually defined as a distance on persistence diagrams, we will frame it as a distance on Reeb graphs to provide a baseline of comparison to the other four distances. We reserve the term \emph{Reeb graph metrics} to refer to the interleaving, functional distortion, the Reeb graph edit distance, and the universal distance.

When we say that we are defining a ``distance'' or ``metric'' between Reeb graphs, we often mean an \textbf{extended pseudometric}.

\begin{definition}
An \textbf{extended pseudometric} is a function $d: X \times X \to [0,+\infty]$ such that for every $x,y,z \in X$ we have
\begin{enumerate}
    \item $d(x,x) = 0$
    \item $d(x,y) = d(y,x)$
    \item $d(x,z) \leq d(x,y) + d(y,z)$
\end{enumerate}
\end{definition}

The term ``extended'' refers to the metric possibly taking values at $+\infty$, and the term ``pseudo" refers to the fact that $d(x,y) = 0$ does not necessarily imply that $x = y$. We will see in Sec.~\ref{sec:distProp} that each of these Reeb graph metrics attain a value of 0 between Reeb graphs $\RR_f$ and $\RR_g$ if and only if the Reeb graphs are isomorphic -- making these distances an \textbf{extended metric} on the space of isomorphism classes of Reeb graphs.

In Sec.~\ref{sec:basicDef}, we provide preliminary definitions of scalar fields, Reeb graphs, and (extended) persistence diagrams. The following four sections are devoted to the individual distances: Sec.~\ref{sec:bottleneckDist} for the bottleneck distance, Sec.~\ref{sec:interleaving} for the interleaving distance, Sec.~\ref{sec:FDD} for the functional distortion distance, Sec.~\ref{sec:rged} for the Reeb graph edit distance, \cref{sec:universalDistance} for the universal distance. We will use small examples in each of these sections to build the understanding and intuition of the reader and then explore full working examples in Sec.~\ref{sec:examples}.

A commonality between the aforementioned distances is that they are each proven to be \textbf{stable}, meaning that small perturbations to the input scalar field do not result in  large changes in the distance. More specifically, we have the following definition of stability.
\begin{definition}
Let $d$ be a distance on Reeb graphs. We say that $d$ is \textbf{stable} if
\[d(\RRf,\RRg) \leq ||f-g||_{\infty} = \max_{x\in \X}|f(x)-g(x)|,\]
for any two Reeb graphs $\RRf,\RRg$ defined on the same domain $\X$.
\end{definition}
Stability has also been utilized for distances defined other topological descriptors such as the persistence diagram~\cite{Steiner2005,Steiner2009}.

The Reeb graph metrics have also all been shown to be more \textbf{discriminative} than the bottleneck distance. Discriminativity is a property which states that the distance is bounded below by some baseline distance up to some constant $c$.

\begin{definition}
Let $d_0$ and $d$ be two distances defined on Reeb graphs. We say that $d$ is more \textbf{discriminative} than $d_0$ (the baseline distance) if there exists some constant $c$ such that \[d_0(\RR_f,\RR_g) \leq c\cdot d(\RR_f,\RR_g),\]
and there exists no constant $c'$ such that $d_0(\RRf,\RRg) = c'\cdot d(\RRf,\RRg)$, for all Reeb graphs $\RRf,\RRg$.
\end{definition}
Discriminativity intuitively states that the distance $d$ is able to find differences between two Reeb graphs which the baseline distance may not. In conjunction with stability, this leads us to believe that these Reeb graph metrics can be used in various application settings. In Sec.~\ref{sec:distProp} we will conglomerate previously proven theorems of stability, discriminativity, universality and other properties. While many of the previously defined properties apply to each of the Reeb graph metrics, some claims have not yet been explicitly proven. We take this opportunity to provide a more comprehensive list of properties for the distances along with supporting proofs. Fig.~\ref{fig:distanceLandscape} depicts the inequality relationships between these distances.

Lastly, we would like to take note of the computational difficulties and overall complexity of these metrics which introduces challenges for both the applied researcher intending to compute these distances, as well as the newcomer who is attempting to develop an intuition for how these distances operate. In Sec.~\ref{sec:examples}, we use examples to illustrate their complexity and develop the reader's intuition. We then discuss the individual computational hurdles for each distance in Sec.~\ref{sec:discussion}.

While we attempt to make this document as self-contained as possible, we do assume general familiarity with topology and the basics of category theory; some familiarity with persistent homology is also beneficial.  We refer the reader to various references on algebraic topology~\cite{Munkres84}, category theory~\cite{Riehl2017}, and topological data analysis~\cite{DeyWang2021,Oudot2015} for further background and more detailed discussion.

\begin{figure}
    \centering
    \includegraphics[width=\textwidth]{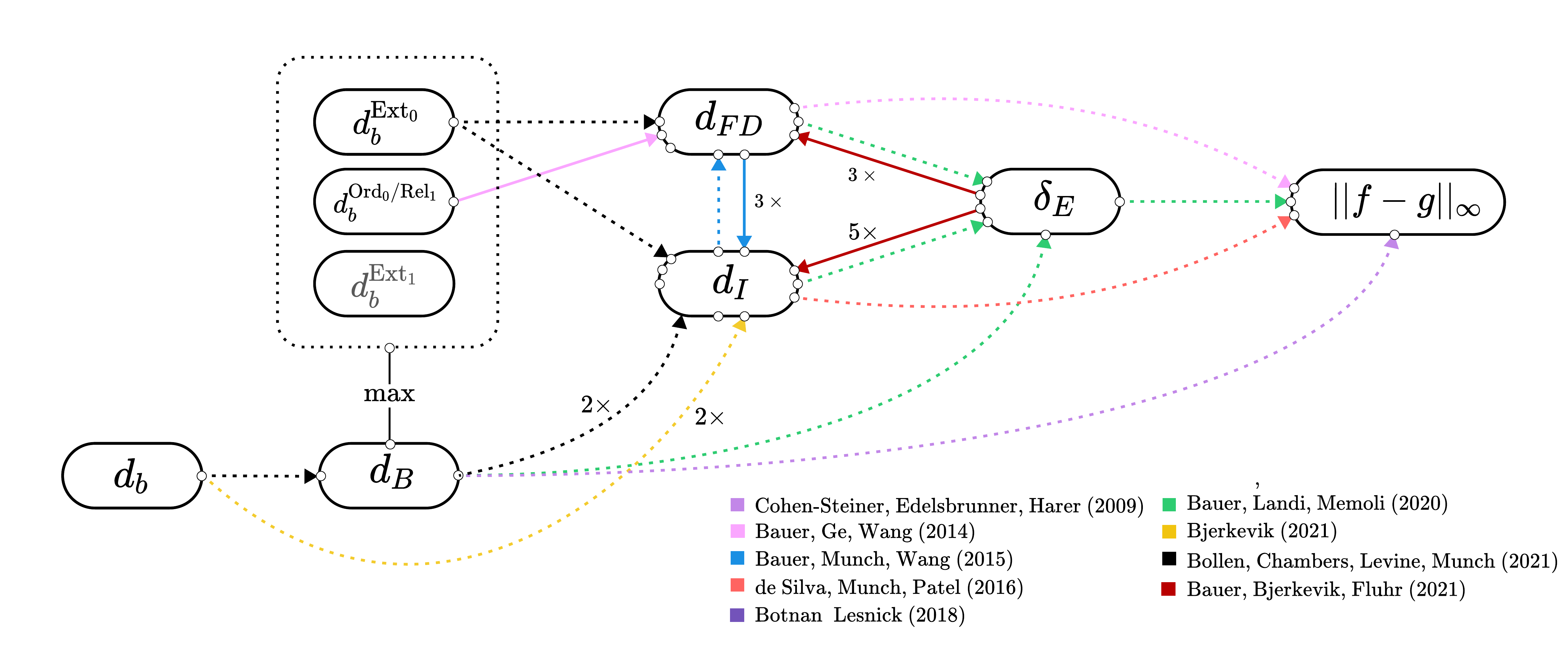}
    \caption{Diagram showing the established relationship between distances defined on Reeb graphs. Here, $d_b$ and $d_B$ indicate the ungraded and graded bottleneck distance, $d_I$ is the interleaving distance on Reeb graphs, $d_{FD}$ is the functional distortion distance, and $\delta_E$ is the universal Reeb graph edit distance.  Color indicates the paper in which these bounds were first proven. Dashed lines indicate that the bounds are tight, while solid lines are not tight. The label $a \xrightarrow{C\times} b$ indicates $a \leq C \times b$.}
    \label{fig:distanceLandscape}
\end{figure}

\subsection{Our contribution}
This work is a constructive survey focusing on these five distances -- the bottleneck distance, the interleaving distance, functional distortion distance, and two edit distances -- as well as an analysis of the properties of each metric so that we can better understand the relationships to one another and use cases for each. Specifically, this paper: 
\begin{itemize}
    \item provides concrete examples for these distances to help develop the intuition of new researchers;
    \item provides returning researchers a reference for fundamental properties of each metric;
    \item compares and contrasts the various metrics and introduces a common nomenclature for their properties in general.
\end{itemize}

\subsection{Notation}
Throughout the literature, Reeb graphs not only have different notation, but the interpretation of them varies. Here, we list the notation we will use throughout this document:
\begin{itemize}
    \item $(\X,f)$ is a constructible $\R$-space/scalar field.
    \item $\RR_f$ is the \emph{topological} Reeb graph of $(\X,f)$ -- the Reeb graph viewed as an $\R$-space $(\X_f,\tilde{f})$, where the space $\X_{f}$ is a topological graph and $\tilde{f}$ is $f$ induced by the quotient to $\X_{f}$. That is, if $p$ is the quotient map carrying $\X$ to $\X_{f}$, then $\tilde{f} = f \circ p$.
    \item $\F$ is the \emph{abstract} Reeb graph of $(\X,f)$ -- the Reeb graph viewed as a cosheaf $\F: \mathbf{Int} \to \mathbf{Set}$.
    \item $\Gamma_f$ is the \emph{combinatorial} Reeb graph of $(\X,f)$ -- the Reeb graph viewed as a labeled multigraph.
    \item $\ExDgm_d(f)$ refers to the $d^{th}$-dimensional extended persistence diagram, and $\ExDgm(f)$ refers to the full extended persistence diagram. 
    We will use the term ``persistence diagram" for the ``full extended persistence diagram" in later sections to avoid cumbersome language, and call the sublevelset persistence diagram the ``ordinary'' diagram when necessary.
\end{itemize}

Each distance treats the Reeb graph differently. However, we will always write our distance measures as simply being a distance between the topological Reeb graphs since we can always convert a topological Reeb graph to its abstract or combinatorial counterpart. Furthermore, unless otherwise noted, we allow for the possibility that two Reeb graphs to be compared 
arise from scalar fields defined on potentially non-homeomorphic spaces, i.e.~$(\X,f)$, $(\Y,g)$ with $\X \not \cong \Y$.

\section{Basic Definitions}
\label{sec:basicDef}

\subsection{Scalar Fields and Reeb Graphs}

\begin{definition}
A \textbf{scalar field} (equivalently an \textbf{$\R$-space}) is a pair $(\X,f)$ where $\X$ is topological space and $f:\X \to \R$ is a continuous real-valued function. The \textbf{dimension} of the scalar field is the dimension of the domain $\X$.
\end{definition}

Without proper restrictions on a scalar field, the Reeb graphs which arise from the scalar field may not be well-behaved. A common criterion that is imposed on a scalar field is restricting $f$ to be a \textbf{Morse} function.

\begin{definition}
A \textbf{non-degenerate} critical point of a smooth function $f$ is any critical point such that the Hessian matrix at that point is nonsingular. A \textbf{Morse function} is a smooth function $f:\X \to \R$ defined on a manifold $\X$ such that all critical points are non-degenerate. We say that a Morse function is \textbf{simple} if all critical points have distinct function values. 
\end{definition}

Morse functions are dense in the space of smooth functions defined on manifolds, i.e. almost all smooth functions on manifolds are Morse. Simulation of simplicity \cite{Edelsbrunner1996} provides algorithmic techniques to turn functions with degenerate critical points into Morse functions. This make the criteria for a function being Morse both useful and practical. 

In what follows, however, many of the definitions and theorems apply to general $n$-dimensional scalar fields whose functions need not necessarily be Morse. Nonetheless, we will focus the examples of this document to \textbf{compact, 2-dimensional manifolds equipped with simple Morse functions}, whose encapsulating set we will denote as $\mS$. Our reasoning is three-fold:

\begin{itemize}
    \item Scalar fields in \mS elicit very specific properties about the types of critical points which will aid in their discussion
    \item Scalar fields in \mS are varied enough to illustrate the Reeb graph metrics properly while being simple enough to visualize
    \item We have tools which can aid in approximating arbitrary functions with simple Morse functions for computational purposes. See Simulation of Simplicity for more details on these techniques\cite{Edelsbrunner1996}.
\end{itemize}

If $(\X,f) \in \mS$, then the scalar field has three types of critical points -- minima, maxima, and saddles. Non-degeneracy of the critical points guarantees that each saddle has exactly two ascending directions and two descending directions.

In this work, there are two different variations of topology that we consider: the level set topology and the sublevel or superlevel set topology.
\begin{definition}
Let $(\sX,\sf)$ be a scalar field and $a\in\R$. The \textbf{level set} of $(\sX,\sf)$ at $a$ is the pre-image of the point $a$ under a function $f$, i.e. $f^{-1}(a) \subseteq \X$. The \textbf{sublevel set} of $(\sX,\sf)$ at $a$ is the pre-image of the interval $(-\infty,a]$ while the \textbf{superlevel set} of $(\sX,\sf)$ at $a$ is the pre-image of the interaval $[a,\infty)$. We denote the sublevel set and superlevel set of $(\sX,\sf)$ at $a$ as $\sX_a$ and $\sX^a$, respectively.
\end{definition}

\begin{figure}
    \centering
    \includegraphics[width=1\textwidth]{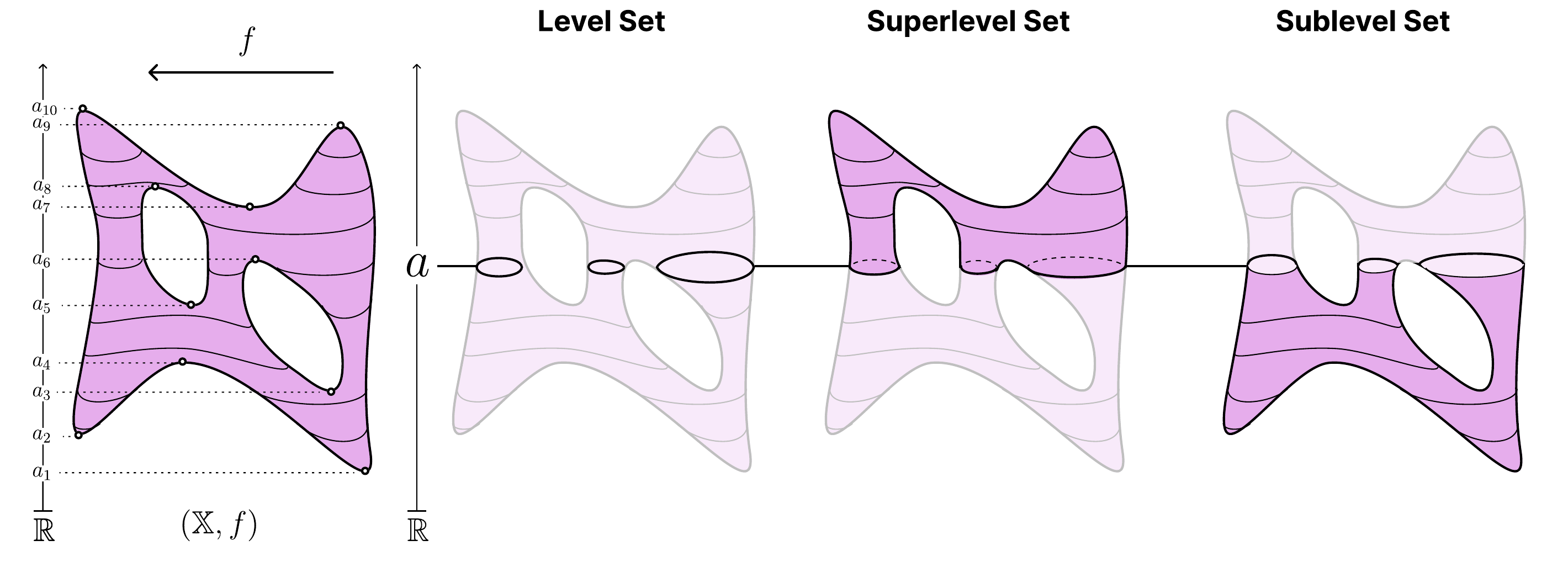}
    \caption{A scalar field $(\X,f)$, where $\X$ is a compact 2-manifold without boundary and $f$ is a Morse function along with a depiction of the level set, superlevel set, and sublevel set at the function value $a$. 
    }
    \label{fig:levelset}
\end{figure}

The Reeb graph is the graph-based topological descriptor which tracks the evolution of the \textit{level set} topology of a scalar field.

\begin{definition}\label{def:reebGraph}
Let $(\sX,\sf)$ be a scalar field. We define an equivalence relation $\sim_f$ on $\sX$ by stating that $x \sim_f y$ if $\sf(x) = \sf(y) = a$ and $x$ and $y$ both lie in the same connected component of the level set $\sf^{-1}(a)$. We define $\rX$ to be the quotient space $\sX / \sim_{\sf}$ and define $\rf:\rX \to \R$ to be the restriction of $\sf$ to the domain $\rX$. The pair $\RRf := (\rX,\rf)$ is called the \textbf{Reeb Graph} of $(\sX,\sf)$. 
\end{definition}

\begin{figure}
    \centering
    \includegraphics[width=1\textwidth]{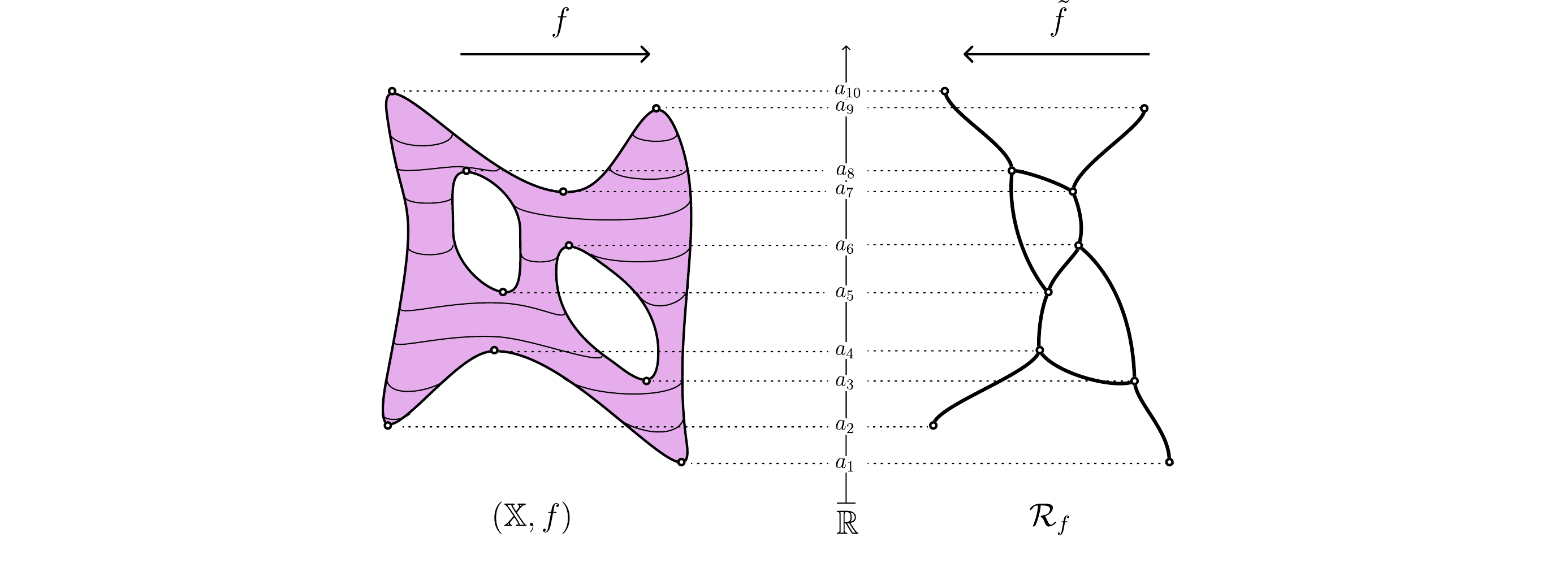}
    \caption{\textbf{(left)} A scalar field $(\X,f)$, where $\X$ is a compact 2-manifold without boundary and $f$ is a Morse function. \textbf{(right)} The Reeb graph $\RR_f$ of the scalar field $(\X,f)$.
    }
    \label{fig:reebGraphEx}
\end{figure}

Note that by this definition, the Reeb graph is itself a scalar field. Throughout this document, we will have instances of using the scalar field notation of a Reeb graph in order to avoid cumbersome notation.

When $\RRf$ arises from a scalar field which is in $\mS$ (and other spaces which we discuss later), the Reeb graph is a 1-dimensional regular CW complex, i.e. a graph. Thus, we will often use the term \textit{vertices} and \textit{edges} of $\RRf$ as each vertex will correspond to a critical point of $(\sX,\sf)$ and edges of $\RRf$ connect critical points to one another via paths where the level set topology does not change. We denote the vertices and edges of the Reeb graph as $V(\RRf)$ and $E(\RRf)$, respectively. We denote the space of Reeb graphs $\RR_f$ whose underlying scalar field $(\sX,\sf)$ is in $\mS$ as $\mR$.

Let $(\sX,\sf)$ be a scalar field and $\RRf$ be its Reeb graph. $\RRf$ has the inherent property that it has at most the same number of 1-cycles as the original space $\sX$ and exactly the same number of connected components. In terms of \textbf{Betti numbers}, we have that 
\[\beta_0(\RRf) =  \beta_0(\sX), \hspace{5em} \beta_1(\RRf) \leq \beta_1(\sX).\]

If $\sX$ is a simply connected domain, then  $\beta_1(\sX) = 0$. This implies that the resulting Reeb graph has no 1-cycles. We call this special case the \textbf{contour tree}.

\begin{definition}
Let $(\sX,\sf)$ be a scalar field where $\sX$ is simply connected, i.e. $\beta_1(\sX) = 0$. Then the resulting Reeb graph $\RRf$ is also simply connected. We call this particular Reeb graph the \textbf{contour tree}. 
\end{definition}

The contour tree having a first-dimensional Betti number of 0 results in graphs which are well-defined trees. Since the Reeb graph metrics and the bottleneck distance will be defined directly on the Reeb graph itself, the definitions of these distances carry over to contour trees. The only difference is that some of the properties exhibited by these distances or relationships between distances are proven to be stronger in the contour tree case. Just as in the Reeb graph case, contour trees of scalar fields in $\mS$ are denoted $\mC$.

A crucial notion of comparison between Reeb graphs is undesrtanding when two Reeb graphs are deemed to be isomorphic. 

\begin{definition}\label{def:functionPreservingReebGraphIso}
We say that a continuous map $\alpha: \RRf \to \RRg$ is a \textbf{function preserving map} if $f = g \circ \alpha$. If there exists a function preserving map which is also a homeomorphism, we say that $\RRg$ and $\RRf$ are \textbf{Reeb graph isomorphic}.
\end{definition}

\begin{figure}
    \centering
    \includegraphics[width=0.9\textwidth]{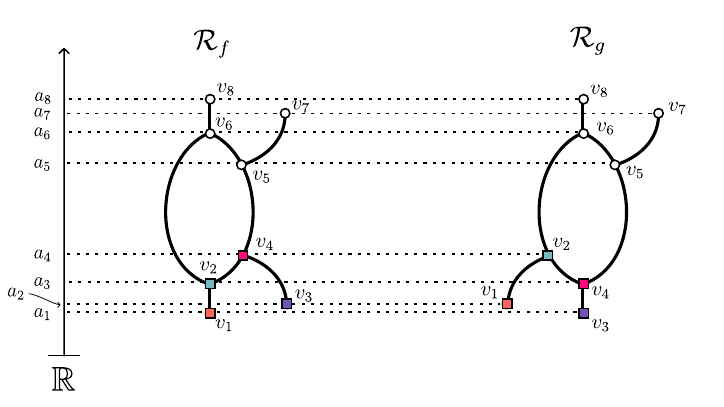}
    \caption{Two Reeb graphs where the function value of the nodes is indicated by height. The labels on the nodes themselves correspond to the labeling which shows that these two Reeb graphs are isomorphic \emph{as graphs}. However, since there is no map between these two Reeb graphs which is an isomorphism and also function preserving, they are \emph{not} Reeb graph isomorphic. }
    \label{fig:graphIsoNonRGIso}
\end{figure}

Note that Reeb graph isomorphism is \textit{not} equivalent to graph isomorphism. Fig.~\ref{fig:graphIsoNonRGIso} depicts an example of two Reeb graphs that are graph isomorphic but not Reeb graph isomorphic. 
In Sec.~\ref{sec:interleaving}, we will see another equivalent definition of Reeb graph isomorphism using the category theory definition of a Reeb graph known as a \textbf{cosheaf}, and in \cref{sec:rged} we will see a version of Reeb graph isomorphism by considering the Reeb graphs as labeled multigraphs.

\subsubsection{Tame and Constructible Scalar Fields}

The literature on Reeb graph metrics is saturated with various criteria for the scalar field which is needed for the respective definitions and theorems. A scalar field being in $\mS$ is sufficient for these criteria to be met, but not necessary. Here, we will discuss tameness and constructibility -- two closely related scalar field restrictions which are used in the definition of bottleneck distance, functional distortion distance, and interleaving distance. 

The notion of \textbf{tameness} is a common condition seen in the literature which essentially guarantees that there is a way to partition the function into sections where the homology does and does not change. This loosens the smoothness condition given by Morse functions.

\begin{definition}\label{def:tame}
We say that a continuous function $f:\X \to \R$, where $\X$ is a topological space, is \textbf{tame} if there is a finite partition $-\infty = a_{-1} < \min(f) = a_0 < \ldots < a_n = \max(f) < a_{n+1} = \infty$ such that the homology groups of the sublevel sets and superlevel sets do not change on any interval between two members of the finite partition. That is, for $s,t \in [a_i,a_{i+1})$ with $s < t$, the homomorphism $H_p(\X_{\leq s}) \rightarrow H_p(\X_{\leq t})$ induced by the inclusion $f^{-1}(-\infty,s] \subseteq f^{-1}(-\infty,t]$ is an isomorphism, and similarly for $s,t \in (a_i,a_{i+1}]$, the homomorphism $H_p(\X_{\geq t}) \rightarrow H_p(\X_{\geq s})$ induced by the inclusion $f^{-1}[t,\infty) \subseteq f^{-1}[s,\infty)$ is an isomorphism. A scalar $(\X,f)$ where $f$ is tame is called a \textbf{tame scalar field}.
\end{definition}

We denote the set of tame scalar fields as \tS. Furthermore, we call the set  $\{a_0,\ldots,a_n\}$ the \textbf{critical values} of the scalar field.

When interleaving distance on Reeb graphs was introduced by de Silva et al. \cite{deSilva2016}, researchers introduced the notion of \textbf{constructibility} which is a stronger restriction on scalar fields than tameness, but still not as restrictive as having smooth functions.

\begin{definition}
\label{def:ConstructibleSpace}
We say that a scalar field $(\X,f)$ is \textbf{constructible} if it is homeomorphic to a scalar field $(\hat{\X},\hat{f})$ which has a finite set $S = \{a_0,\ldots,a_n\}$ (sorted in increasing order) and is constructed in the following way:
\begin{itemize}
    \item For each $0 \leq i \leq n$ there is a space $\mathbb{V}_i$ which is locally path-connected and compact;
    \item For each $0 \leq i \leq n-1$ there is a space $\mathbb{E}_i$ which is also locally path-connected and compact.;
    \item For each $0 \leq i \leq n-1$ there exist continuous attaching maps $\mathbbm{l}_i:\mathbb{E}_i \to \mathbb{V}_i$ (left) and $\mathbbm{r}_i:\mathbb{E}_i \to \mathbb{V}_{i+1}$ (right).
\end{itemize}
We let $\hat{\X}$ be the disjoint union of $\mathbb{V}_i \times \{a_i\}$ and $\mathbb{E}\times [a_i,a_{i+1}]$ by making the identifications $(\mathbbm{l}_i(x),a_i) \sim (x,a_i)$ and $(\mathbbm{r}_i(x),a_{i+1}) \sim (x,a_{i+1})$ for all $i$ and all $x \in \mathbb{E}_i$. The spaces $\mathbb{V}_i \times \{a_i\}$ are known as \textbf{critical fibers} and $\mathbb{E} \times [a_i,a_{i+1}]$ are known as \textbf{non-critical fibers}. The set $S = \{a_0,\ldots,a_n\}$ is known as the set of \textbf{critical values} of the constructible scalar field.
\end{definition}

\begin{figure}
    \centering
    \includegraphics[width=0.9\textwidth]{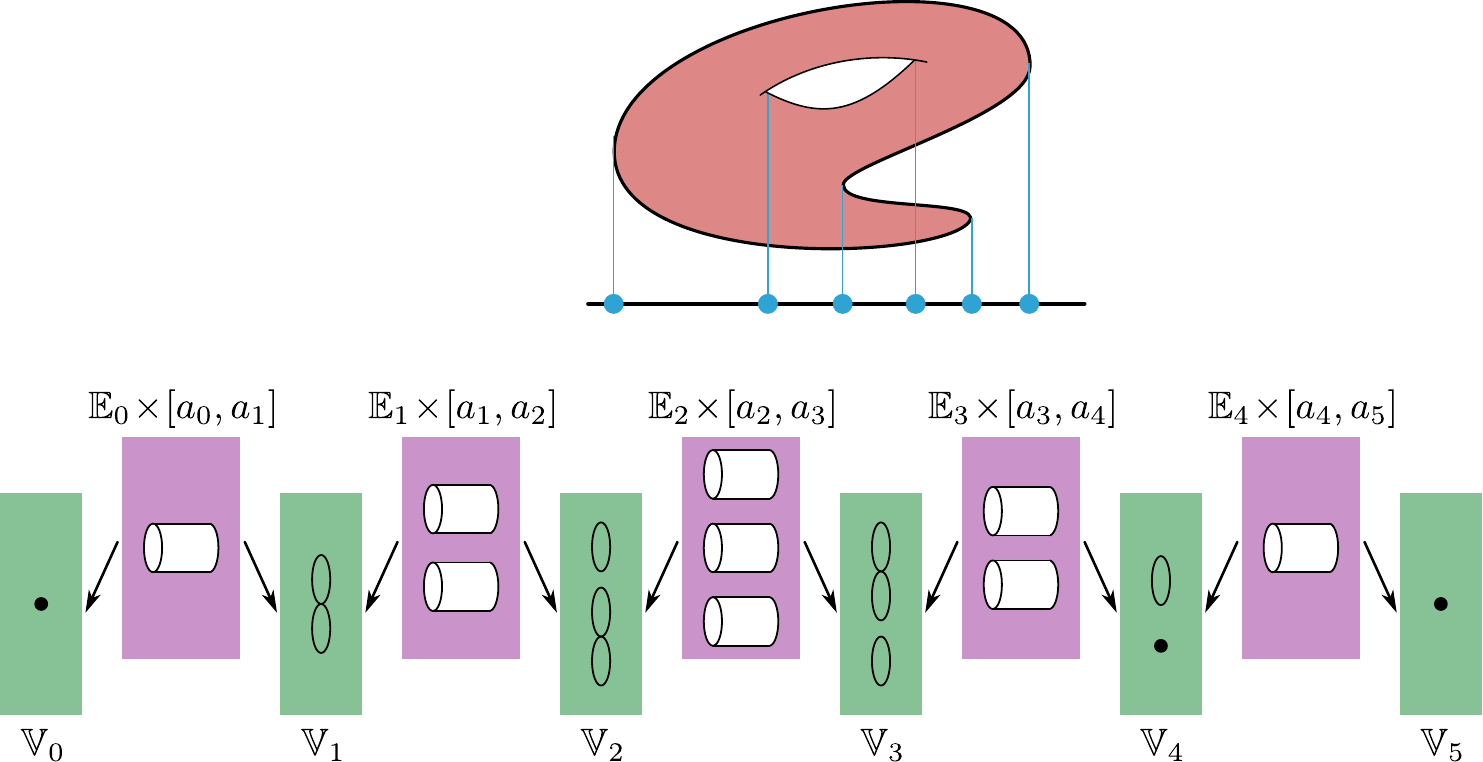}
    \caption{An example of a constructible $\R$-space with notation following Def.~\ref{def:ConstructibleSpace}. Figure from \cite{deSilva2016}.}
    \label{fig:ConstructibleSpace}
\end{figure}

We denote the set of constructible scalar fields as $\cS$.

\begin{proposition}
A constructible scalar field is tame. That is, $\cS \subseteq \tS$.
\end{proposition}
\begin{proof}
The critical values of a constructible scalar field directly correspond to the critical values in the tame sense. Furthermore, by definition, the homology groups of the level sets do not change along non-critical fibers. This implies that the homology groups of the sublevel sets remain constant along these fibers as well. Thus, a constructible scalar field has a finite set of critical values which segments the scalar field into spaces where the sublevel set homology does not change. Thus, the constructible scalar field is also tame. 
\end{proof}

\subsubsection{Piecewise Linear Scalar Fields}

PL scalar fields are a common criteria for scalar fields since they are readily packaged for computation. The intuition is that we can approximate smooth functions using PL functions in order to translate from the smooth domain to the discrete domain. Here, we introduce PL scalar fields since they are used for the construction of the universal distance in \cref{sec:universalDistance}. For the more computational aspects of discretizing scalar fields for topological analysis, we refer the reader to \cite{DeyWang2021} and \cite{diFabio2017}.

\begin{definition}
A \textbf{triangulation} of a space $\X$ is a simplicial complex $K$ paired with a homeomorphism from the underlying space of $K$ to the space $\X$. We say that a space $\X$ is \textbf{triangulable} if there exists a triangulation for $\X$. A \textbf{piecewise linear (PL) function} $f:|K| \to \R$ is a function defined on the underlying space $|K|$ of a triangulated space $\X$ such that $f(x) = \sum_i b_i(x)f(u_i)$, where $u_i$ are the vertices of $\X$ and $b_i(x)$ are the barycentric coordinates of $x$. If $\X$ is a compact, triangulable space and $f$ is a PL function, then we call the pair $(\X,f)$ a \textbf{piecewise linear (PL) scalar field}.
\end{definition}

We denote the space of piecewise linear scalar fields as $\plS$. \textbf{Critical points} of a PL scalar field follow the usual definition of critical point.

\begin{proposition}[de Silva et al.~\cite{deSilva2016}]
A PL function $f$ defined on a compact polyhedron $\X$ is a constructible scalar field.
\end{proposition}

\begin{corollary}
A PL function $f$ defined on a compact polyhedron $\X$ is a tame scalar field.
\end{corollary}

\subsection{Extended Persistence}
\label{sec:extendedPersistence}

One of the most commonly used topological descriptor for analysis is the \textbf{persistence diagram} (or equivalently the \textbf{barcode}) which is a set-based topological descriptor. Given a scalar field (or its Reeb graph), we can imagine sweeping upwards in a linear fashion, keeping track of when particular features appear and disappear. \textbf{Persistence} intuitively captures the length of time that features of a scalar field (or other data sets \cite{Edelsbrunner2003b}) take to disappear once they have been introduced. A persistence diagram retains this information as a multiset of points $(a,b)$ which each represent a feature which is created (\textbf{born}) at $a$ and is destroyed (\textbf{dies}) at $b$.

Here we provide the construction of the persistence diagram and extended persistence diagram. In Sec.~\ref{sec:bottleneckDist} we introduce the bottleneck distance which provides us with a good baseline for the comparison of the distances defined on Reeb graphs because 1) it can be efficiently computed, 2) it shares similar properties to the Reeb graph metrics, and 3) it has already been heavily utilized for comparison of persistence diagrams; see \cite{Yan2021} for a survey which uses bottleneck distance as a discriminative baseline.

Persistence diagrams are often defined directly on the scalar field. However, in order for us to compare the persistence diagrams to the Reeb graph, \textbf{we will define our persistence diagrams directly on the graph-based structures}. To do this, we simply leverage the fact that the Reeb graph is itself a 1-dimensional scalar field. We refer the reader to books on topological data analysis \cite{Oudot2015,DeyWang2021} for a more detailed overview of persistence diagrams defined on higher dimensional spaces and their corresponding distance metrics.

\subsection*{Definition of extended persistence}

Let $(\X,f)$ be a tame scalar field with critical values $\{a_1,\ldots,a_k\}$ and let $\{b_0,\ldots,b_k\}$ be a set of real numbers such that 
\[
b_0 < a_1 < b_1 < a_2 < \ldots < b_{k-1} < a_k < b_k. 
\]
This induces a sequence of nested sublevel sets
\[ 
\emptyset = \X_{b_0} \subset \X_{b_1} \subset \ldots \subset \X_{b_{k-1}} \subset \X_{b_k} = \X,
\] 
called a \textbf{filtration} of the scalar field $(\X,f)$. We can then associate each $\X_{b_i}$ with a corresponding homology group $H_d(\X_{b_i})$ for a fixed dimension $d$ to obtain the sequence 
\[
\emptyset = H_d(\X_{b_0}) \rightarrow H_d(\X_{b_1}) \rightarrow \ldots \rightarrow H_d(\X_{b_{k-1}}) \rightarrow H_d(\X_{b_{k}}),
\]
where each arrow between homology group represents the homomorphism 
$h^{i,j}_d: H_d(\X_{b_i}) \rightarrow H_d(\X_{b_j})$ 
induced by the inclusion $\X_{b_i} \subset \X_{b_j}$. 
Note that we have chosen these $\{b_0,\ldots,b_k\}$ to be specifically interleaved between the critical values of $f$ so that the homology (in dimension $d$) of the sequence changes at every iteration\footnote{The homology groups do not necessarily change, but the only possible places that these sublevel sets have different topologies is when we pass over critical values. Choosing values that are not surrounding the critical values would cause our sequence to have multiple homology groups that are guaranteed to be repeated.}.

The $d^{th}$ Betti number is the rank of the $d$-homology group. Colloquially it quantifies the number of $d$-dimensional holes in the particular topological space $\X$. Thus, the $0$-dimensional Betti numbers tell us the number of connected components, while the $1$-dimensional Betti numbers tell us the number of loops, and so on.

\begin{definition}
The \textbf{$\bm{d^{th}}$-persistent homology groups} are the images of the homology group homomorphisms, $H^{i,j}_d : = \Im(h^{i,j}_d)$ and the \textbf{$\bm{d^{th}}$-persistent Betti numbers} are their corresponding ranks, $\beta^{i,j}_d = \mathsf{Rank}(H^{i,j}_d)$. 
\end{definition}

The persistent Betti numbers now correspond to the number of $d$-dimensional holes present in specific portions of the filtration. As we step through the filtration, the number of features (and thus the persistent Betti numbers) may change. The notion of \textbf{persistence} quantifies this change.

\begin{definition}[{\cite[Definition~3.5]{DeyWang2021}}]
A non-trivial $d$-{th} homology class $\alpha \in H_d(\X_a)$ is \textbf{born} at $\X_i$,$i \leq a$, if $\alpha \in H^{i,a}_d$ but $\alpha \notin H^{i-1,a}_d$. Similarly, a non-trivial $d$-th homology class $\alpha \in H_d(\X_a)$ \textbf{dies} at $\X_j$, $a<j$, if $h^{a,j-1}_d(\alpha)\neq 0$ but $h^{a,j}_d(\alpha) = 0$ . The \textbf{persistence} of the class $\alpha$ is $|a_j - a_i|$. If a class $\alpha$ is born but never dies, then we say that the persistence of the class $\alpha$ is $+\infty$.
\end{definition}

Low persistence features are often (but not always \cite{Bubenik2020}) attributed to \textit{noise} or \textit{insignificant} features of the data being studied, while high persistence features are often associated with \textit{significant} features.
Of course, what determines significant or insignificant features is dependent on the application domain and question being asked of the data. In any case, persistence allows us to scale our view of the data as we see fit. Persistence can be captured through the use of a \textbf{persistence diagram}.

\begin{definition}
Let $\alpha$ be a class which is born at $\X_i$ and dies at $\X_j$. The \textbf{persistence pair} representing this class $\alpha$ is the ordered pair $(a_i,a_j)$ in the extended plane $\bar{\R}^2 := \R\cup\{+\infty\} \times \R\cup\{+\infty\}$. The \textbf{$\bm{d^{th}}$-persistence diagram} of $f$, denoted as $\Dgm_d(f)$, is the multiset of all persistence pairs. 
\end{definition}

Often, a persistence diagram is visualized using a scatter plot with the addition of the line $x=y$ to serve as a visual cue for the persistence values; pairs farther away from the diagonal have a larger persistence value. The ordered pairs of a persistence diagram correspond directly to pairs of critical points in a Morse scalar field, which in turn correspond to vertices in the Reeb graph. For example, index 0 critical points create connected components (0-dimensional homology classes) which are then destroyed by down-forks (saddles with two decreasing edges). However, it is possible that a class $\alpha$ is born at $\X_i$ and never dies, such as a cycle in a Reeb graph corresponding to a $1$-dimensional hole in the scalar field. Such classes are known as \textbf{essential} homology classes and are represented as $(a_i,+\infty)$ in the persistence diagram, where the class is born at $\X_i$.

We can immediately see that this can cause issues by considering two 1-cycles that are born at the same top function value of the loop. 
These two cycles will be represented by two identical persistence pairs, even if the bottoms of the loops are at different function values. 

To alleviate this, we can leverage Poincar\'{e} and Lefschetz duality to create a new sequence of homology groups where we begin and end with the trivial group \cite{Steiner2009}. This guarantees that each homology class that is born will also die at a finite value, replacing all the problematic points paired with $\infty$, and implying that each critical point in the Reeb graph will be matched with another critical point at least once.

Let $H_d(\X,\X^a)$ denote the relative homology group of $\X$ with the superlevel set $\X^a$. From this, we can create a new sequence of homology groups

\begin{align*}
    0 = & H_d(\X_{b_0}) \rightarrow H_d(\X_{b_1}) \rightarrow \ldots \rightarrow H_d(\X_{b_{k-1}}) \rightarrow H_d(\X_{b_k}) \\
    = & H_d(\X,\X^{b_k}) \rightarrow H_d(\X,\X^{b_{k-1}}) \rightarrow \ldots \rightarrow H_d(\X,\X^{b_1}) \rightarrow H_d(\X,\X^{b_0}) = 0,
\end{align*}

which we call the \textbf{$d^{th}$ extended filtration sequence}. We define the sequence $H_d(\X_{b_0}) \rightarrow \ldots \rightarrow H_d(\X_{b_k})$ as the \textbf{upwards} sequence and the sequence $H_d(\X,\X^{b_k}) \rightarrow \ldots \rightarrow H_d(\X,\X^{b_0})$ as the \textbf{downwards} sequence.

Now, the classes which are born and die along this extended sequence can be partitioned into three different groups: 1) the classes which are born and die in the upwards sequence, 2) the classes which are born and die in the downwards sequence, and 3) the classes which are born in the upwards sequence and die in the downwards sequence.

\begin{definition}
Let $(\X,f)$ be a tame scalar field whose critical values are $\{a_1,\ldots,a_k\}$. Let $\{b_0,\ldots,b_k\}$ be a set of real numbers such that $b_0 < a_1 < b_1 < a_2 \ldots < b_{k-1} < a_k < b_k$. Finally, let $S$ denote the $d^{th}$ extended filtration sequence of the scalar field $(\X,f)$. We construct three different multisets of points, denoted as $\Ord_d(f)$, $\Ext_d(f)$, and $\Rel_d(f)$, as follows:
\begin{itemize}
    \item If $\alpha \in H_d(\X_a)$ is born at $\X_{b_i}$ and dies at $\X_{b_j}$, $i < a < j$, then $(a_i,a_j) \in \Ord_d(f)$
    \item If $\alpha \in H_d(\X^a)$ is born at $\X^{b_j}$ and dies at $\X^{b_i}$, $i < a < j$, then $(a_j,a_i) \in \Rel_d(f)$
    \item If $\alpha \in H_d(\X_a)$, is born at $\X_{b_i}$ and dies at $\X^{b_j}$, $i < a$, then $(a_i,a_j) \in \Ext_d(f)$ 
\end{itemize}
Each of these multisets are considered a persistence diagram themselves. We denote this collection of diagrams as the \textbf{$d^{th}$ subdiagrams} of $f$. 
\end{definition}

\begin{definition}
The \textbf{$\bm{d^{th}}$-extended persistence diagram} of $f$, denoted as $\ExDgm_d(f)$, is the union of the persistence pairs contained in the $d^{th}$ subdiagrams $\Ord_d(f)$, $\Rel_d(f)$, and $\Ext_d(f)$. 
Finally, we define the $\textbf{full extended persistence diagram}$ of $f$ as the union of the $d^{th}$-extended persistence diagrams for all dimensions $d$, and denote it as $\ExDgm(f)$.
\end{definition}

Note that in this definition, points of the extended persistence diagram can be both above and below the diagonal. 
Preliminary results in \cite{agarwal2006} showed a pairing between all critical points of a 2-manifold while \cite{Steiner2009} extended this to general manifolds. 
While the Reeb graph is \emph{not} a manifold, we are still guaranteed that each critical point will be matched with another.

Simple Morse functions defined on compact manifolds create Reeb graphs which are \textbf{generic}: each critical value is unique and no single critical point can be the birth and death of two separate persistent homology classes. This prevents critical points of the Reeb graph from belonging to two pairs in the persistence diagram. In most cases, genericness is not needed to prove the theorems in subsequent sections. However, just as we will provide examples which are Morse functions on compact manifolds for conceptual ease, our discussion will only deal with generic Reeb graphs as well.

As we stated before, in order for us to compare the bottleneck distance to the other Reeb graph metrics, we construct the full extended persistence diagram directly on the Reeb graph. As a result, our persistence diagrams will have exactly four subdiagrams~\cite{agarwal2006}: $\text{Ord}_0$, $\text{Ext}_0$, $\text{Rel}_1$, and $\text{Ext}_1$. To  make this distinction clear, we will use the notation $\ExDgm(\tilde{f})$  when discussing the persistence diagram constructed on the Reeb graph (rather than $\ExDgm(f)$), where $\tilde{f}$ is the function defined on the Reeb graph $\RR_f = (\rX,\rf)$.
Below we provide a standard example of constructing this diagram and the process behind it.

\begin{figure}
    \centering
    \includegraphics[width=1\textwidth]{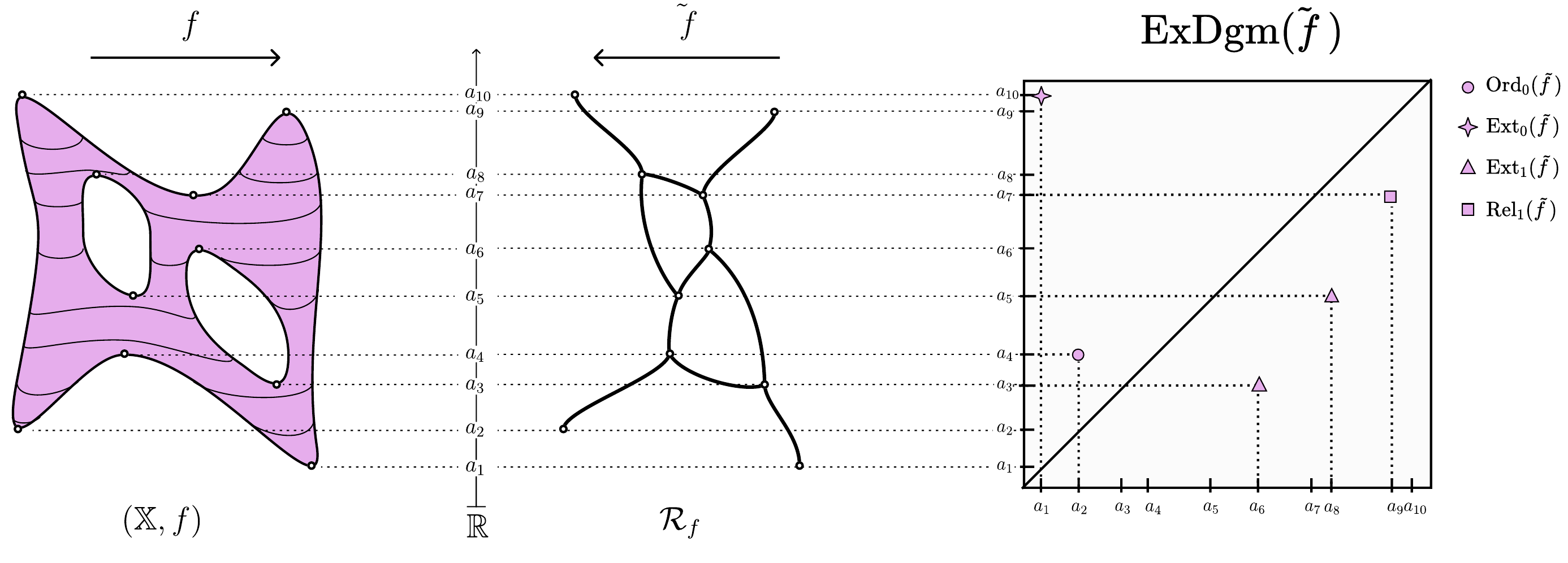}
    \caption{\textbf{(left)} A scalar field $(\X,f)$, where $\X$ is a compact 2-manifold without boundary and $f$ is a Morse function. \textbf{(center)} The Reeb graph $\RR_f$ of the scalar field $(\X,f)$. \textbf{(right)} The extended persistence diagrams for the scalar field. Circles denote the points in $\Ord_0$, squares denote points in $\Rel_1$, triangles denote the points in $\Ext_1$, and stars denote the single point in $\Ext_0$.
    }
    \label{fig:surfaceForPersistence}
\end{figure}

\begin{example}
Fig.~\ref{fig:surfaceForPersistence}(left) shows a genus-2 surface embedded in $\mathbb{R}^3$ along with its Reeb graph (center). Fig.~\ref{fig:surfaceForPersistence}(right) shows the full extended persistence diagram for this Reeb graph. To find the persistence pairs, we begin by sweeping upwards and tracking the features which are born and destroyed.
Critical points $a_1$ and $a_2$ both create classes in $H_0$ by introducing new connected components in the sublevel sets; $a_6$ and $a_8$ both create classes in $H_1$ by closing the cycles which started at $a_3$ and $a_5$, respectively. These two cycles are examples of \textbf{essential homology classes} and so will only be destroyed on the downwards sweep. We pair $a_2$ with $a_4$ since $a_4$ merges two connected components together. 

Going downwards, $a_{10}$ destroys the class in $H_0$ created by $a_1$; $a_9$ creates a relative $H_1$ class which is destroyed at $a_7$; $a_3$ and $a_5$ destroy the classes in $H_1$ created by $a_6$ and $a_8$, respectively. 
Thus, we can summarize the off-diagonal points of the diagram $\ExDgm(f)$ as follows:
\begin{align*}
    \text{Ord}_0 & = \{(a_2,a_4)\} \\
    \text{Ext}_0 & = \{(a_1,a_{10})\} \\
    \text{Rel}_1 & = \{(a_9,a_7)\} \\
    \text{Ext}_1 & = \{(a_6,a_3),(a_8,a_5)\}
\end{align*}
Fig.~\ref{fig:surfaceForPersistence}(b) shows the full extended persistence diagram with accompanying legend.
\end{example}

\paragraph{Fact.} Each subdiagram present in the full extended persistence diagram of $\tilde{f}$ coincides with the corresponding subdiagrams of the full extended persistence diagrams of $f$. That is,
\begin{align*}
    \Ord_0(\rf) & = \Ord_0(\sf) \\
    \Ext_0(\rf) & = \Ext_0(\sf) \\
    \Rel_0(\rf) & = \Rel_1(\sf) \\
    \Ext_0(\rf) & = \Ext_1(\sf) 
\end{align*}

\section{Bottleneck Distance}
\label{sec:bottleneckDist}

Since the topological summary provided by the extended persistence diagram is a useful tool for describing the data encoded in a scalar field, it is natural for us to consider how we would compare two persistence diagrams to one another. The \textbf{bottleneck distance} provides us with such a device. While the bottleneck distance was originally defined in terms of ordinary persistence diagrams, Steiner et al. \cite{Steiner2009} extended the results and definitions to extended persistence diagrams.

The idea behind bottleneck distance is to find a minimal cost matching between the points of the diagrams. 
There are two notions of bottleneck distance that appear in the literature, and since both are called ``bottleneck distance'', it is often difficult to tease out which is being discussed in a given paper. 
The core issue is that, given some collection of persistence diagrams that come either from different dimensions of persistence or from different types of subdiagrams in the extended persistence diagrams, we can either allow only matchings between points of the same subdiagrams, or we can ignore that labeling information and look for matchings that can potentially cross these barriers.  
In order to carefully discuss results in the literature, we therefore separate these ideas into what we call \textit{ungraded} and \textit{graded} bottleneck distances, which we will denote as $d_b$ and $d_B$, respectively.  We note that while the majority of the literature uses the graded bottleneck distances~\cite{DeyWang2021,Carriere2017}, the ungraded version is used to derive stability bounds in more algebraic settings, notably when using interlevel set persistence~\cite{Botnan2018,Bjerkevik2021}; hence, we include both versions in our analysis and discussion.

\subsection{Ungraded and graded bottleneck distance definitions}
Assume we are given two persistence diagrams with no additional labeling information on different types of points. 
The ungraded bottleneck distance finds a best possible matching between these sets of points, allowing a point from one diagram to be matched to the diagonal of the other diagram.   More formally:

\begin{definition}\label{def:UngradedBottleneckDistance}
Let $D_1,D_2$ be two persistence diagrams. Let $\Delta := \{(x,x)\mid x \in \R\}$. We define the  \textbf{(ungraded) bottleneck distance} $\db{}$ between two persistence diagrams $D_1,D_2$ as 
\[
\db{}(D_1,D_2) = \inf_{\zeta} \sup_{x \in D_1\cup\Delta}||x - \zeta(x)||_{\infty},
\]
where $\zeta$ is a bijection between the points of the multisets $D_1\cup\Delta$ and $D_2\cup\Delta$. 
\end{definition}
In this setting, each pair is allowed to be matched to the diagonal which essentially gives the matching a choice of either incurring a cost based on its comparison to another feature of the opposite diagram or incur a cost based solely on its own size.
An alternate viewpoint which achieves the same result but might be more accessible to the combinatorially minded is the following.

\begin{definition}[Combinatorial version]
\label{def:UngradedBottleneckDistance-v2}

Let $D_1,D_2$ be two persistence diagrams and let $\lambda$ denote an \textbf{empty node}. We define $\bar{D_i}:=D_i\cup\{\lambda\}$. A \textbf{matching} $M$ between $D_1$ and $D_2$ is a binary relation $M \subseteq \bar{D_1} \times \bar{D_2}$ such that each element from $D_1$ and $D_2$ appear in exactly one pair $(x,y) \in M$.

The \textbf{cost} of a pair $(x,y)\in M$ is defined as
\[c(x,y) = \begin{cases} 
      \max\{\mid x_1-y_1\mid,\mid x_2-y_2\mid \} & x\in D_1,y\in D_2 \\
      \frac12\mid x_1-x_2 \mid & x\in D_1, y = \lambda \\
      \frac12|y_1-y_2| & x = \lambda, y \in D_2 
  \end{cases}
\]

The cost of a matching $M$, denoted as $c(M)$, is then the largest cost of all pairs in the matching.

The \textbf{(ungraded) bottleneck distance} between $D_1$ and $D_2$ is 
\begin{equation*}
    \db{}(D_1,D_2) = \inf_{M\in\mathscr{M}} c(M) = \inf_{M\in\mathscr{M}}\max_{(x,y)\in M} c(x,y),
\end{equation*}
where $\mathscr{M}$ denotes all possible matchings between $\bar{D}_1$ and $\bar{D}_2$.
\end{definition}

This definition of \textit{ungraded} bottleneck distance is often seen when the input diagrams to be compared each have points of only one (and the same) type. 
For instance, $D_1$ and $D_2$ could each be a single subdiagram, such as $\Ord_0$. 
When this is the case, the matchings obviously only match points of the same type.
However, in the case where multiple types are available but are forgotten, the distance can still be defined. 
Thus, we also give the definition of the \textit{graded} bottleneck distance which can accept multiple types of points in the diagram and ensures that they are differentiated. Recall that a \textbf{subdiagram} of a $d^{th}$-extended persistence diagram is one of three types: $\Ord_d,\Rel_d$ or $\Ext_d$.

\begin{definition}[Graded bottleneck distance]
\label{def:GradedBottleneck}

Let $\ExDgm(f)$ and $\ExDgm(g)$ be two full extended persistence diagrams. Denote the individual ungraded distances between subdiagram types of a fixed dimension $d$ as
$$
\db{\type_d} = \db{}(\type_d(f), \type_d(g)).
$$
Then the \textbf{(graded) bottleneck distance} is 
\begin{equation*}
  \dB{}(\ExDgm(f),\ExDgm(g)) = \max_{d\geq 0}\bigg( \max\Big\{\db{\Ord_d},\db{\Rel_d},\db{\Ext_d}\Big\}\bigg)  
\end{equation*}

\end{definition}

As previously noted, the majority of work available uses (often implicitly) this graded bottleneck distance when comparing diagrams with multiple types, and so it is the one more commonly introduced in the literature~\cite{Carriere2017,DeyWang2021}.
Sometimes the (graded) bottleneck distance is equivalently defined by writing a more combinatorial definition that looks like Def.~\ref{def:UngradedBottleneckDistance-v2} but requires that matchings go between points of the same type. 
However, when carefully checking definitions, some of the algebraic literature, particularly that exploring the idea of \textbf{interlevel set persistence}~\cite{Bjerkevik2016a,Botnan2018}, is in fact using an ungraded version on full extended persistence diagram information. 
We defer a full discussion of the details of interlevel set persistence to Sec.~\ref{sec:appx:interlevelset}, but note that this distinction will be important when discussing the relationship between available metrics in Sec.~\ref{sec:distProp}.

We note one immediate relationship between the ungraded and graded bottleneck distances.
\begin{proposition}
\label{prop:bottleneckBoundedByBottleneck}
Let $D_1$ and $D_2$ be two full extended persistence diagrams. Then,
\[d_b(D_1,D_2) \leq d_B(D_1,D_2).\]
\end{proposition}
\begin{proof}
A matching which preserves labeling as used for $d_B$ is  also a valid matching which does not preserve labeling.
\end{proof}

However, an optimal matching with labeling might not be optimal when types are ignored, making $d_b$ potentially smaller, proving the inequality.

Our definition of graded bottleneck distance is defined solely on the full extended persistence diagram. One can see that this definition is possibly more restrictive than in needs to be. Instead, one could define the graded bottleneck distance on any persistence diagram where we take into account  which subdiagram -- both in type and dimension -- the points in the persistence diagram come from. We will refrain from this explicit definition, since the definition we have already provided is sufficient for our use case.

\subsection{Viewing the bottleneck distances as a distance on Reeb graphs}

Since we are constructing our persistence diagrams from the Reeb graph, we can  view the bottleneck distance as a distance between Reeb graphs. 
This will allow us to directly compare the properties of bottleneck distance to the properties of the Reeb graph metrics that we will define in Sec.~\ref{sec:interleaving},Sec.~\ref{sec:FDD}, Sec.~\ref{sec:rged} and Sec.~\ref{sec:universalDistance}. 
\begin{definition}
The \textbf{(ungraded/graded) bottleneck distance} between two Reeb graphs is the (ungraded/graded) bottleneck distance between the extended persistence diagrams. 
Specifically, 
\begin{align*}
    \db{}(\RRf,\RRg) &:=  \db{}(\ExDgm(\rf),\ExDgm(\rg)),\\
    \dB{}(\RRf,\RRg) &:=  \dB{}(\ExDgm(\rf),\ExDgm(\rg))
\end{align*}
where the first does not take the types of points into account, while the second does. Similarly, we denote the ungraded bottleneck distance between individual subdiagrams of the Reeb graphs as \[d_b^{\type_d}(\RRf,\RRg) := d_b(\type_d(\rf),\type_d(\rg)).\]
\end{definition}

\begin{example}
Recall that, given a scalar field $(\X,f)$, the resulting Reeb graph is defined as $\RRf = (\rX,\rf)$. To compute the full extended persistence diagram of the Reeb graph, we then use the induced function $\rf$. This diagram has only four different subdiagrams: $\Ord_0,\Rel_1,\Ext_0$ and $\Ext_1$. \cref{fig:BottleneckDistancesAreDifferent} illustrates a scenario in which the inequality of \ref{prop:bottleneckBoundedByBottleneck} is strict.
When comparing the two resulting Reeb graphs, we see that the up-leaf of $\RRg$ is matched with the loop of $\RRf$, making $d_b(\RRf,\RRg) = |a_4-a_3|$. However, when labeling is taken into account, we have that $d_B(\RRf,\RRg) = \frac12|a_3-a_2| > |a_4-a_3|$ since both must be matched to the diagonal.
\end{example}

\begin{figure}
    \centering
    \includegraphics[width=0.9\textwidth]{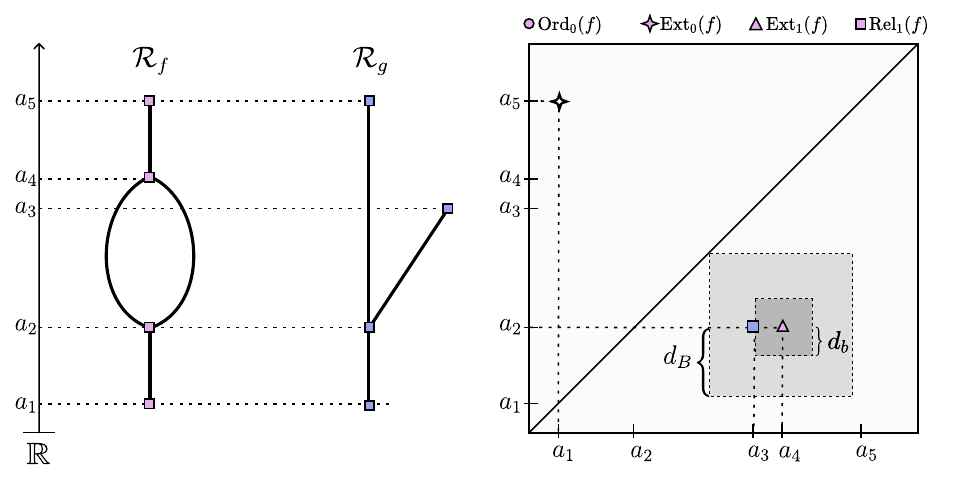}
    \caption{Two Reeb graphs with non-equal graded and ungraded bottleneck distances. Specifically, $d_b = \mid a_4 - a_3\mid < \frac12 \mid a_3 - a_2\mid = d_B$.
    }
    \label{fig:BottleneckDistancesAreDifferent}
\end{figure}

\begin{remark}
\label{rem:globalpairing}
The graded bottleneck distance does not require that the points in $\Ext_0$ (the global max - global min pairs) be matched to each other. Indeed, both may be matched to the diagonal (i.e. matched to the empty node). Suppose we have two Reeb graphs $\RRf,\RRg$. Let $(x_1,x_2)\in \Ext_0(\rf)$ and $(y_1,y_2) \in \Ext_0(\rg)$. Given their persistence diagrams, there is some threshold value $\e$ such that $d_B(\ExDgm(\rf),\ExDgm(\rg + \delta)) = A$, for all $\delta >= \e$, where $A = \max\{\frac12|x_1-x_2|,\frac12|y_1-y_2|\}$. More generally, this implies that if the global maxima and global minima differ by a large enough value, then it will always be best to delete or insert the entire diagrams. As we will see, this is unlike the Reeb graph metrics which will take the global min and maxima difference into account regardless of the this difference.
\end{remark}

\section{Functional Distortion Distance}
\label{sec:FDD}

\subsection*{History}

The functional distortion distance was first defined as a metric by Bauer et al.~\cite{Bauer2014}, inspired from the well-known Gromov-Hausdorff distance which was first introduced by Mikhael Gromov \cite{Gromov1981}. The Gromov-Hausdorff distance is a way to measure the distance between two Banach or metric spaces. More formally, suppose $A$ and $B$ are two metric spaces and let $i_A:A \to Z$ and $i_B:B \to Z$ be isometric embeddings into a common metric space $Z$. We can then find the Hausdorff distance between the embedded spaces: $d_H(i_A(A),i_B(B))$.  

The goal of the Gromov-Hausdorff distance is then to find the minimum Hausdorff distance achieved when ranging over all possible embeddings and the common space to which they are embedded. Intuitively, we are trying to determine a common space where we can embed both $A$ and $B$, while preserving the integrity of the spaces (hence the embeddings being isometries), such that the $A$ and $B$ fit nicely together. We can picture $A$ and $B$ as being two crumpled up pieces of paper (of varying sizes) and our common metric space to be a flat surface. One way to measure the difference in sizes between $A$ and $B$ is to stretch both out flat onto the surface and then compare them in their flattened form. Trying to determine the sizes while the paper is still crumpled would be a much more difficult task. 
The Gromov-Hausdorff (GH) distance has multiple different equivalent definitions; see \cite{Memoli2008}. Functional distortion distance (FDD) borrows from a very specific variation of the GH distance.

\subsection*{Definition}

Here, we break down the definition of FDD into several parts which we will stitch together to form the final definition. We begin by noting that the functional distortion distance is defined on tame scalar fields and we will assume that each Reeb graph discussed here is a tame Reeb graph; see Def.~\ref{def:tame}. Furthermore, we would like to remind the reader that given $\RRf,\RRg$, the induced functions are denoted as $\rf,\rg$, respectively.

\begin{definition}
Let $u,v \in \RRf$ (not necessarily vertices) and let $\pi$ be a continuous path between $u$ and $v$, denoted $u\rightsquigarrow v$. The \textbf{range} of this path is the interval $\emph{range}(\pi) = [\min_{x\in\pi}\rf(x),\max_{x\in\pi}\rf(x)]$. The \textbf{height} is the length of the range, denoted $\emph{height}(\pi) = \max_{x\in\pi}\rf(x) - \min_{x\in\pi}\rf(x)$. We define the distance between $u$ and $v$ to be \[d_{\rf}(u,v) = \min_{\pi:u\rightsquigarrow v}\emph{height}(\pi),\]
where $\pi$ ranges over all continuous paths from $u$ to $v$.
\end{definition}
\begin{definition}
Let $\Phi:\RRf \to \RRg$, $\Psi:\RRf \to \RRg$ be two continuous maps.\footnote{These continuous maps need not be function preserving like in \cref{def:functionPreservingReebGraphIso}.} We define the $G(\Phi,\Psi)$ as \[G(\Phi,\Psi) = \{(x,\Phi(x)):x\in\RRf\}\cup\{(\Psi(y),y):y\in\RRg\}.\]
\end{definition}
$G(\Phi,\Psi)$ is the union of the two graphs of $\Phi$ and $\Psi$.
\begin{definition}
The \textbf{point distortion} $\lambda$ between $(x,y),(x',y') \in G(\Phi,\Psi)$ is defined as \[\lambda((x,y),(x',y')) = \frac12|d_{\rf}(x,x') - d_{\rg}(y,y')|.\]
The \textbf{map distortion} $D(\Phi,\Psi)$ between $\RRf$ and $\RRg$ is the supremum of point distortions ranging over all possible pairs in $G(\Phi,\Psi)$. That is, \[D(\Phi,\Psi) = \sup_{(x,y),(x',y') \in G(\Phi,\Psi)}\lambda((x,y),(x',y')).\]
\end{definition}
In essence, this distance will measure how much each Reeb graph is being altered to map into the other Reeb graph. For example, if $x$ and $x'$ are relatively close (in terms of $d_{\rf}$) but their outputs under $\Phi$ are far apart (in terms of $d_{\rg}$), then this distance will be larger. This distance finds the worst case scenario where two points are close on one Reeb graph and their correspondences are far on the other. Keep in mind that there are two maps $\Phi$ and $\Psi$. Thus, even if both maps are non-surjective, every point of each Reeb graph still has at least one correspondence in $G(\Phi,\Psi)$.

\begin{example}\label{example:FDDExample}
Fig.~\ref{fig:MapDistortionExamples} displays two different pairs of continuous maps between $\RRf$ and $\RRg$. The functions $f$ and $g$ are height functions which map each point of $\RRf$ and $\RRg$ horizontally to the real line. So, for example, $\rf(a_1) = \rg(b_1)$ and $\rf(a_3) = \rg(b_4)$. We assume here that the spaces $\X$ and $\Y$ are compact, 2-manifolds without boundary.

In $\mathbf{(a)}$, the map $\Phi$ maps like an isometry up to $a_2$, where it then maps the rest of $\RRf$ into the leaf of $\RRg$.  The point distortion between the points $(a_3,b_3)$ and $(a_2,b_2)$ would then be $|d_{\rf}(a_3,a_2)-d_{\rg}(b_3,b_2)| = \rf(a_3)-\rg(b_3)$. For $\Psi$, the map also acts like an isometry, except the leaf is collapsed, mapping horizontally to $\RRf$. In this case, $G(\Phi,\Psi)$ contains the points $(a'_2,b'_3)$ and $(a'_2,b_3)$. Thus, the point distortion between these would be $|d_{\rf}(a'_2,a'_2) - d_{\rg}(b'_3,b_3)| = \rg(b_3)-\rg(b_2)$. We can check that other pairs of points from the supergraph will not lead to a higher map distortion. Therefore, $D(\Phi,\Psi) = \max\{\rf(a_3)-\rg(b_3),\rg(b_3)-\rg(b_2)\}$.

\end{example}

\begin{definition}
The \textbf{functional distortion distance} is defined as \[d_{FD}(\RRf,\RRg) = \inf_{\Phi,\Psi}\max\{D(\Phi,\Psi),||\rf-\rg\circ\Phi||_{\infty},||\rf\circ\Psi -\rg||_{\infty}\},\] where $\Phi$ and $\Psi$ range over all continuous maps between $\RRf$ and $\RRg$.
\end{definition}

If $\Phi$ and $\Psi$ are translations or negations, the distances $d_{\rf}$ and $d_{\rg}$ are not affected since these will preserve the relative closeness of the pairs $(x,x')$ and $(y,y')$. In this case, $D(\Phi,\Psi)$ would be 0. The two terms $||\rf-\rg\circ\Phi||_{\infty}$ and $||\rf\circ\Psi-\rg||$ are introduced to address this fact. They measure the length of the translation or similar isometries. See Fig.~\ref{fig:fddIsometryExamples}(a) for an example of translation, and Fig.~\ref{fig:fddIsometryExamples}(b) for a negation example.

\begin{figure}
    \centering
    \includegraphics[width=0.9\textwidth]{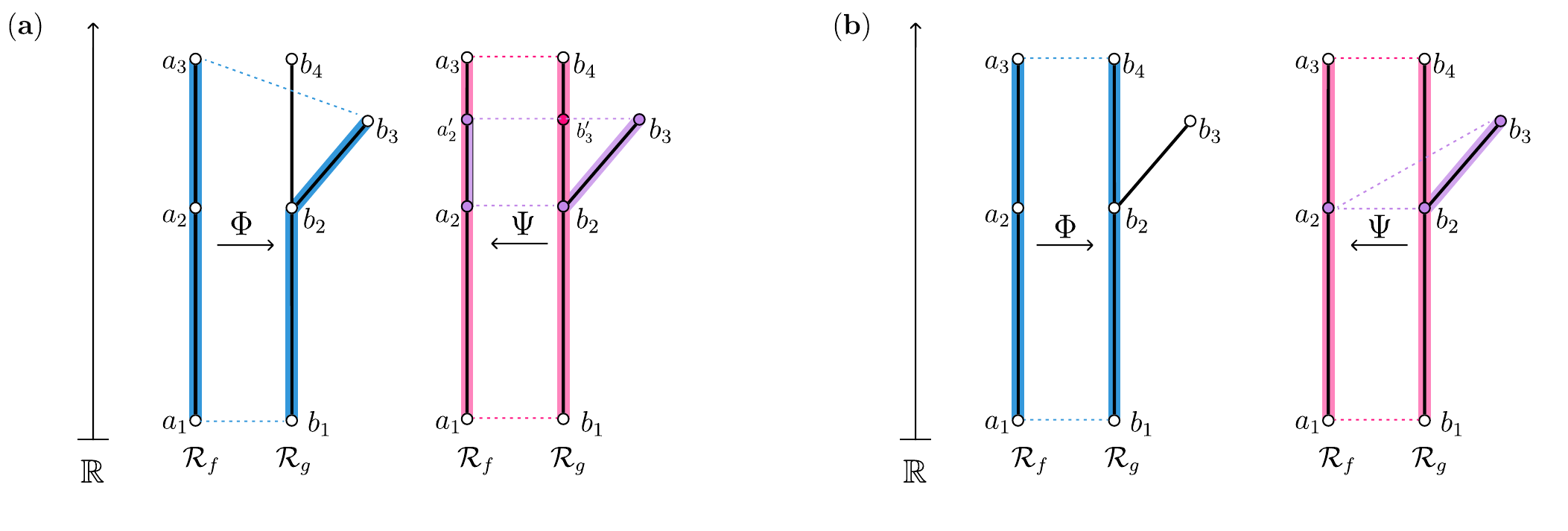}
    \caption{Two examples of continuous maps between Reeb graphs $\RRf$ and $\RRg$. Dotted lines indicate the correspondences that these maps are creating, while the color indicates which section we are referring to.  $\mathbf{(a)}$ The map $\Phi$ distorts $\RRf$ to fit into the leaf on the right side. The map $\Psi$ maps each value straight across. $\mathbf{(b)}$ The map $\Phi$ is an isometry, which will not alter the distortion alone. The distortion will be solely be based on $\Psi$, which is almost an isometry except for collapsing the leaf to a single point. It turns out that the map contracting the leaf to a single base point at $a_2$ will result in an equivalent distortion value to mapping it horizontally as in in $\mathbf{(a)}$.}
    \label{fig:MapDistortionExamples}
\end{figure}

In Fig.~\ref{fig:MapDistortionExamples}(b), the map $\Phi$ is an isometry and therefore no points in the supergraph which come from $\Phi$ will contribute to the distortion value. The map $\Psi$ is close to an isometry besides contracting the leaf to a single point. The point distortion between the pairs $(a_2,b_2)$ and $(a_2,b_3)$ is simply $\rg(b_3)-\rg(b_2)$. Note that this is the same distortion value which was achieved in part $\mathbf{(a)}$ when we mapped the leaf straight across to $\RRf$. This comes from the definition of $d_g$ which only looks at the height of the path that is traversed from one point to the next and does not take into account the total distance traversed. However, note that contracting this leaf to a single point at the base of the leaf causes the value of $||\rf\circ\Psi-\rg||_{\infty}$ to be at least $\rg(b_3)-\rf(a_2)$. In Fig.~\ref{fig:MapDistortionExamples}$(a)$, this value is $||\rf\circ\Psi-\rg||_{\infty} = 0$ since $\Psi$ is \textbf{function preserving}.

\begin{figure}
    \centering
    \includegraphics[width=0.85\textwidth]{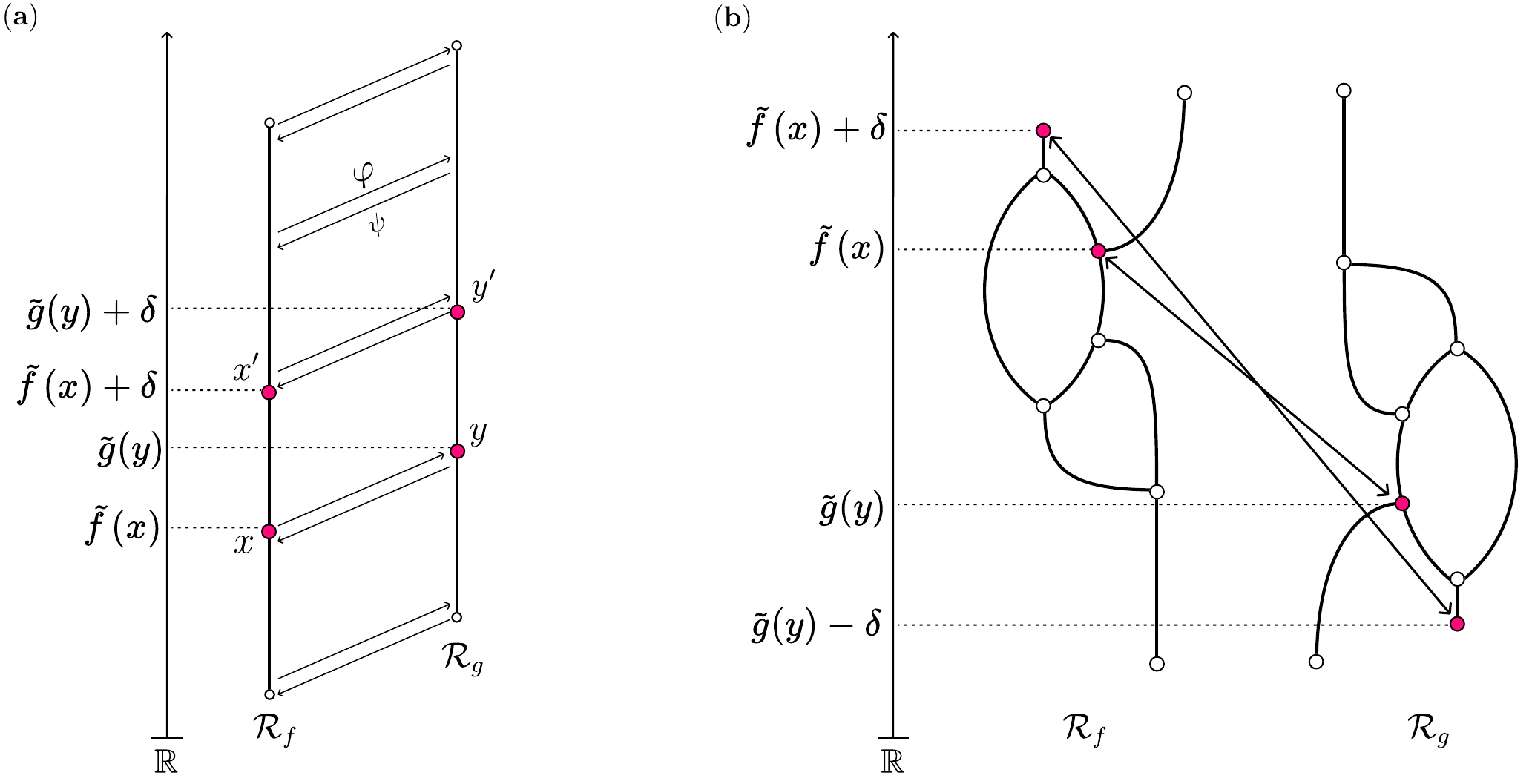}
    \caption{$\mathbf{(a)}$ An example of maps where $\Phi$ is a translation and $\Phi = \Psi^{-1}$. In this, $d_{\rf}(x,x') = \delta = d_{\rg}(y,y')$, which implies that $D(\Phi,\Psi) = 0$. Without the $||\rf-\rg\circ\Phi||_{\infty}$ and $||\rf\circ\Psi -\rg||_{\infty}$ terms, the functional distortion distance between these two Reeb graphs would be 0. With these terms included, we can see that the functional distortion distance will be $\delta$, which is equal to the magnitude of the translation. $\mathbf{(b)}$ An example of maps where $\Phi$ and $\Psi$ are negations of each other. Because of this, the point distortion between any two pairs of points will always be 0.}
    \label{fig:fddIsometryExamples}
\end{figure}

\section{Interleaving Distance}
\label{sec:interleaving}

\subsection{History}

The interleaving distance on Reeb graphs takes root in earlier work that defines the interleaving distance for persistence modules \cite{Chazal2009b},
and is heavily inspired by the subsequent category theoretic treatment \cite{Bubenik2014,Bubenik2014a}. 
It intuitively captures distance between two \textbf{cosheaves} $\F$ and $\G$ representing the two Reeb graphs by defining an approximate isomorphism between them, known as an $\e$-interleaving. The distance is then defined as being the infimum $\e$ such that there exists an $\e$-interleaving between $\F$ and $\G$.

In order to utilize interleavings to define a distance measure on Reeb graphs, we can encode the data of a Reeb graph in a constructible set-valued cosheaf \cite{Curry2014,Curry2015a,Curry2016}.
In fact, it is known that this metric is a special case of a more general theory of interleaving distances given on a \textit{category with a flow} \cite{deSilva2018,Stefanou2018, Cruz2019}; this more general theory also encompasses other metrics including the $\ell_\infty$ distance on points or functions, regular Hausdorff distance, and the Gromov-Hausdorff distance \cite{Stefanou2018,Bubenik2017a}.

Interleaving metrics have been studied in the context of
$\R$-spaces \cite{Blumberg2017},
multiparameter persistence modules \cite{Lesnick2015},
merge trees \cite{Morozov2013},
formigrams \cite{Kim2017,Kim2019a},
and on more general category theoretic constructions \cite{Botnan2020,Scoccola2020a}, as well as developed for Reeb graphs~\cite{deSilva2016,Chambers2021}.
There are also interesting restrictions to labeled merge trees, where one can pass to a matrix representation and show that the interleaving distance is equivalent to the point-wise $\ell_\infty$ distance
\cite{Munch2019,Gasparovic2019,Yan2019a,Stefanou2020}. Furthermore, the interleaving distance has been used in evaluating the quality of the mapper graph \cite{Singh2007}, which can be proven to be an approximation of the Reeb graph using this metric \cite{Munch2016,Brown2019}.

\subsection{Definition}

We can construct a category of Reeb graphs, denoted as $\mathbf{Reeb}$, by defining the objects to be constructible Reeb graphs and the \textbf{morphisms} to be continuous, function preserving maps; see Def.~\ref{def:functionPreservingReebGraphIso}. It turns out that this category is equivalent to the category of \textbf{cosheaves}. We provide an abridged construction of the category of cosheaves by first introducing the notion of \textbf{pre-cosheaves}.

A pre-cosheaf for our purposes\footnote{In general, we can define pre-cosheaves where the domain category is the open sets of any topological space and the range category is unrestricted, but is most commonly defined to be the category of vector spaces.} is a functor $\F$ from the category of connected open intervals on the real line $\mathbf{Int}$ to the category of sets $\mathbf{Set}$. Intuitively, it is a way to assign data to the open intervals of the real line in a way that respects inclusion of the intervals. Given a constructible Reeb graph $\RR_f$, we can construct its pre-cosheaf $\F$ by the formulas 
\[\F(I) = \pi_0(f^{-1}(I)), \hspace{0.5in} \F[I\subseteq J] = \pi_0[f^{-1}(I) \subseteq f^{-1}(J)],\] 
where $\pi_0(U)$ is the set of path connected components of the space $U$. Figure \ref{fig:reebToCosheaf} provides a depiction of converting a Reeb graph $\RR_f$ to its pre-cosheaf $\F$. The \textbf{category of pre-cosheaves} is denoted as $\Pre$.

It turns out that the Reeb graph satisfies additional ``gluing'' constraints which guarantee that its pre-cosheaf is actually a well-defined \textbf{cosheaf}; see \cite{deSilva2016} for this equivalence and \cite{Curry2014} for a rigorous treatment of cosheaves and their applications to topological data analysis. For the remainder of this document, we will refer to $\F$ as a cosheaf rather than a pre-cosheaf. The category of cosheaves, denoted as $\Csh$, is a proper subcategory of $\Pre$, defined as the pre-cosheaves which satisfy the aforementioned constraints.

Stating that two Reeb graphs $\RR_f,\RR_g$ are Reeb graph isomorphic is equivalent to stating that their associated cosheaves $\F,\G$ are isomorphic. Recall that an isomoprhism between two functors is a pair of natural transformations $\phi:\F \Rightarrow \G$, $\psi:\G \Rightarrow \F$ such that $\psi_I \circ \phi_I = \mathbf{Id}_{\F(I)}$ and $\phi_I \circ \psi_I = \mathbf{Id}_{\G(I)}$, for all $I \in \mathbf{Int}$. If we do not have a true isomorphism between these cosheaves, we can approximate the isomorphisms to form an $\e$-\textbf{interleaving}. Fig.~\ref{fig:nonZeroInterleaving} shows the same Reeb graphs as Fig.~\ref{fig:graphIsoNonRGIso} to demonstrate why these graphs fail to be isomorphic as cosheaves. Note that the dashed lines in this figure are the maps $\phi_I,\phi_J$ and $\phi_K$. For there to be an isomorphism between the cosheaves $\F$ and $\G$ in Fig.~\ref{fig:nonZeroInterleaving}, we need to define a natural transformation $\eta: \F \Rightarrow \G$ such that $\eta_U$ is an isomorphism for all open intervals $U \subset \R$. The arrows from $\F(I)$ and $\F(J)$ to $F(K)$ denote the maps induced by the inclusions $I, J \subset K$ (and similarly for $\G$). We can see that it is possible for us to define isomorphism for $\F(J)$ and $\F(K)$ which respect these inclusions, but there exists no isomorphism between $\F(I)$ and $\G(I)$ that does. Thus any map between these Reeb graphs that respect these inclusions will be a non-injective function.

\begin{figure}
    \centering
    \includegraphics[width=\textwidth]{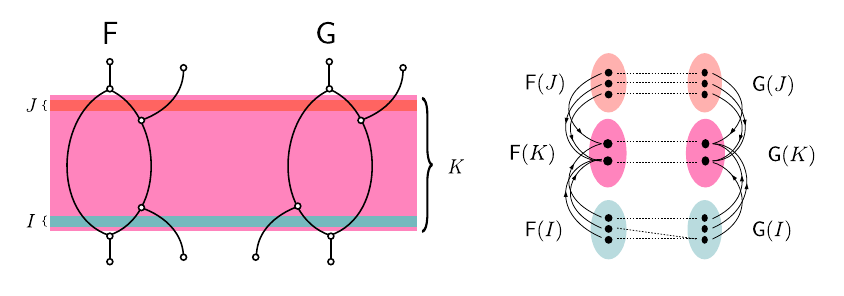}
    \caption{Two cosheaves $\F$ and $\G$ with three specifically chosen intervals $I,J \subset K$, with $I$ and $J$ disjoint to display that these Reeb graphs are graph isomorphic but not \emph{cosheaf isomorphic} -- and equivalently not Reeb graph isomorphic.}
    \label{fig:nonZeroInterleaving}
\end{figure}

\begin{definition} Let $I = (a,b) \subseteq \mathbb{R}$ and $I^{\varepsilon} = (a-\varepsilon,b+\varepsilon)$. The $\e$-\textbf{smoothing functor}, $\SS_\e:\Csh \to \Csh$, where $\varepsilon > 0$, is defined by $\SS_\e(\F)(I) = \F(I^\e)$ for each $I$ and morphisms are induced by inclusion.
\end{definition}

In essence, the $\e$-smoothing functor expands each interval $I$ by $\e$ in both directions before assigning data. This implies that the infimum width of an interval seen by the functor is now $2\e$ rather than a single point. In several cases, the increase in the intervals causes the data associated to these intervals to be fundamentally changed, sometimes removing features entirely. Fig.~\ref{fig:smoothingExamples} shows two examples of smoothing for various simple features. We will discuss the effects of smoothing on larger features composed of multiple smaller features in \cref{example:compound} from Sec.~\ref{sec:examples}.

Note that while the smoothing operation is done on the cosheaves directly, we can represent this using the Reeb graph since the fundamental structure of the cosheaf is captured completely in the Reeb graph, as seen in Fig.~\ref{fig:reebToCosheaf}.

\begin{figure}
    \centering
    \includegraphics[width=\textwidth]{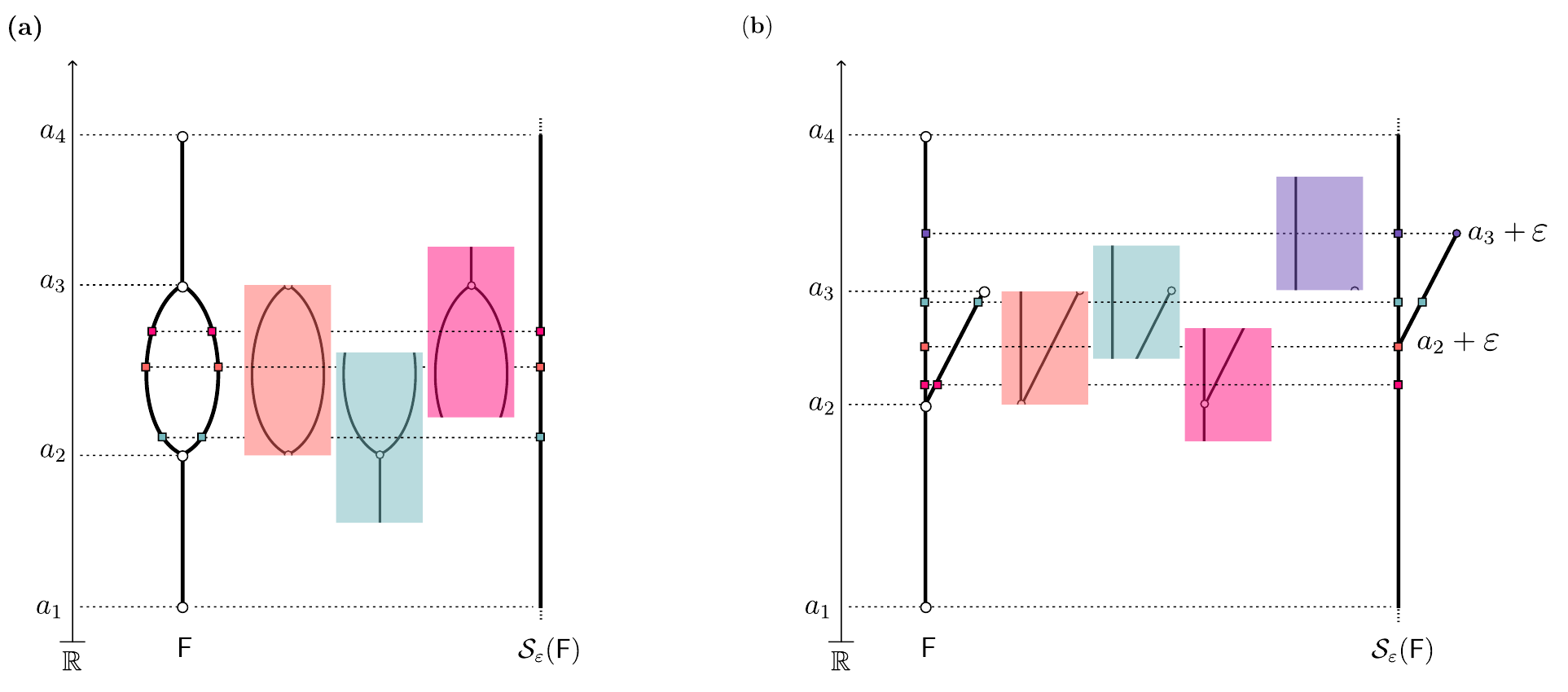}
    \caption{\textbf{(a)} The cosheaf representation $\mathsf{F}$ and its $\e$-smoothed version $\mathcal{S}_{\e}(\mathsf{F})$ of the scalar field $(\sX,\sf)$ where $\sX$ is a torus and $\sf$ is a height function. To cover the hole completely, $\e$ has to be large enough so that every interval $I$, the expanded interval $I^{\e}$ will only have one path connected component. Setting $\e \geq \frac{a_3-a_2}{2}$ will guarantee this. \textbf{(b)} A Reeb graph with a single up-leaf and its $\e$-smoothed version. As the center of the interval passes $\frac{a_3+a_2}{2}$, the number of components changes from one to two, essentially creating a leaf in the smoothed version that is shifted upwards. Note that the last component (purple) maps to only \textit{one} component in the smoothed Reeb graph. Similarly, a down-leaf will be shifted downwards.}
    \label{fig:smoothingExamples}
\end{figure}

\begin{figure}
    \centering
    \includegraphics[width=0.9\textwidth]{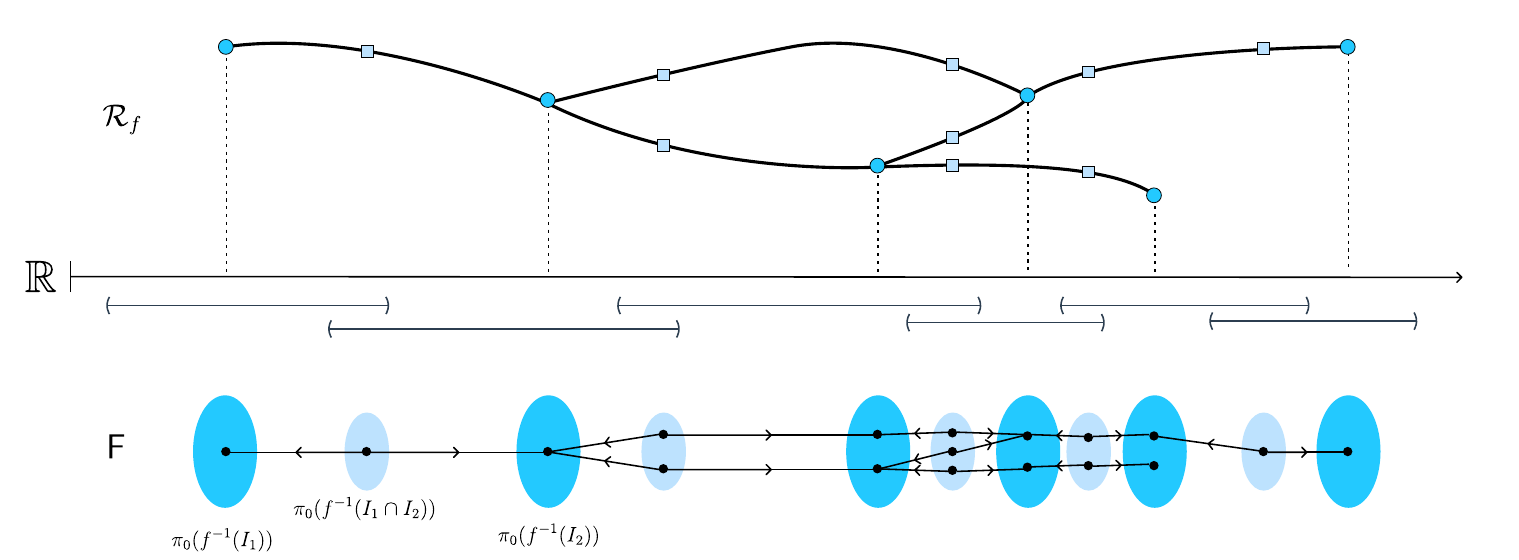}
    \caption{A Reeb graph (top) with its corresponding cosheaf (bottom). 
    The larger, darker blue sets were chosen specifically to surround the critical points of the Reeb graph. 
    The smaller, light blue sets are the pairwise intersection of the sets surrounding it. The arrows represent the morphisms from each small set to the larger sets they are included in. While a cosheaf is defined for \emph{all} intervals on the real line, the intervals above are enough to capture the topology of the Reeb graph fully. This figure was adapted from \cite{deSilva2016}. }

    \label{fig:reebToCosheaf}
\end{figure}

\begin{definition}
Let $\F:\mathbf{Int} \to \mathbf{Set}$ be a cosheaf. The natural transformation $\sigma^{\e}_{\F}:\F\Rightarrow \SS_{\e}(\F)$ is the morphism constructed from the inclusion $I\subset I^{\e}$. That is, $\sigma^{\e}_{\F} : = \F[I\subset I^{\e}]$. We say that two cosheaves $\F,\G$ are $\varepsilon$-\textbf{interleaved} if there exists a pair of natural transformations $\varphi: \F \Rightarrow \SS_\e(\G)$ and $\psi: \G\Rightarrow \SS_\e(\F)$ such that the following diagrams commute: 
\end{definition}

\begin{center}
\begin{tikzcd}
\mathsf{F} \arrow[dd, "\sigma^{2\varepsilon}_{\mathsf{F}}"', Rightarrow] \arrow[rd, "\varphi", Rightarrow] &                                                                                                   &  &                                                                                                       & \mathsf{G} \arrow[ld, "\psi"', Rightarrow] \arrow[dd, "\sigma^{2\varepsilon}_{\mathsf{G}}", Rightarrow] \\
                                                                                                           & \mathcal{S}_{\varepsilon}(\mathsf{G}) \arrow[ld, "{\mathcal{S}_{\varepsilon}[\psi]}", Rightarrow] &  & \mathcal{S}_{\varepsilon}(\mathsf{F}) \arrow[rd, "{\mathcal{S}_{\varepsilon}[\varphi]}"', Rightarrow] &                                                                                                         \\
\mathcal{S}_{2\varepsilon}(\mathsf{F})                                                                     &                                                                                                   &  &                                                                                                       & \mathcal{S}_{2\varepsilon}(\mathsf{G})                                                                 
\end{tikzcd}
\end{center}

If $\varepsilon = 0$, then this is exactly the definition of an isomorphism between $\F$ and $\G$. When two cosheaves are $\varepsilon$-interleaved, we say that there exists an $\varepsilon$-\textbf{interleaving} between them.

\begin{definition}
The interleaving distance between two constructible Reeb graphs $\RR_f,\RR_g$ is the infimum $\e$ such that their respective cosheaves are $\e$-interleaved. Formally,
\[d_I(\RR_f,\RR_g) = \inf_{\e \in \R^+}\{\e\mid\text{ there exists an $\e$-interleaving between $\F$ and $\G$} \},\]
where $\F,\G$ are their respective cosheaves.
\end{definition}

While the definitions above are rooted in category theory, it has been shown that there is a geometric way of understanding the interleaving distance \cite{deSilva2016}. We think of our smoothing operation instead as first ``thickening'' the Reeb graph by $\e$ which creates a new scalar field, and then finding the Reeb graph of the newly constructed space, resulting in a more ``coarse'' Reeb graph. 

The propositions and theorems that lead to the following construction are originally presented by first working with more general scalar fields and then proving that the results carry over to the constructible case. Here, we will assume that the scalar fields and Reeb graphs that we are discussing are constructible and leave out the more generalized work. In turn, we will drop the term `constructible' in most cases. We refer the reader to the original work \cite{deSilva2016} for a more in-depth treatment of this topic.

To start, we introduce the \textbf{category of scalar fields}, which we denote as $\Rtop$. Objects in this category are scalar fields and the morphisms between these objects are continuous, function preserving maps (see \cref{def:functionPreservingReebGraphIso}).

\begin{definition}
For $\e \geq 0$, the \textbf{thickening functor} $\TT_{\e}:\Rtop \to \Rtop$ is defined as follows:
\begin{itemize}
    \item Let $(\X,f)\in\Rtop$. Then $\TT_{\e}(\sX,\sf) = (\X^{\e},f^{\e})$ where $\X^{\e} = \X \times [-\e,\e]$ and $f_{\e}(x,t) = f(x) + t$.
    \item Let $\alpha: (\X,f) \to (\Y,g)$ be a morphism in $\Rtop$. Then $\TT_{\e}[\alpha]: (\X^{\e},f^{\e}) \to (\Y^{\e},f^{\e}):(x,t) \mapsto (\alpha(x),t)$.
\end{itemize}
\end{definition}

Since Reeb graphs are themseleves scalar fields, we have a well-defined subcategory of $\Rtop$ known as the \textbf{category of Reeb graphs}, which is denoted as $\Reeb$. Just as in $\Rtop$, the morphisms in $\Reeb$ are continuous, function preserving maps.

\begin{definition}
We define the \textbf{Reeb functor} $\RR:\Rtop \to \Reeb$ as the functor which maps a scalar field $(\X,f)$ to its Reeb graph $\RR_f = (\X_f,\tilde{f})$. That is, 
\[\RR(\X,f) := (\X_f,\tilde{f}) = \RR_f.\]
\end{definition}

\begin{definition}
Let $\RR_f = (\X_f,\tilde{f}) \in \Reeb$ . We define the \textbf{Reeb smoothing functor} $\UU_{\e}:\Reeb \to \Reeb$ as the functor which thickens a Reeb graph $\RR_f$ into a new scalar field, and then converts it back into a Reeb graph. That is,
\[\UU_{\e}(\X,f) := \RR(\TT_{\e}(\X_f,\tilde{f})) = \RR(\X_f^{\e},\tilde{f}^{\e}).\]
\end{definition}

Since $\UU_{\varepsilon}(\RR_f)$ is indeed another Reeb graph, we can define morphisms $\alpha:\RR_f \to \UU_{\varepsilon}(\RR_g)$ and $\beta: \RR_g \to \UU_{\varepsilon}(\RR_f)$. We then define other morphisms $\iota_{\varepsilon} $ and $\alpha_{\varepsilon}$ as follows:
\begin{align}
    \iota_{\e}: \RR_f \to \UU_{\e}(\RR_f),                  & \hspace{50PX} x \mapsto [x,0], \\
    \alpha^{\e}_{2\e}: \UU_{\e}(\RR_f) \to \UU_{2\e}(\RR_g), & \hspace{50PX} [x,t] \mapsto [\alpha(x),t]
\end{align}
where $[x,t]$ refers to the equivalence class of $(x,t)$ under the quotient map $\rho^{\e}_f: \TT_{\e}(\RR_f) \to \UU_{\e}(\RR_f)$.

Fig.~\ref{fig:thickeningExample} depicts an example of a Reeb graph $\RR_f$ along with its smoothed version $\UU_{\e}(\RR_f)$.

\begin{figure}
    \centering
    \includegraphics[width=0.9\textwidth]{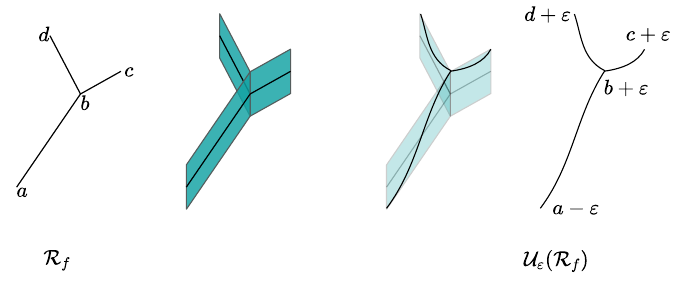}
    \caption{A depiction of the thickening process which expands the original Reeb graph $\RR_f$ into a 2-dimensional scalar field and then constructs the Reeb graph of that resulting scalar field -- creating $\mathcal{U}_{\e}(\RR_f)$. Figure adapted from \cite{deSilva2016}.}
    \label{fig:thickeningExample}
\end{figure}

\begin{definition}
We say that two Reeb graphs $\RR_f$,$\RR_g$ are $\e$-\textbf{interleaved} if there exist morphisms $\alpha_{\e}:\RRf \to \RRg$,$\beta_{\e}:\RRg \to \RRf$ such that the following diagram commutes:
\begin{center}
\begin{tikzcd}
\mathcal{R}_f \arrow[dd, "\iota_{2\varepsilon}"'] \arrow[rd, "\alpha_{\e}"] &                                                                            &                                                                              & \mathcal{R}_g \arrow[ld, "\beta_{\e}"'] \arrow[dd, "\iota_{2\e}"] \\
                                                                      & \mathcal{U}_{\e}(\mathcal{R}_g) \arrow[ld, "\beta^{\e}_{2\e}"] & \mathcal{U}_{\e}(\mathcal{R}_f) \arrow[rd, "\alpha^{\e}_{2\e}"'] &                                                                       \\
\mathcal{U}_{2\e}(\mathcal{R}_f)                              &                                                                            &                                                                              & \mathcal{U}_{2\e}(\mathcal{R}_g)                            
\end{tikzcd}
\end{center}

The \textbf{interleaving distance} between two Reeb graphs is the infimum $\e$ such that $\RR_f$ and $\RR_g$ are $\e$-interleaved. That is,
\[d_I(\RR_f,\RR_g) = \inf_{\e \in \R^+}\{\e|\text{ there exists an $\e$-interleaving between $\RR_f$ and $\RR_g$} \}.\]

\end{definition}

\begin{figure}
    \centering
    \includegraphics[width = .4\textwidth]{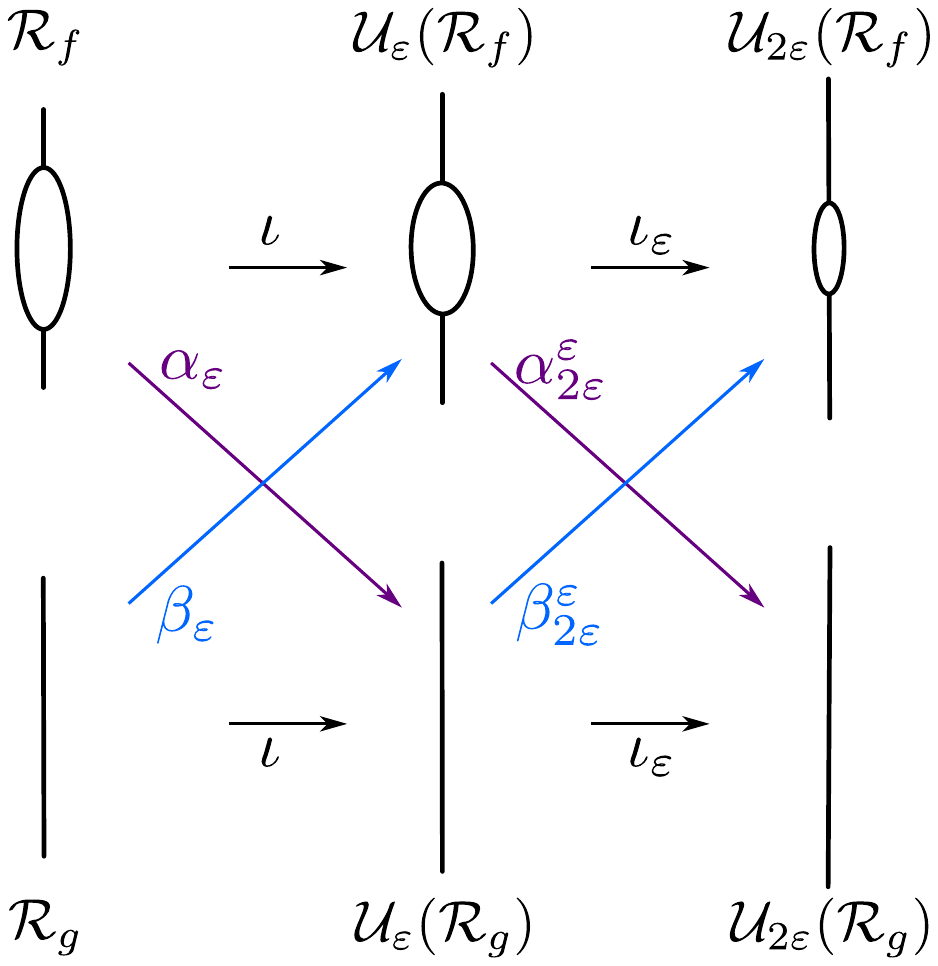}
    \qquad
    \includegraphics[width = .4\textwidth]{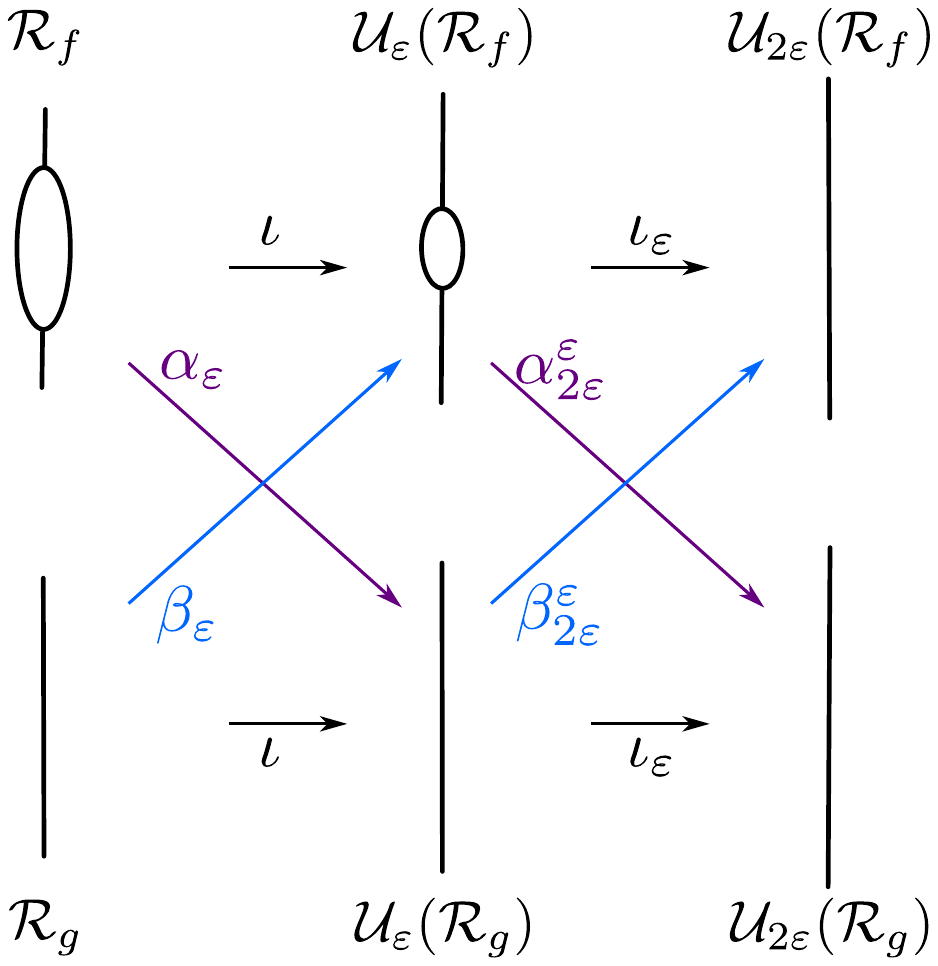}
    \caption{An example of two Reeb graphs and computation of the interleaving distance. In this example, the first $\e$ chosen is not large enough to allow for an interleaving since points on either side of the loop of $\mathcal{R}_f$ will stay on opposite sides of the loop following the top maps, but end up on the same side going down and up. However, in the second case, we have chosen a larger $\e$ value, making only one possible point to be sent to in $\mathcal{U}_{2\e}(\RR_f)$, implying the diagrams commute. }
    \label{fig:InterleavingExample}
\end{figure}

Here we have provided two different, yet equivalent, ways of computing the interleaving distance: (1) consider an equivalent structure of a Reeb graph called a cosheaf and construct an interleaving between two cosheaves $\F$ and $\G$ by increasing the minimum interval size until there exists an interleaving based on the definition of isomorphisms between cosheaves, or (2) thicken the Reeb graphs themselves by adding an extra dimension to the space and alter the original function to ultimately construct a new, 2-dimensional scalar field and then take the Reeb graph of this newly constructed space in an attempt to create two Reeb graphs which are $\e$-interleaved based on the definition of function preserving maps between Reeb graphs. In Sec.~\ref{sec:examples}, we will mostly rely on the cosheaf definition of interleaving distance for computation since, in practice, it seems to lead to more concise proofs.
See Fig.~\ref{fig:InterleavingExample} for two choices of $\e$ values, one of which is not large enough to allow for an interleaving, while the second, increased value does.

\subsection{Truncated Interleaving Distance}

Truncated smoothing is a more recently developed variation of the smoothing functor, which intuitively smooths and then ``chops off" tails in the smoothed graph~\cite{Chambers2021}. 

We define a \textbf{path} from $x$ to $x'$ in a constructible Reeb graph $\RR_f:=(\X_{f},\tilde{f})$ to be a continuous map $\pi:[0,1] \to \X_f$ such that $\pi(0) = x$ and $\pi(1) = x'$. A path is called an $\textbf{up-path}$ if it is monotone-increasing with respect to the function $\tilde{f}$, i.e. $\tilde{f}(\pi(t)) \leq \tilde{f}(\pi(t'))$ whenever $t \leq t'$. Similarly, a \textbf{down-path} is a path which is monotone-decreasing. 

Let $U_\tau $ be the set of points of $\RR_f$ that do not have a height $\tau$ up-path and let $D_\tau$ be the set of points of $\RR_f$ that do not have a height $\tau$ down-path. We defined the truncation of Reeb graph $\RR_f$ by 
$$T^{\tau}(\RR_f)= \RR_f \backslash (U_\tau \cup D_\tau)$$
so that we keep only the subgraph of $\RR_f$ that consists of the points that have both up-path and down-path of height $\tau$. 
We then define the truncated smoothing of Reeb graph  $\RR_f$ as
$S^{\tau}_{\e}(\RR_f) = T^\tau S_\e(\RR_f)$.
See Fig.~\ref{fig:graphandSmoothing} for an illustration of this operation for several values of $\e$ and $\tau$.   

\begin{figure}
    \centering
  \includegraphics[width=\textwidth, page=2]{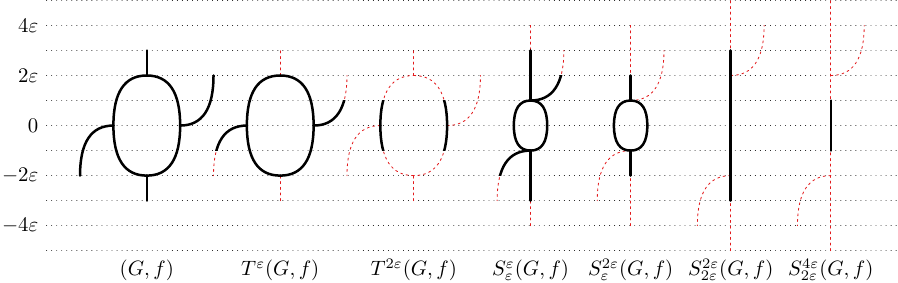}    
  \caption{A Reeb graph shown alongside several smoothed and truncated versions; figure from~\cite{Chambers2021}. }
    \label{fig:graphandSmoothing}
\end{figure}

The truncated smoothing operation inherits many of the useful properties of regular smoothing. 
First, for $0<\tau \le 2\e$, $S_\e^\tau$ is a functor.  In addition, we have a map $\eta: S_\e^\tau(\RR_f) \to S_{\e'}^{\tau'}$ for any $0 \leq \tau'-\tau \leq \e'-\e$. 
Note that we abuse notation and write $\eta$ since this map is a restriction of the map $\eta:S_\e(\RR_f) \to S_{\e'}(\RR_f)$.\footnote{This map is denoted $\rho$ in \cite{Chambers2021}. }

In essence, the $0 \leq \tau'-\tau \leq \e'-\e$ restriction can be viewed as requiring a slope of at most 1 in the $\e$-$\tau$ parameter space. 
So, fixing a parameter $m \in [0,1]$, we define the truncated interleaving distance as follows.

\begin{definition}
For $\e\geq 0$, $m \in [0,1]$ and a pair of Reeb graphs $\RR_f$ and $\RR_g$, an $\e$-interleaving is defined as a pair of function preserving maps $\phi:\RR_f \to S_\e^{m\e}(\RR_g)$ and $\psi:\RR_g \to S_\e^{m\e}(\RR_f)$ such that 
\begin{equation*}
    \begin{tikzcd}[row sep = huge, column sep = huge]
    \RR_f \ar[r, "\eta"] \ar[dr, "\phi"', very near start] 
        & S_\e^{m\e}(\RR_f) \ar[r, "\eta"] \ar[dr, "{S_\e^{m\e}[\phi]}"', very near start] 
        & S_{2\e}^{2m\e}(\RR_f) \\
    \RR_g \ar[r, "\eta"] \ar[ur, "\psi", very near start] 
        & S_\e^{m\e}(\RR_g) \ar[r, "\eta"] \ar[ur, "{S_\e^{m\e}[\psi]}", very near start]  
        & S_{2\e}^{2m\e}(\RR_g) 
    \end{tikzcd}
\end{equation*}
If such a pair exists, we say that $\RR_f$ and $\RR_g$ are $\e$-interleaved (using $S_\e^{m\e}$ when not clear from context) and define the interleaving distance as 
\begin{equation*}
    d_I^m(\RR_f,\RR_g) = \inf_{\e\geq 0} \{ \RR_f \text{ and }\RR_g \text{ are }\e\text{-interleaved} \}.
\end{equation*}
\end{definition}
Note that if $m=0$, we recover the original interleaving distance. 
As proved in \cite{Chambers2021},  $d_I^m$ is an extended pseudometric for every $m \in [0,1]$. 
Further, this family excluding $m=1$ are all strongly equivalent. This was proved in \cite{Chambers2021}; here we strengthen the results to include bounds for the general case in Cor.~\ref{cor:strongEquivTruncated}.
\begin{theorem}\label{thm:truncatedEquivalence}
[{\cite[Theorem~2.27]{Chambers2021}}]
For any pair
    $m,m' \in [0,1)$ with $0 \leq m'-m < 1-m'$
    the metrics $d_I^m$ and $d_I^{m'}$ are strongly equivalent.
    Specifically, given Reeb graphs $\RR_f$ and $\RR_g$,
    \begin{equation*}
    d_I^m(\RR_f, \RR_g)\leq d_I^{m'}(\RR_f, \RR_g)\leq \frac{1-m}{1-m'}d_I^m(\RR_f, \RR_g).
    \end{equation*}
\end{theorem}

\begin{corollary}\label{cor:strongEquivTruncated}
For all pairs $0 \leq M \leq M' < 1$, $d_I^M$ and $d_I^{M'}$ are strongly equivalent. Specifically,
\[d_I^M(\RR_f, \RR_g)\leq d_I^{M'}(\RR_f, \RR_g)\leq \frac{1-M}{1-M'}d_I^M(\RR_f, \RR_g).\]
\end{corollary}

\begin{proof}
In the following, we let $d_I^{m_i}$ denote $d_I^{m_i}(\RR_f,\RR_g)$ to avoid cumbersome notation. Suppose we have a sequence of values $\{m_0,\ldots,m_n\}$ such that $0 < m_{i+1} - m_i < 1 - m_{i+1}$. We can show by induction that $d_I^{m_0}\leq d_I^{{m_n}}\leq \frac{1-{m_0}}{1-{m_n}}d_I^{m_0}$. Given that $d_I^{m_0}\leq d_I^{{m_1}}\leq \frac{1-{m_0}}{1-{m_1}}d_I^{m_0}$ and  $d_I^{m_1}\leq d_I^{{m_2}}\leq \frac{1-{m_1}}{1-{m_2}}d_I^{m_1}$, we obtain
\[d_I^{m_0} \leq d_I^{m_2} \leq \frac{1-m_1}{1-m_2}d_I^{m_1} \leq \frac{1-m_1}{1-m_2}\cdot \frac{1-m_0}{1-m_1}d_I^{m_0} = \frac{1-m_0}{1-m_2}d_I^{m_0}.\]

Now, suppose $d_I^{m_0} \leq d_I^{m_j} \leq \frac{1-m_0}{1-m_j}d_I^{m_0}$. Then, similar to our base case, we get 
\[d_I^{m_0} \leq d_I^{m_j} \leq \frac{1-m_j}{1-m_{j+1}}\cdot \frac{1-m_0}{1-m_j}d_I^{m_0} = \frac{1-m_0}{1-m_{j+1}}d_I^{m_0}.\]

What is left to show is that given two numbers $0 \leq M \leq M' < 1$, there exists a sequence $\{M = m_0,m_1,\ldots,m_{n-1},m_n = M'\}$ such that $0 < m_{i+1} - m_i < 1 - m_{i+1}$ for all $i \in \{0,\ldots,n-1\}$. Here we construct a sequence of numbers $\{A_N,\ldots,A_N'\}$ such that each member of this sequence satisfies $0 < A_{i+1} - A_i < 1 - A_{i+1}$ and such that $0 < A_N - M < 1 - A_N$ and $0 < M' - A_{N'} < 1 - M'$. This creates the final sequence $\{M,A_N,\ldots,A_{N'},M'\}$.

First, note that we can rearrange these inequalities into the form $A_{i+1} \in [A_i, \frac{1+A_i}{2})$. If $M' \in [M,\frac{1+M}{2})$, then we are done. Otherwise, let $A_n := \frac{m_0}{2^n} - \frac{1}{2^n} + M'$ Then, $A_n$ defines a sequence of real numbers such that $A_{i+1} \in [A_i, \frac{1+A_i}{2})$. Since $A_n$ converges to $M'$, we have that for all $\e > 0$, there exists some $N$ such that for all $n \geq N$, $|M' - A_n| < \e$. If we choose $\e = 1-M'$, we get that $A_n < M'$ and
\begin{align*}
    \frac{1+A_n}{2} = & \frac{1+A_n}{2} + \frac{M'}{2} - \frac{M'}{2} 
                    =  \frac12 + \frac{M'}{2} - \frac12(M' - A_n) \\
                 \geq & \frac12 + \frac{M'}{2} - \frac12\e 
                    = \frac12 + \frac{M'}{2} - \frac12(1-M') 
                    =  M'.
\end{align*}
Thus, there exists some $N'$ such that $M' \in [A_{N'},\frac{1+A_{N'}}{2})$. To show that there exists some $N < N'$ such that $A_N \in [M,\frac{1+M}{2})$, we first let $\delta = M' - M$. Since $\frac{1-M}{2^n}$ converges to $0$ as $n\to \infty$, there exists some $N$ such that $\frac{1-M}{2^n} < \delta$. We want to show that there exists some $N$ such that $A_{N-1} < M \leq A_N < \frac{1+A_N}{2}$. To do this, first note that $A_0 < M$ and thus, if we know that there exists some $A_n \geq M$, then there exists some $A_{N} \geq$ such that $A_{N-1} < M$. We now show that if $n$ satisfies $\frac{1-M}{2^n} < \delta$, then $A_n \geq M$.
\begin{align*}
    A_n = & \frac{M}{2^n} - \frac{1}{2^{n}} + M' = \frac{M-1}{2^n} + M' = -\frac{1-M}{2^n} + M - M + M' \\
    = &-\frac{1-M}{2^n} + M + \delta \geq -\delta + M - \delta = M
\end{align*}
Now, suppose that $N$ is such that $A_{N-1} < M \leq A_N$. To show that $A_N < \frac{1+M}{2}$, note that $\frac{1+M}{2} - M = \frac{1-M}{2}$ and that $A_{N} - A_{N-1} = \frac{1-M}{2^N}$. Thus, $A_N = A_{N-1}+\frac{1-M}{2^N} \leq A_{N-1}+\frac{1-M}{2} < M + \frac{1-M}{2} = \frac{1+M}{2}$. Thus, $A_{N} \in [M,\frac{1+M}{2})$. This means that $\{M,A_N,\ldots,A_{N'},M'\}$ is a sequence beginning at $M$ and ending at $M'$ which satisfies all the necessary criteria.

\end{proof}

\section{Reeb Graph Edit Distance}
\label{sec:rged}

\subsection*{History}

The graph edit distance (GED) was first formalized by Sanfeliu and Fu~\cite{Sanfeliu1983} as a similarity measure for graphs; their original motivation/application was recognizing lower case handwritten English characters.  While various approximations and heuristics for the general problem are known~\cite{Gao2009}, it is nevertheless NP-Hard to compute~\cite{garey1979computers,Zhiping2009} as well as being APX-hard to approximate~\cite{Chih-Long1994}.  There are many well studied variants of graph edit distance for restricted classes of graphs, including the geometric graph distance~\cite{Cheong2009}, which compares embedded plane graphs.  We refer the interested reader to a comprehensive survey on the graph edit distance for discussion of further  variants~\cite{Bille2005}.

The Reeb graph edit distance was first introduced for 1-dimensional closed curves by Di Fabio and Landi ~\cite{DiFabio2012} and then later introduced for 2-dimensional manifolds by the same authors~\cite{DiFabio2016}. Originally, the core intuition was that we can think of Reeb graphs as having various topological \emph{events} -- such as the prescense of a maximum or minimum -- and to change one Reeb graph to another, we add, remove, or rearrange these events. For closed curves, the elementary deformations simply allowed for the birth, death, and relabeling (changing of function values) of events. When extrapolating to 2-dimensional manifolds, our set of edit operations is expanded to include three different operations responsible for changing the adjacencies of events. While the core idea of birth and death events in these two scenarios are the same, the structure of these events are fundamentally different. More specifically, introducing an event for a 1-dimensional closed curve involved adding two additional vertices to an edge to create a maxima - minima pair. For 2-dimensional manifolds, birth events were defined as adding upwards pointing or downwards pointing leaves.  These sets of elementary deformations are different from the classic graph edit distance where the edits can be categorized as insertion, deletion, and relabeling of vertices and edges, often times without restrictions on degrees of vertices or positions of edges \cite{Blumenthal2019}.

Below we will describe the edit distance discussed in \cite{DiFabio2016}, which we will call the \textbf{Reeb graph edit distance} and denote it as $d_E$. This requires a rather restricted setting; each Reeb graph $\RRf, \RRg$ are defined on homeomorphic, compact 2-manifolds without boundary. We will then separately discuss the one of the distances constructed from \cite{Bauer2020}, which we will call the \textbf{universal distance} and denote it as $\delta_E$. While this latter distance is completely categorical, it is useful to think of this distance as having a canonical set of edit operations -- either those stemming from \cite{DiFabio2016} when applicable, or the more general operations from \cite{Bauer2016}.

It is important to note that while the universal distance $\delta_E$ is inspired by its predecessors, it has not been shown that it is equivalent to the Reeb graph edit distance $d_E$ in a common setting, nor is it equivalent to the more directly related distance of \cite{Bauer2016}.

\subsection*{Definition}

While functional distortion distance considered $\RRf$ as a topological space and the interleaving distance considered $\RRf$ as a cosheaf $\F$, here we only need to consider the Reeb graph a combinatorial Reeb graph $\Gamma_f$.

Here, we will assume that each Reeb graph $\RR_f \in \mR$ with the additional criterion that the scalar fields have no boundary. We denote this space of scalar fields as $\mSB$ and its corresponding space of Reeb graphs as $\mRB$.

\begin{definition}
A \textbf{labeled (multi)graph} is a tuple $(\Gamma,\ell)$ where $\Gamma$ is a (multi)graph with vertex set $V(\Gamma)$, edge (multi)set $E(\Gamma)$, and a labeling function $\ell:V(\Gamma) \to \ell$, for some label set $\ell$.
\end{definition}

\begin{definition}
Let $\RR_f \in \mRB$. The \textbf{combinatorial Reeb graph} is the pair $(\Gamma_f,\ell_f)$ where $\Gamma_f := \rX$ and $\ell_f$ is the restriction of $\rf$ to the vertices of $\Gamma_f$.
\end{definition}

Since $\RR_f \in \mRB$, the vertices of $\Gamma_f$ correspond directly to critical points in the underlying scalar field.

In this section, we will define a set of elementary deformations which are a set of allowed operations for deforming a combinatorial Reeb graph $\Gamma_f$. Elementary deformations can be chained together to form a sequence of operations to which we will then associate a \textbf{cost}. The goal is to find a sequence of deformations $S$ which deforms a labeled multigraph $\Gamma \cong \Gamma_f$ into another labeled multigraph $\Gamma' \cong \Gamma_g$. We then define the Reeb graph edit distance as the infimum cost ranging over all sequences carrying $\Gamma_f$ to $\Gamma_g$.

\begin{definition}\label{def:multigraphIsomoprhism}
Let $(\Gamma_f,\ell_f),(\Gamma_g,\ell_g)$ be two combinatorial Reeb graphs. We say that these two Reeb graphs are \textbf{Reeb graph isomorphic} if there exists a bijection $\alpha: V(\Gamma_f) \to V(\Gamma_g)$
\begin{enumerate}[label=(\arabic*)]
    \item $e = e(v,v')$ is in $E(\Gamma_f)$ if and only if $e(\alpha(v),\alpha(v'))$ is in $E(\Gamma_g)$ and
    \item for every $v \in V(\Gamma_f)$, $f(v) = g(\alpha(v))$.
\end{enumerate}
If $(\X,f),(\Y,g) \in \mRB$, then we say they are $\textbf{merge tree isomorphic}$.
\end{definition}

\subsubsection*{Elementary Deformations}

The operations that are permitted to be performed on the graph $\Gamma$ are known as \textbf{elementary deformations}. The allowed operations are specific to Reeb graphs in $\mRB$ and are designed to add, remove, relabel, and change adjacencies of events. These deformations will be applied one after another in order to carry one Reeb graph $\RRf\in\mRB$ to another Reeb graph $\RRg\in\mRB$. Here we will provide brief definitions of each elementary deformation; see \cite{DiFabio2016} for explicit definitions of each deformation. Fig.~\ref{fig:elementaryDeformations} depicts each elementary deformation as well as their inverses.

For the remainder of these definitions, we denote the initial graph as $(\Gamma_0,l_0)$ and the transformed graph as $T(\Gamma_0,l_0) = (\Gamma_1,\ell_1)$, where $T$ is an elementary deformation.

\begin{definition}[Basic Deformations]\label{def:basicDeformations}
The set of basic deformations consist of three different types:
\begin{enumerate}
    \item an \textbf{elementary birth deformation} (B-type) adds a pair $\{u_0,u_1\}$ such that $u_0$  is a degree-3 vertex bisecting an existing edge and $u_1$ is a local maxima or minima. The edge $e(u_0,u_1)$ is known as a \textbf{up-leaf} if $f(u_1) > f(u_0)$ and a \textbf{down-leaf} otherwise.
    \item an \textbf{elementary death deformation} (D-type) removes an existing leaf from the initial graph.
    \item a \textbf{elementary relabel deformation} (R-type) assigns new function values to any number of vertices of $(\Gamma_0,l_0)$, while only changing the function value ordering of at most two non-adjacent vertices.
\end{enumerate}
\end{definition}

\begin{remark}\label{rem:singleStepRelabel}
The restriction given to the relabel operation prevents us from changing any up-leaf into a down-leaf in a single step.
\end{remark}

\begin{definition}[K-type Deformations]\label{def:kTypeDeformations}
    The $K$-type deformations are responsible for changing the adjacencies and function value labelings of saddles and leaves.
    \begin{enumerate}
    \item An  \textbf{elementary $K_1$ type deformation} moves an up-leaf (down-leaf) whose root bisects an up-leaf (down-leaf) $e_1$ to bisect an up-leaf (down-leaf) $e_2$ that shares a root with $e_1$.
    \item Let $u_1,u_2$ be two adjacent degree-3 vertices such that $u_1$ has two downwards-pointing edges incident to it while $u_2$ has two upwards-point edges incidient to it. An \textbf{elementary $K_2$ type deformation} switches the function value ordering of $u_1,u_2$ while and simultaneously changing the adjacencies of the connecting edges such that both $u_1$ and $u_2$ each have one downwards-point edge and one upwards-pointing edge.
    \item Let $u_1,u_2$ be two adjacent degree-3 vertices such that both $u_1$ and $u_2$ each are incident to one upwards-pointing edge and one downwards-pointing edge. An \textbf{elementary $K_3$ type deformation} switches the function value ordering of $u_1,u_2$ while and simultaneously changing the adjacencies of the connecting edges such that both $u_1$ now has two downwards-pointing edges incident to it and $u_2$ has two upwards-pointing edges incident to it.
    \end{enumerate}
\end{definition}

Each deformation we have described above inherently has an inverse deformation. Birth and death deformations are inverses of each other, $K_2$ and $K_3$-type deformations are inverses of each other, and both relabel and $K_1$ type deformations are inverses of themselves.

\begin{figure}
    \centering
    \includegraphics[width=0.9\textwidth]{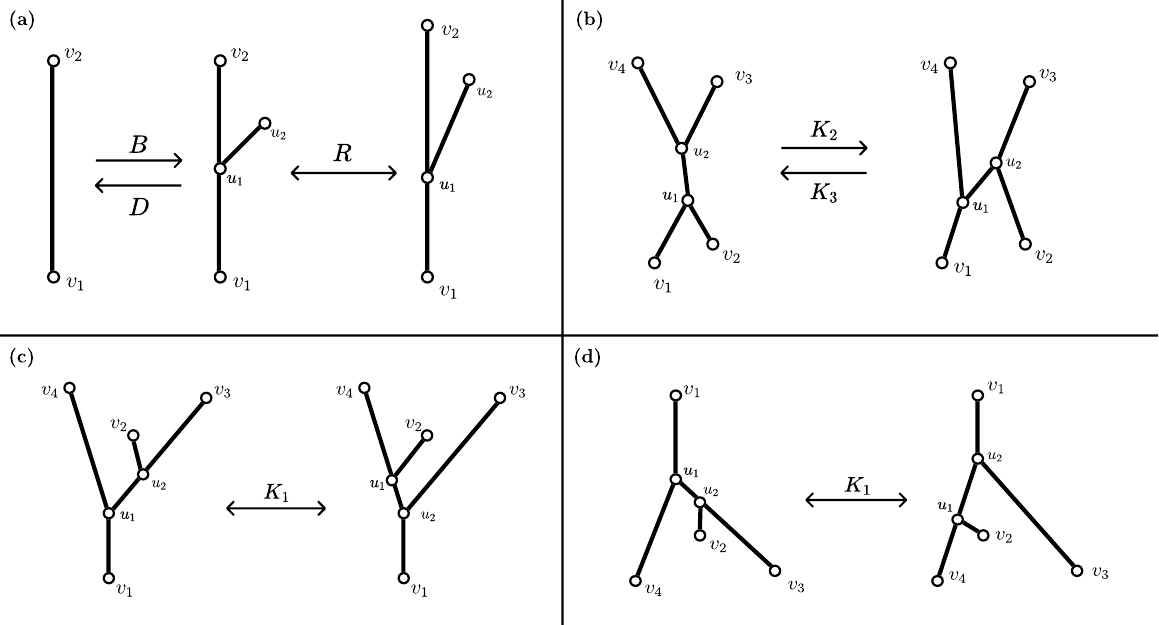}
    \caption{Simple depiction of elementary deformations. \textbf{(a)} A birth (B), death (D) and relabel (R) deformation. \textbf{(b)} A $K_2$-type deformation and its inverse; a $K_3$-type deformations. \textbf{(c)} A $K_1$-type deformation with all up-leaves. \textbf{(d)} A $K_1$-type deformation with all down-leaves.} 
    \label{fig:elementaryDeformations}
\end{figure}

\subsubsection*{Edit Sequences}

Since $T(\Gamma,\ell)$ is another labeled multigraph, we can chain together these elementary deformations to form a sequence of edit operations.

\begin{definition}
An $\textbf{edit sequence}$ of the labeled graph $(\Gamma,\ell)$ is any finite ordered sequence $S = (T_1,T_2,\ldots,T_n)$ of edit operations such that $T_1$ is an edit operation acting on $(\Gamma,\ell)$, $T_2$ is an operation acting on $T_1(\Gamma,\ell)$, and so on. We denote the result of applying this entire sequence to $(\Gamma,\ell)$ as $S(\Gamma,\ell)$. Finally, we denote $\mathcal{S}((\Gamma,\ell),(\Gamma',\ell'))$ to be the set of edit sequences $S$ that carry $(\Gamma,\ell)$ to $(\Gamma',\ell')$.
\end{definition}

\begin{definition}
Let $T \in \SS((\Gamma_f,\ell_f),(\Gamma_g,\ell_g))$ be an elementary deformation. The \textbf{cost} of $T$ is defined as follows:
\begin{itemize}[label={-}]
    \item If $T$ is an elementary birth deformation which inserts $u_1,u_2 \in V(\Gamma_g)$, then the cost of $T$ is \[c(T) = \frac{|\ell_g(u_1)-\ell_g(u_2)|}{2}.\]
    \item If $T$ is an elementary death deformation which removes vertices $u_1,u_2 \in V(\Gamma_f)$, then the cost of $T$ is 
    \[c(T) = \frac{|\ell_f(u_1)-\ell_f(u_2)|}{2}.\]
    \item If $T$ is an elementary relabel deformation, then the cost of $T$ is 
    \[c(T) = \max_{v \in V(\Gamma_f)}|\ell_f(v) - \ell_g(v)|.\]
    \item If $T$ is a $K_i$-type for $i = 1,2,3$, changing the relabeling and vertices of $u_1,u_2 \in V(\Gamma_f)$, then the cost of $T$ is 
    \[c(T) = \max\{|\ell_f(u_1)-\ell_g(u_1)|,|\ell_f(u_2)-\ell_g(u_2)|\}.\]
\end{itemize}
\end{definition}

\begin{definition}
Let $S = (T_1,\ldots,T_n) \in \mathcal{S}((\Gamma,\ell),(\Gamma',\ell'))$. Setting $(\Gamma_1,\ell_1) := (\Gamma,\ell)$ and $(\Gamma_{n+1},\ell_{n+1}) := (\Gamma',\ell')$, let $(\Gamma_{i+1},\ell_{i+1}) = T_i(\Gamma_i,l_i)$, for all $i \in \{1,\ldots,n\}$. The \textbf{cost} of $S$ is then 
\[c(S) = \sum_{i=1}^n c(T_i).\]
\end{definition}

Just as in standard graph edit distance, our goal will be to find the best possible sequence by attempting to minimize the associated cost. The Reeb graph edit distance is then defined to be the infimum cost over all possible sequences carrying one Reeb graph $\RRf$ to another $\RRg$.

\begin{definition}
The \textbf{Reeb graph edit distance}, $d_E$, between two Reeb graphs $\RRf,\RRg\in\mRB$ is defined to be
\[d_E((\Gamma_f,\ell_f),(\Gamma'_g,\ell_g)) = \inf_{\mathcal{S}((\Gamma,\ell),(\Gamma',\ell'))}c(S),\]
where $(\Gamma_f,\ell_f)$, $(\Gamma_g,\ell_g)$ are the combinatorial Reeb graphs of $\RRf,\RRg$, respectively and $(\Gamma_f,\ell_f) \cong (\Gamma,\ell)$ and $(\Gamma_g,\ell_g) \cong (\Gamma',\ell')$.
\end{definition}

It has been shown that there always exists an edit sequence carrying one Reeb graph $\RRf \in \mRB$ to another $\RRg \in \mRB$ \cite{DiFabio2016} and thus the edit distance is well defined.

By \cref{rem:singleStepRelabel}, if we were to decide to ``change'' an up-leaf into a down-leaf via relabel operations, this will always be more expensive than simply deleting the up-leaf and then inserting the down-leaf. This is a crucial notion that is present in all Reeb graph metrics and the graded bottleneck distance -- the distance will never compare features of different types to one another.

\optional{
This style of edit distance is similar to the traditional style of the (labeled) graph edit distance; each deformation in a sequence has a cost, the cost of the sequence is the sum of the individual costs, and our goal is to find the lowest cost sequence from our source graph to the target graph.
}

\begin{remark}\label{rem:spreadRelabel}
There is a rather subtle consequence of the costs defined above: although adding leaves individually each incur a cost, we often have sequences where, if we need to have multiple birth/death events in the sequence, it is best to use relabel operations to expand/shrink the leaves rather than add/remove them at their final/initial size. This is because a single relabel operation only incurs the \emph{max} cost of each relabeling. We show an example of this in Ex.~\ref{example:graphIsoEqPers}.
\end{remark}

\section{Universal distance}
\label{sec:universalDistance}

\subsection*{History}
One of the main drawbacks to the Reeb graph edit distance was its inability to compare spaces which were non-homeomorphic. Furthermore, having different sets of edit operations dependent on the dimension of the space we are working with causes a lack of cohesiveness. A new distance was introduced in \cite{Bauer2016} which bridged this gap. It simultaneously provided a set of edit operations which worked for a wide class of Reeb graphs as well as allowing for insertion and deletion of loops -- the operation needed to be able to compare non-homeomorphic spaces. However, an issue with the cost model allowed for the addition of features while incurring no cost -- essentially rendering the distance defined on these operations invalid; see Rem.~\ref{rem:bug}. This work was later expanded upon in \cite{Bauer2020} which fixed this error by creating a purely categorical framework for defining a distance between Reeb graphs. In addition, this work proved that this distance is the largest possible stable distance on Reeb graphs and is thus labeled as being a \textbf{universal distance}; see Sec.~\ref{sec:universality} for a discussion of this property.

\subsection*{Definition}

Suppose we have two PL scalar fields $(\X,f)$ and $(\X,g)$ defined on the same domain $\X$. Then, we can construct two different Reeb graphs, $\RRf,\RRg$, equipped with scalar functions $\tilde{f}$  and $\tilde{g}$. Naturally, there exist quotient maps $p_f$ and $p_g$ such that $f = \tilde{f} \circ p_f$ and $g = \tilde{g} \circ p_g$. In other words, the following diagram commutes:

\begin{equation}\label{dgm:relabel}
    \begin{tikzcd}
    \mathbb{R}                           &                                                                                    & \mathbb{R}                            \\
    \mathcal{R}_f \arrow[u, "\tilde{f}"] &                                                                                    & \mathcal{R}_g \arrow[u, "\tilde{g}"'] \\
                                         & \mathbb{X} \arrow[lu, "p_f"] \arrow[ru, "p_g"'] \arrow[ruu, "g"] \arrow[luu, "f"'] &                                      
    \end{tikzcd}
\end{equation}

Similarly, we can construct two slightly different topological spaces $\X_1$ and $\X_2$ which, when endowed with different functions $f_1$ and $f_2$, can create the same Reeb graph $\RRf$. This creates the following commutative diagram:

\begin{equation}\label{dgm:edit}
    \begin{tikzcd}
                                                       & \mathbb{R}                            &                                                    \\
                                                       & \mathcal{R}_f \arrow[u, "\tilde{f}"'] &                                                    \\
    \mathbb{X}_1 \arrow[ru, "p_1"'] \arrow[ruu, "f_1"] &                                       & \mathbb{X}_2 \arrow[lu, "p_2"] \arrow[luu, "f_2"']
    \end{tikzcd}
\end{equation}

Given two Reeb graphs $\RRf$ and $\RRg$, we can find a sequence of Reeb graphs 
$$\{\RRf = \RR_1, \RR_2, \ldots, \RR_{n-1},\RR_n = \RRg\}$$
and a sequence of topological spaces $\{X_1,\ldots,X_{n-1}\}$ such that the following diagram commutes:

\begin{equation}\label{dgm:zigzag}
\adjustbox{width=0.9\textwidth,center}{
\begin{tikzcd}
\mathbb{R}                                             &                                                      & \mathbb{R}                             &                                                 &        &                                                        & \mathbb{R}                                     &                                                                & \mathbb{R}                                             \\
\mathcal{R}_f = \mathcal{R}_1 \arrow[u, "\tilde{f}_1"] &                                                      & \mathcal{R}_2 \arrow[u, "\tilde{f}_2"] &                                                 & \ldots &                                                        & \mathcal{R}_{n-1} \arrow[u, "\tilde{f}_{n-1}"] &                                                                & \mathcal{R}_n = \mathcal{R}_g \arrow[u, "\tilde{f}_n"] \\
                                                       & X_1 \arrow[lu, "{p_{1,1}}"'] \arrow[ru, "{p_{1,2}}"] &                                        & X_2 \arrow[lu, "{p_{2,2}}"'] \arrow[ru, dotted] &        & X_{n-2} \arrow[ru, "{p_{n-2,n-1}}"] \arrow[lu, dotted] &                                                & X_{n-1} \arrow{lu}[swap]{ p_{n-1,n-1}} \arrow{ru}[swap]{p_{n-1,n}} &                                                       
\end{tikzcd}
}
\end{equation}

This zigzag diagram above can be thought of as being constructed by combining alternating copies of the Dgm.~\ref{dgm:relabel} and Dgm.~\ref{dgm:edit} -- essentially constructing a sequence of edit and relabel operations carrying $\RRf$ to $\RRg$. 

\begin{remark}
The universal distance is constructed in the context of connected, compact, triangulable spaces defined as \textbf{Reeb domains} and quotient maps from the scalar field to the Reeb graph called \textbf{Reeb quotient maps}, which have the additional restriction of being PL. For the sake of being cohesive with the rest of this document, we will continue this section without using these terms explicitly and try our best to frame this work to be more consistent with our previous definitions of Reeb graphs.
\end{remark}

Now, rather than think of $X_1,\ldots,X_{n-1}$ as being arbitrary topological spaces, here we will restrict $X_i$ such that each is a 1-dimensional CW complex. This allows us to solely focus our topological changes to the spaces $X_1,\ldots,X_{n-1}$ while focusing our function value changes to the spaces $\RR_1,\ldots,\RR_n$.

\begin{definition}
\label{def:zigzagDgm}
Let $\RRf,\RRg$ be two Reeb graphs. A \textbf{zigzag diagram} $Z$ is a sequence of Reeb graphs $\mathsf{R} = \{\RRf = \RR_1,\RR_2,\ldots,\RR_{n-1},\RR_n = \RRg\}$ coupled with a sequence of graphs $\mathsf{X} =\{X_1,\ldots,X_{n-1}\}$ such that for each $X_i$
there are two valid maps $p_{i,i}:V(X_{i}) \to V(\RR_{i}) ,p_{i,i+1}: V(X_{i})\to V(\RR_{i+1})$ which respect edge assignments. That is, if $e(x_j,x_k) \in E(X_i)$, then $e(p_{i,i}(x_j),p_{i,i}(x_k)) \in E(\RR_i)$. The sequence $\mathsf{X}$ will be called the \textbf{connecting spaces} of $Z$ while $\mathsf{R}$ is called the Reeb graphs of $Z$
\end{definition}

Similar to the Reeb graph edit distance case, we are always guaranteed a way to construct this zigzag diagram between any two PL Reeb graphs. The fact that these quotient maps respect edge assignments comes directly from surjectivity, continuity, and piecewise linearity of the quotient maps used. 

In general, the Reeb graphs in $\mathsf{R}$ need not abide by the same constraints that the original Reeb graphs follow. For example, if $\RRf,\RRg \in \mR$, then neither $\RRf$ nor $\RRg$ may have vertices of degree 4. Although $\mathsf{R}$ will commonly involve objects which have vertices of degree 4 when we need to make topological changes to the Reeb graph. See \cref{fig:zigzagExample} for an example.

This zigzag diagram can now be considered a sequence of relabel and edit operations which carry the Reeb graph $\RRf$ to $\RRg$. We define a cost of this zigzag diagram by first constructing the zigzag diagram's \textbf{limit}. For us to do this, we first introduce the notion of the \textbf{pullback} of a pair of morphisms.

\begin{definition}
The \textbf{pullback} of two morphisms $f:X_1 \to Y$, $g:X_2 \to Y$ is the set of pairs $(x_1,x_2)$, where $x_1\in X_1$, $x_2 \in X_2$ and $f(x_1) = g(x_2)$. The pullback is denoted $X_1 \times_Y X_2$. The pullback is also known as the \textbf{fiber product} of the morphisms $f:X_1 \to Y$ and $g:X_2 \to Y$.
\end{definition}

The pullback is intuitively then the cartesian product of spaces $X_1$, $X_2$ with the restriction that the elements in the tuple $(x_1,x_2)$ map to the same point in $Y$. In Dgm.~\ref{dgm:relabel}, the pullback is $X_1 \times_{\RRf} X_2$, which is the limit of the diagram.

The limit of Dgm.~\ref{dgm:zigzag} is then the iterated pullback of the zigzag diagram. That is, \[\limit = X_1 \times_{\RR_2} \times X_2 \times_{\RR_3} X_3 \times_{\RR_4} \ldots \times_{\RR_{n-2}} X_{n-2} \times_{\RR_{n-1}} X_{n-1}.\] By construction, there exists quotient maps $q_i:\limit \to X_i$ for all $i = \{1,\ldots,{n-1}\}$ which is defined as the projection onto the $i^{th}$ coordinate. That is, $q_i(x_1,\ldots,x_i,\ldots,x_{n-1}) = x_i$. Then a point $x = (x_1,x_2,\ldots,x_{n-1}) \in \limit$ if 
\[(p_{i,i+1}\circ q_i)(x_1,\ldots,x_{n-1}) = p_{i,i+1}(x_i) = p_{i+1,i+1}(x_{i+1}) = (p_{i+1,i+1}\circ q_{i+1})(x_1,\ldots,x_{n-1})\] for all $i \in \{1,\ldots,n-2\}$.

\begin{definition}
Let $Z$ be a zigzag diagram which carries $\RR_{\sf_1}$ to $\RR_{\sf_n}$ and let $\limit$ be the limit of $Z$. We define the \textbf{spread} of $\limit$ to be the function 
\[s^\limit: \limit \to \R, \hspace{10PX} x \mapsto \max_{i=1,\ldots,n}\uf_i(x) - \min_{i=1,\ldots,n}\uf_i(x),\]
where $\uf_i = \tilde{f}_i \circ p_{i-1,i} \circ q_{i-1} = \tilde{f}_i \circ p_{i,i} \circ q_i$. Moreover, we define the \textbf{cost} of $Z$ to be the supremum of the spread of $\limit$. That is,
\[c_Z = ||s^\limit||_{\infty} = \sup_{x\in \limit}\bigg(\max_{i=1,\ldots,n}\uf_i(x) - \min_{i=1,\ldots,n}\uf_i(x)\bigg)\]
\end{definition}

Finally, we can define the edit distance between Reeb graphs $\RRf$ and $\RRg$:

\begin{definition}
We define the \textbf{universal distance} $\delta_E$ between $\RRf$ and $\RRg$ to be the infimum cost over all zigzag diagrams carrying $\RRf$ to $\RRg$. That is,
\[\delta_E(\RRf,\RRg) = \inf_Z c_Z\]
\end{definition}

The distance defined above is denoted as $\delta_{eGraph}$ in \cite{Bauer2020}.

\begin{figure}
    \centering
    \includegraphics[width=\textwidth]{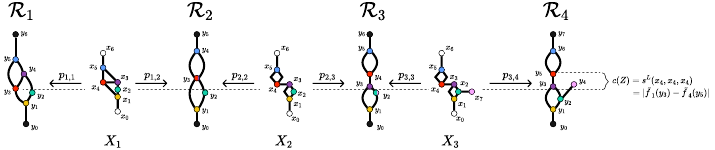}
    \caption{Zigzag diagram carrying $\RR_1$ to $\RR_4$. Color in the connecting spaces indicate the quotient mappings to the Reeb graphs. The vertical positioning of vertices of the connecting spaces do not correspond to function value unlike the Reeb graphs. Two-toned vertices indicate that the quotient map of two distinct vertices map to the same point in that Reeb graph. The vertices of the Reeb graphs are labeled in increasing function order.}
    \label{fig:zigzagExample}
\end{figure}

\begin{example}\label{example:UniversalExample}
\cref{fig:zigzagExample} depicts a zigzag diagram for carrying a Reeb graph $\RR_1$ to $\RR_4$. Note the vertex $x_2$ maps to a non-critical point in each Reeb graph except for in $\RR_4$. This is to show that even though the introduction of this new leaf is not until the final Reeb graph, the vertices associated with with this leaf still must map to some point in each of the Reeb graphs to be a well-defined zigzag diagram. 

In order for us to carry $\RR_1$ to $\RR_4$, we first separate the 1-cycles. This is done by having $X_1$ and $X_2$ constructed in such a way that they $x_3$ is no longer connected to $x_5$ and $x_4$ is no longer connected to $x_1$. The quotient maps $p_{1,2}$ and $p_{2,2}$ then map $x_3$ and $x_4$ to the same point in $\RR_2$. In order for us to insert the leaf in $\RR_4$, we simply introduce a new edge and new vertex to the final connecting space $X_3$.

Aside from the global max and global min of the Reeb graph, each color indicates its own member, i.e. $(x_1,x_1,x_1),(x_2,x_2,x_2),\ldots,(x_5,x_5,x_5)$ are all elements of the limit. In addition, since $x_3\in X_1$ and $x_4 \in X_2$ map to the same point in $\RR_2$, we have $(x_3,x_3,x_4)$ and $(x_4,x_4,x_3)$ in the limit as well. From this image, we can see that the largest spread of these elements come from $(x_4,x_4,x_4) = |\rf_1(y_3)-\rf_4(y_5)|$.
\end{example}

\begin{remark}\label{rem:bug}
The cost model defined in \cite{Bauer2016} closely resembles the cost model defined for the zigzag diagrams of \cite{Bauer2020} above. Unfortunately, there is a small error which essentially allows for adding any size leaf with a cost approaching 0. If we want to relabel a maxima node $v_1$, for example, to have function value $a$, instead we can do the following: relabel $v_1$ to have function value $a/3$, then add an edge $e(v_1,v_2)$ where $v_2$ is degree 1 and then relabel its function value to be $2a/3$. Repeat once more with a new node $v_3$ and set its function value to $a$. We then delete vertices $v_1$ and $v_2$. Since the cost model only tracks the individual vertices, the cost induced by this is $a/3$ instead of $a$. As the number of intermittent vertices increases towards $\infty$, the cost would decrease towards $0$. This issue does not change any proofs of completeness of this set of elementary deformations. By constructing the cost model in this categorical way, this error is completely fixed. 
\end{remark}

\section{Distance Properties}
\label{sec:distProp}

We will next seek to understand the relationships between each Reeb graph metric to each other and to the bottleneck distance, as well as understanding crucial results of stability which make these distances robust to input perturbations. 
In the following sections, we summarize the relationships between each of the Reeb graph metrics -- including their relationship to bottleneck distances and the $L^{\infty}$ distance. A visual summary of the results are depicted in Fig.~\ref{fig:distanceLandscape}. We will also provide explicit statements concerning other properties that the Reeb graph metrics exhibit which distinguish them from the bottleneck distance. Since these properties were originally defined for only a subset of these distances, we will provide additional statements with proof for the remaining distances to bridge this gap.

Table~\ref{table:distProps} displays the relevant citations for each property of each distance. We refer the reader to Appendix \ref{sec:appx:truncatedProperties} for explicit statements about the properties of the truncated interleaving distance.

\begin{table}
\centering
\begin{tabular}{@{}ccccccc@{}}
                         & $d_B$              & $d_I$              & $d^m_I$             & $\dfd$          & $d_E$              & $\du$       \\ \midrule
Stable                   & \cite{Steiner2009} & \cite{deSilva2016} & -                   & \cite{Bauer2014}  & \cite{DiFabio2016} & \cite{Bauer2020} \\
Discriminative           & -                  & \cite{Bauer2015b}, \ref{cor:interleaving-discrim}  & \cite{Chambers2021}, \ref{cor:truncated-interleaving-discrim} & \cite{Bauer2014}, \ref{cor:FDD-discrim}  & \cite{DiFabio2016} & \cite{Bauer2020} \\
isomorphism indiscernible    & -                  & \cite{deSilva2016} & \cite{Chambers2021} & \cite{Bauer2015b}, \ref{cor:FDD-iso} & \cite{DiFabio2016} & \cite{Bauer2020} \\
Path Component Sensitive & -                  & \cite{deSilva2016} & \cite{Chambers2021} & \cite{Bauer2015b} & -                  & \cite{Bauer2020}, \ref{prop:universal-pcs} \\
Universal                & -                  & -                  & -                   & -                 & -                  & \cite{Bauer2020} \\ \bottomrule
\end{tabular}
\caption{Table of distance properties. Entry corresponds to a citation where the distance was proved. We supply additional references to statements we contributed in this work which solidify these properties. We denote disproven properties or properties that are not applicable to the given distance with ``-".\label{table:distProps}}
\end{table}

Before delving into the properties, we list theorems which prove strong equivalence of the interleaving, functional distortion, and universal distance on Reeb graphs.

\begin{theorem}[{\cite[Theorem~16]{Bauer2015b}}]\label{thm:strongEquiv}
Let $\RRf$ and $\RRg$ be two constructible Reeb graphs. Then, the functional distortion distance and interleaving distance are strongly equivalent. Specifically,
\[d_I(\RRf,\RRg) \leq \dfd(\RRf,\RRg) \leq 3d_I(\RRf,\RRg).\]
\end{theorem}

\begin{theorem}[{\cite[Theorem~11]{Bauer2021}}]\label{thm:strongEquivUniv}
Let $\RRf$ and $\RRg$ be two constructible Reeb graphs\footnote{The Reeb graphs need not be constructible, according to \cite{Bauer2021}, but constructibility is a sufficient condition for these inequalities to hold.} Then, the interleaving distance and functional distortion distance are strongly equivalent to the universal distance. Specifically,
\begin{align*}
    d_I(\RRf,\RRg) \leq & \du(\RRf,\RRg) \leq 5d_I(\RRf,\RRg)\\ \dfd(\RRf,\RRg) \leq & \du(\RRf,\RRg) \leq 3\dfd(\RRf,\RRg).
\end{align*}
\end{theorem}

\subsection{Stability}

The notion of stability in TDA was first introduced for the persistence diagrams in \cite{Steiner2005}, where it is stated as demonstrating that small changes in the input data would only lead to small changes in the persistence diagram. 
Although stability is often described as some property of the topological signature, it is more a property of the signature \textit{and} the metric chosen.

\begin{property}[\textbf{Stability}]\label{p:stability}
The $L^\infty$ \textbf{distance} between $f$ and $g$ is defined as \[||f-g||_\infty := \max_{x \in \X} |f(x) - g(x)|.\] If $d$ is a distance on Reeb graphs, we say that $d$ is \textbf{stable} if \[d(\RRf,\RRg) \leq ||\sf-\sg||_\infty,\]
for all scalar fields $(\sX,\sf),(\sX,\sg)$ and corresponding Reeb graphs $\RRf,\RRg$. 
\end{property}

The difficulty in the stability of the Reeb graph metrics is that small perturbations may lead to topological changes in the resulting graph. From a computational standpoint, this is one of the fundamental properties that makes implementations of Reeb graph metrics complex. We discuss more experimentally driven issues of stability in \cref{chapter:MTMD}.  \josh{fix me}

\begin{theorem}[{\cite[Stability Theorem for Tame Functions]{Steiner2009}}]
Let $(\X,f)$, $(\X,g)$ be two tame scalar fields where $\X$ is a triangulable topological space. Then we have \[\db{}(\emph{ExDgm}(f),\emph{ExDgm}(g)) \leq ||f-g||_{\infty}.\]
\end{theorem}

As stated in \cite{Steiner2009}, this above result can be strengthened to apply to the individual subdiagrams. That is, 
\[\db{}(\text{ExDgm}^{\text{type}}(f),\text{ExDgm}^{\text{type}}(g) \leq ||f-g||_{\infty},\]
implying that $\dB{}(\RRf,\RRg) \leq ||f-g||_{\infty}$.

We next provide the references from the literature showing that every Reeb graph distance discussed in this paper is stable.

\begin{theorem}[{\cite[Theorem~4.1]{Bauer2014}}]
Let $(\X,f)$, $(\X,g)$ be two tame scalar fields and whose quotient maps $\rho_f:\X \to \X_f$ and $\rho_g:\X \to \X_g$ have continuous sections (i.e. $\rho_f \circ s_f = \mathbf{id}_{\RRf}$). Then, \[\dfd(\RRf,\RRg) \leq ||f-g||_{\infty}.\]
\end{theorem}

\begin{theorem}[{\cite[Theorem~4.4]{deSilva2016}}]
Let $\RRf$ and $\RRg$ be two constructible Reeb graphs defined on the same domain $\X$. Then
\[d_I(\RRf,\RRg) \leq ||f-g||_{\infty}.\]
\end{theorem}

\begin{proposition}[{\cite[Theorem~28]{DiFabio2016}}]
Let $\RRf, \RRg \in \MM$. Then \[d_E(\RRf,\RRg) \leq ||f-g||_{\infty}.\]
\end{proposition}

\begin{theorem}[{\cite[Theorem~5.6]{Bauer2020}}]
Let $\RRf$ and $\RRg$ be PL Reeb graphs and let $\X$ be a connected, triangulable space $\X$. Let $\rho_f:\X \to \RRf$ and $\rho_g:\X \to \RRg$ be quotient maps such that $\tilde{f} \circ \rho_f = f, \tilde{g} \circ \rho_g = g$. Then \[\du(\RRf,\RRg) \leq ||f-g||_{\infty}.\]
\end{theorem}

\truncatedOptionalImportant{
The strong equivalence of the interleaving distance and truncated interleaving does not imply that the truncated interleaving distance is stable in the traditional sense. 
\begin{proposition}
Let $\RRf$ and $\RRg$ be two constructible Reeb graphs defined on the same space $\X$. Then for a fixed $m \in [0,1)$, we have
\[d^m_I(\RRf,\RRg) \leq \frac{1}{1-m}||f-g||_{\infty}.\]
\end{proposition}
}

\subsection{isomorphism indiscernibility}

Each Reeb graph metric attains a value of $0$ if and only if the Reeb graphs are Reeb graph isomorphic to one another. This is a core property which distinguishes these from the bottleneck distance.
Stating that a Reeb graph metric is isomorphism indiscernible is equivalent to saying that the distance is an extended metric on the isomorphism classes of Reeb graphs, rather than an extended pseudometric on Reeb graphs.

\begin{property}[\textbf{Isomorphism Indiscernibility}]\label{p:isomorphismInvariance}
Let $\RRf$ and $\RRg$ be two Reeb graphs. We say that $d$ is \textbf{isomorphism indiscernible} if $d(\RRf,\RRg) = 0$ if and only if $\RRf$ and $\RRg$ are Reeb graph isomorphic. 
\end{property}

\begin{theorem}
The bottleneck distance is not isomorphism indiscernible.
\end{theorem}

\begin{proof}
See Ex.~\ref{example:graphIsoEqPers} for a counterexample.
\end{proof}

\begin{proposition}[{\cite[Proposition~4.6]{deSilva2016}}]
Let $\RRf$ and $\RRg$ be constructible Reeb graphs. Then, the interleaving distance is isomorphism indiscernible.
\end{proposition}

\begin{corollary}\label{cor:FDD-iso}
Let $\RRf$ and $\RRg$ be constructible Reeb graphs. Then, the functional distortion distance is isomorphism indiscernible.
\end{corollary}

\begin{proof}
This follows directly from strong equivalence of the interleaving distance and functional distortion distance; see Thm.~\ref{thm:strongEquiv}.
\end{proof}

\begin{corollary}[{\cite[Corollary~39]{DiFabio2016}}]
Let $\RRf,\RRg \in \MM$. Then the Reeb graph edit distance is isomorphism indiscernible.
\end{corollary}

\begin{corollary}[{\cite[Corollary~3]{Bauer2020}}]
Let $\RRf$ and $\RRg$ be PL scalar fields. Then, the universal distance is isomorphism indiscernible.

\end{corollary}

\begin{corollary}
The bottleneck distance is an extended pseudometric on the space of isomorphism classes of Reeb graphs, while the interleaving, functional distortion, Reeb graph edit, and universal distance are extended metrics on the space of isomorphism classes of Reeb graphs.
\end{corollary}

\subsection{Discriminativity}

\textbf{Discriminativity} is a property which states that a distance is bounded below by a baseline distance up to a constant. Essentially, this states that the baseline distance is not able to discern between two objects as well as the non-baseline distance in certain cases. Our Reeb graph metrics are constructed in such a way that they will be both stable and more discriminative than the bottleneck distance. 

\begin{property}[\textbf{Discriminativity}]\label{p:discriminativity}
Let $d_0$ and $d$ be two distances on a class of topological descriptors $\mathcal{D}$. We say that $d$ is more \textbf{discriminative} than $d_0$ (the baseline) if for any two descriptors $D_1,D_2\in\mathcal{D}$, there exists a constant $c > 0$ such that \[d_0(D_1,D_2) \leq c\cdot d(D_1,D_2),\]
and if there does \emph{not} exist a constant $c'$ such that $d_0(D_1,D_2) = c'\cdot d(D_1,D_2)$. We call the constant $c$ the \textbf{discriminativity constant}.
\end{property}

As stated in Sec.~\ref{sec:bottleneckDist}, the difference between the graded and ungraded bottleneck distance has caused some discrepancies when stating the discriminative power of these Reeb graph metrics. Bauer et al. \cite{Bauer2014} discussed the relationship between the functional distortion distance and the ungraded bottleneck distance defined on the individual diagrams of $\text{Ord}_0$, $\text{Rel}_1$, and $\text{Ext}_1$. Subsequent bounds for the interleaving distance were constructed for the same diagrams in \cite{Bauer2015b} when the interleaving distance and functional distortion distance were shown to be strongly equivalent. 

Note that the Reeb graph metrics are not scaled versions of the bottleneck distance since they are isomorphism indiscernible while the bottleneck distance is not. Thus, to state that these distances are more discriminative than the bottleneck distance only requires that $d_B(\RRf,\RRg) \leq c\cdot d(\RRf,\RRg)$ for some $c$. 

As stated in \cite{Bauer2015b} and to the best of our knowledge, there has currently been no example in which the functional distortion distance is strictly greater than the interleaving distance. Thus, we cannot immediately conclude that the functional distortion distance is more discriminative than the interleaving distance.

For statements involving the functional distortion distance and interleaving distance, we will assume that $\RRf$ and $\RRg$ are constructible Reeb graphs.

A major issue in the literature is that overloading of terminology (``the bottleneck distance''), particularly when it comes to category theoretic translations of these ideas across inputs. This has led to some confusion as to what exactly is being computed. 
Indeed, the investigation of algebraic generalizations of persistence in particular has led to generalizations of the idea of bottleneck distance which can be applied to wildly different algebraic constructions. 

Understanding the relationship between the bottleneck distance between extended persistence diagrams and the interleaving distance betwween Reeb graphs requires us to study the interleaving distance for interlevelset persistence  (see Appx.~\ref{sec:appx:interlevelset} for further details and examples). 
In this case, data about $H_p(f\inv(a,b))$ is encoded for each interval $(a,b)$; this is closely related to the Reeb graph functorial representation $\pi_0(a,b)$, but results in a vector space rather than a set.
Again, hiding relevant details in Appx.~\ref{sec:appx:interlevelset}, these objects can be broken up into well structured ``blocks'' defined over collections of intervals. 
Representing an interval $[a,b]$ as a point in the plane $(a,b)$, these blocks exactly correspond to the four possible types of connected intervals in $\R$: open, half open (on either side) and closed. 
The result of the discussion in that paper is that one can forget the full block structure, thinking of the representative barcode as simply as an interval in $\R$, and then find a bottleneck matching between the bars. 

This culminates in \cite[Thm.~3.8]{Bjerkevik2021}, which states that the bottleneck distance between the intervals as viewed in $\R$ is upper-bounded by twice the Reeb graph interleaving distance. 
With some careful checking of specifics, it can be seen that the relevant objects as identified in fact correspond to the ungraded bottleneck distance. 
While that does not make the result untrue, it does mean that it is not as tight as it could be. 
Thus, we show the following strengthening of \cite[Thm.~3.8]{Bjerkevik2021} in Appx.~\ref{sec:appx:interlevelset}.

\begin{theorem}[{\cite[Theorem~3.8]{Bjerkevik2021}}, Strengthened]
\label{thm:strengthenedBottleneck_and_interleaving}
Let $\RRf, \RRg \in \MM$. 
Then the interleaving distance is more discriminative than the \textit{graded} bottleneck distance with a discriminativity constant of $2$. 
That is, 
\[d_B(\RRf,\RRg) \leq 2d_I(\RRf,\RRg).\]
\end{theorem}

Finally, we note that the combination of Thm.~\ref{thm:strengthenedBottleneck_and_interleaving} with strong equivalence (Thm.~\ref{thm:strongEquiv}) results in a bound on the functional distortion distance.

\begin{corollary}
Let $\RRf, \RRg \in \MM$. Then the functional distortion distance is more discriminative than the ungraded bottleneck distance with a discriminativity constant of $2$. That is, \[d_B(\RRf,\RRg) \leq 2\dfd(\RRf,\RRg).\]
\end{corollary}

Below, we discuss the relevant theorems from the literature which indicate that for some individual subdiagrams, the discriminativity constant for the bottleneck distance compared to the functional distortion and interleaving distance can be reduced. We add an additional final bound comparing $\dbc{0}{Ext}$ to the interleaving and functional distortion distance.

\begin{theorem}[{\cite[Theorem~4.2]{Bauer2014}}]\label{thm:FDD-discrim-0-orig}
The functional distortion distance is more discriminative than the ungraded bottleneck distance defined on the $0$-dimensional ordinary persistence diagrams of $f,g$ and their reflections $-f,-g$. That is
\[\dbc{0}{}(\Dgm(f),\Dgm(g)) \leq \dfd(\RRf,\RRg)\] and \[\dbc{1}{}(\Dgm(-f),\Dgm(-g)) \leq \dfd(\RRf,\RRg)\]
\end{theorem}

\begin{corollary}\label{cor:FDD-discrim-0}
The functional distortion distance is more discriminative than the ungraded bottleneck distance defined on diagrams $\Ord_0$ and $\Rel_1$. That is
\[\db{0}{\Ord}(\RRf,\RRg) \leq \dfd(\RRf,\RRg)\] and \[\db{1}{\Rel}(\RRf,\RRg) \leq \dfd(\RRf,\RRg)\]
\end{corollary}

\begin{proof}
The diagrams $\Dgm(f)$ and $\Dgm(-f)$ are essentially equivalent diagrams to $\Ord_0(f)$ and $\Rel_1(f)$, except for that $\Dgm(f)$ and $\Dgm(-f)$ have points at $+\infty$ which represents the global minimum \cite{Steiner2009}. The bottleneck distance defined on these diagrams then measures the distance between these points at $+\infty$ by just measuring the difference in their birth times (the differences in the global minimums). The addition of these points only increases the bottleneck distance -- making $\db{0}{\Ord}$ and $\db{1}{\Rel}$ a lower bound to these values.
\end{proof}

The original construction for discriminativity of the functional distortion distance in \cite{Bauer2014} looked at subdiagrams of the full extended persistence diagram in a way that only captured the difference in the global minima, but not the global maxima. The below proposition makes certain that the difference in the global minimum - global maximum pair under the bottleneck distances is bounded above by the functional distortion distance.

\begin{proposition}\label{prop:FDD-discrim-0-ext}
The functional distortion distance is more discriminative than the ungraded bottleneck distance defined on  $\Ext_0$. That is
\[\dbc{0}{Ext}(\RRf,\RRg) \leq \dfd(\RRf,\RRg).\]
\end{proposition}

 We now introduce the following Lemma to aid in the proof of Prop.~\ref{prop:FDD-discrim-0-ext}.

\begin{lemma}\label{lem:FDD-global}
Let $\RRf$ and $\RRg$ be two connected, constructible Reeb graphs. Let $(v_0,v_1)$ and $(u_0,u_1)$ be the global max, global min pair of $\RRf$ and $\RRg$, respectively. Then, 
\[\dfd(\RRf,\RRg) \geq \max\{|\rf(v_1)-\rg(u_1)|,|\rf(v_0)-\rg(u_0)\}.\]
\end{lemma}

\begin{proof}
Let $\rf(v_1) = a_1, \rf(v_0) = a_0, \rg(u_1) = b_1,$ and $\rg(u_0) = b_0$. If $\Phi(v_1) = u_1$ and $\Phi(v_0) = u_0$, then $\dfd(\RRf,\RRg) \geq ||\rf-\rg\circ\Phi||_{\infty} \geq \max\{|a_1-b_1|,|a_0-b_0|\}$. If $\Phi(v_1) = x_1$ and $\Phi(v_0) = x_0$, for some $x_0 \neq u_0, x_1 \neq u_1$, then $\dfd(\RRf,\RRg) \geq ||\rf-\rg\circ\Phi||_{\infty} \geq \max\{|a_1-\rg(x_1)|,|a_0-\rg(x_0)|\} \geq \max\{|a_1-b_1|,|a_0-b_0|\}$, since $x_0$,$x_1$ are not the the global min and global max of $\RRg$.

Thus, $\dfd(\RRf,\RRg) \geq \max\{|\rf(v_1)-\rg(u_1)|,|\rf(v_0)-\rg(u_0)\}. $
\end{proof}

\begin{proof}[Proof of Prop.~\ref{prop:FDD-discrim-0-ext}]
We will first assume that $\RRf$ and $\RRg$ have only one connected component. Then in each persistence diagram defined on the Reeb graphs there is only one pair in the class $\text{Ext}_0(\rf)$ -- the persistence pair representing the global minimum and global maximum.

Let $\RRf,\RRg$ be two constructible Reeb graphs. Let $(a_0,a_1) \in \text{Ext}_0(\rf)$ and $(b_0,b_1) \in \text{Ext}_0(\rg)$, and let the vertices of the Reeb graphs which achieve these values to be $v_0,v_1 \in \RRf$ and $u_0,u_1 \in \RRg$. Furthermore, without loss of generality, assume that there is only one vertex in each graph which achieve the global max and global min values. There are two choices for a bijection $\eta$ between these diagrams: (1) let $\eta((a_0,a_1)) = (b_0,b_1)$ or (2) let $\eta$ map both points to the diagonal. Thus, \[\db{0}{\Ext}(\RRf,\RRg) = \min\{\max\{|a_1-b_1|,|a_0-b_0|\},\max\{\frac12|a_1-a_0|,\frac12|b_1-b_0|\}\}.\]
By Lem.~\ref{lem:FDD-global}, we know that $\dfd(\RRf,\RRg) \geq \max\{|a_1-b_1|,|a_0-b_0|\}.$

Now, suppose that $\RRf$ and $\RRg$ both have $n > 1$ connected component. Then the bottleneck distance on $\Ext_0$ finds the best pairing between the global pairs of each component, possible mapping some pairs to the diagonal. When computing the functional distortion distance on multiple connected components, we find best way to map one component of $\RRf$ to a component of $\RRg$; see Sec.~\ref{sec:appx:multipleConnected}. For every possible mapping between global pairs, we are guaranteed that $\dfd \geq \db{0}{\Ext}$.

 Thus, $\dfd(\RRf,\RRg) \geq \dbc{0}{Ext}(\RRf,\RRg)$

\end{proof}

Again, the relationship with between the interleaving distance and $\Ext_0$ is missing in the literature, so we provide it here.

\begin{proposition}\label{prop:interleaving-discrim-0-ext}
The interleaving distance is more discriminative than $\dbc{0}{Ext}$. That is
\[\dbc{0}{Ext}(\RRf,\RRg) \leq d_{I}(\RRf,\RRg).\]
\end{proposition}
\begin{lemma}\label{lem:interleaving-global}
Let $\RRf$ and $\RRg$ be two connected, constructible Reeb graphs. Let $v_0,v_1$ and $u_0,u_1$ be the global min and global max of $\RRf$ and $\RRg$, respectively. Then, 
\[d_{I}(\RRf,\RRg) \geq \max\{|\rf(v_1)-\rg(u_1)|,|\rf(v_0)-\rg(u_0)\}.\]
\end{lemma}

\begin{proof}
Let $\rf(v_1) = a_1, \rf(v_0) = a_0, \rg(u_1) = b_1,$ and $\rg(u_0) = b_0$. Without loss of generality, assume $b_1 > a_1$ and assume $\max\{|a_1-b_1|,|a_0-b_0|\} = b_1 - a_1$. Then, set let $A = b_1-a_1$. 

Suppose that cosheaves $\F$ and $\G$ are $\e$-interleaved with $\e < A$. Then, there exists a $\delta > 0 $ such that $\e +\delta < A$ and an interval $I = (c - \delta, c + \delta)$ such that $I^{\e} = (c - \delta - \e, c + \delta + \e)$. Since $a_1$ is the global max of $\RRf$ and $c - \delta - \e > a_1$, we have that $\G(I) \neq \emptyset$ while $\F(I^{\e}) = \emptyset$. This implies that $\text{Im}(\mathcal{S}_{e}[\phi] \circ \psi) = \emptyset$ and $\text{Im}(\sigma^{2\e}_{\G}) \neq \emptyset$, meaning that $\F$ and $\G$ cannot be $\e$-interleaved. 

\end{proof}

\begin{proof}[Proof of Prop.~\ref{prop:interleaving-discrim-0-ext}]
Just as in the proof for Prop.~\ref{prop:FDD-discrim-0-ext}, we know that $\db{0}{\Ext}(\RRf,\RRg) = \min\{\max\{|a_1-b_1|,|a_0-b_0|\},\max\{\frac12|a_1-a_0|,\frac12|b_1-b_0|\}\}.$, where $(a_0,a_1)$ and $(b_0,b_1)$ are the global persistence pairs of $\RRf$ and $\RRg$.

By Lem.~\ref{lem:interleaving-global}, we know that $d_I(\RRf,\RRg) \geq \max\{|a_1-b_1|,|a_0-b_0|\}$. Thus, $d_{I}(\RRf,\RRg) \geq \dbc{0}{Ext}(\RRf,\RRg)$.
\end{proof}

We would like to make note that the discriminativity bounds for the interleaving and functional distortion distance bounds which we have discussed above are all tight; see Fig.~\ref{fig:distanceLandscape}. Finally, we focus on the discriminativity of the edit distances.

\begin{corollary}[{\cite[Corollary~40, Corollary~41]{DiFabio2016}}]
Let $\RRf, \RRg \in \MM$. Then the Reeb graph edit distance is more discriminative than the bottleneck distance, interleaving distance, and functional distortion distance. That is,
\[d_B(\RRf,\RRg) \leq d_E(\RRf,\RRg)\]
and
\[d_I(\RRf,\RRg) \leq \dfd(\RRf,\RRg) \leq d_E(\RRf,\RRg).\]
\end{corollary}

\begin{corollary}
Let $\RRf$ and $\RRg$ be PL Reeb graphs. Then, the universal distance is more discriminative than the bottleneck, the interleaving, and the functional distortion distance. That is
\[d_B(\RRf,\RRg) \leq \du(\RRf,\RRg)\]
and
\[d_I(\RRf,\RRg) \leq \dfd(\RRf,\RRg) \leq \du(\RRf,\RRg).\]
\end{corollary}

\begin{proof}
This follows from the universality of the universal distance coupled with stability of the bottleneck, functional distortion, and interleaving distance; see Sec.~\ref{sec:universality}.
\end{proof}

\subsubsection{Discriminativity on Contour Trees}
The \textbf{contour tree} is a Reeb graph where the underlying domain $\X$ is simply connected, i.e. $\X$ is path-connected and its fundamental group is trivial. Reeb graphs have the inherent property that the $1^{st}$ Betti number is bounded above by that of the original domain \cite{DeyWang2021}. This implies that the Reeb graph of a domain with a first Betti number of $0$ has no cycles and is thus a tree. In this case, we have a stronger result for the disciminativity of the functional distortion distance compared to the bottleneck distance.

\begin{proposition}\label{prop:FDD-discrim-contour}
The functional distortion distance is more discriminative than the bottleneck distance with a discriminativity constant of 1 when the Reeb graphs are defined on simply connected domains. That is
\[d_b(\CC_f,\CC_g) \leq d_B(\CC_f,\CC_g) \leq \dfd(\CC_f,\CC_g).\]
\end{proposition} 

\begin{proof}
Let $\CC_f, \CC_g$ be two contour trees. The 1-dimensional extended persistence diagrams $\Ext_1(\rf)$, $\Ext_1(\rg)$ are then empty since there are no cycles in $\CC_f$ and $\CC_g$. Combining the bounds of Cor.~\ref{cor:FDD-discrim-0} and Prop.~\ref{prop:FDD-discrim-0-ext} we achieve \[d_b(\CC_f,\CC_g) \leq d_B(\CC_f,\CC_g) \leq \dfd(\CC_f,\CC_g).\]
\end{proof}

\subsubsection{Intrinsic metrics and discriminativity}

Given that the functional distortion distance between two Reeb graphs $\RRf$,$\RRg$ is sufficiently small, the functional distortion distance and graded bottleneck distance are strongly equivalent \cite{Carriere2017}. That is,
\begin{theorem}[{\cite[Theorem~9]{Carriere2017}}]
Let $\RR_f,\RR_g$ be two Reeb graphs whose (ordered) critical values are $\{a_1,\ldots,a_n\}$ and $\{b_1,\ldots,b_m\}$, respectively. Let $a_f = \min_{1 \leq i < n} a_{i+1} - a_i$, $a_g = \min_{1 \leq j < m} b_{j+1} - b_j$, and let $K \in (0, 1/22]$. If $\dfd(\RRf , \RRg) \leq \max\{a_f , a_g\}/(8(1 + 22K))$, then
\[K\dfd(\RRf,\RRg) \leq d_B(\RRf,\RRg) \leq 3\dfd(\RRf,\RRg).\]
\end{theorem}

The original Theorem \cite[Theorem~9]{Carriere2017} uses a constant of $2$ rather than $3$ for the upper bound on the functional distortion distance. This is due to an incorrect citation to \cite{Bjerkevik2021} which is using the ungraded bottleneck distance rather than the graded bottleneck distance.

This theorem leads to the construction of \textbf{induced intrinsic metrics}: the intrinsic bottleneck distance $\hat{d}_B$, the intrinsic interleaving distance $\hat{d}_I$, and the intrinsic functional distortion distance $\hat{d}_{FD}$. Intuitively, these metrics define the cost between two Reeb graphs $\RRf$ and $\RRg$ to be the length of the longest path from $\RRf$ to $\RRg$ in the space of Reeb graphs, where the length of the path is defined using the original distances $d_B$, $d_I$, and $\dfd$. These metrics have the property that they are all strongly equivalent to one another.

\subsection{Universality}\label{sec:universality}

If we believe that discriminativeness is a desirable property for comparison of Reeb graphs, it would be wise to find distances which are as discriminative as possible -- a \textbf{universal} distance. 

\begin{property}[\textbf{Universality}]\label{p:universality}
We say that a stable distance $d_U$ between two Reeb graphs $\RRf$ and $\RRg$ is \textbf{universal} if for all other stable distances $d_S$, we have 
\[d_S(\RRf,\RRg) \leq d_U(\RRf,\RRg).\]
\end{property}

\begin{theorem}[{\cite[Corollary~5.7]{Bauer2020}}]
The universal distance is universal.
\end{theorem}

\begin{proposition}
Neither the functional distortion distance, interleaving distance, nor bottleneck distance are universal.
\end{proposition}

\begin{proof}
While Ex.~\ref{example:graphIsoEqPers} is not using a PL Reeb graph, the fact that these are compact 2-manifolds guarantees that it is triangulable. Furthermore, we can easily define a PL functions $\hat{f}$ and $\hat{g}$ which approximate $f$ and $g$ well enough to ensure that $(\X,\hat{f})$ and $(\Y,\hat{g})$ have Reeb graphs equivalent to $\RRf$ and $\RRg$, respectively. Thus, Ex.~\ref{example:graphIsoEqPers} is enough to show that these distances are not universal.
\end{proof}

\truncatedOptionalImportant{
\begin{remark}
Since the truncated interleaving distance is not a stable distance, there is no guarantee that the universal distance bounds the truncated interleaving distance for any fixed $m$.
\end{remark}
}

\subsection{Path Component Sensitivity}\label{subsec:PCSensitivity}

Path component sensitivity is an inherent property of each of the Reeb graph metrics -- it essentially states that the Reeb graph metrics are not readily equipped to compare Reeb graphs with different numbers of connected components.

\begin{property}[\textbf{Path Component Sensitivity}]\label{p:pathComponentSensitive}
Let $(\X,f)$ and $(\Y,g)$ be two constructible scalar fields. 
 We say that $d$ is \textbf{path component sensitive} if  $d(\RRf,\RRg) = \infty$
 if and only if 
 $\RRf$ and $\RRg$ have a different number of path connected components.   
\end{property}

\begin{proposition}
The bottleneck distance is not path component sensitive.
\end{proposition}
\begin{proof}
Suppose $\RRf$ has more connected components than $\RRg$. By definition, the extended persistence diagram has no points at $+\infty$ and the bijection $\eta$ between two extended persistence diagrams does not discern from which connected components the features arise. Thus, we are able to pair features from multiple connected components in $\RRf$ to the same connected component in $\RRg$. Since all points are finite, the bottleneck distance is also finite. Thus, the bottleneck distance is not path component sensitive.
\end{proof}

In the next statements we see that all the Reeb graph metrics are path component sensitive.

\begin{proposition}[{\cite[Proposition~4.5]{deSilva2016}}]
The interleaving distance is path component sensitive.
\end{proposition}

\begin{corollary}
The functional distortion distance is path component sensitive.
\end{corollary}

\begin{proof}
This follows directly from strong equivalence of the interleaving distance and functional distortion distance.
\end{proof}

\begin{proposition}\label{prop:universal-pcs}
The universal distance is path component sensitive.
\end{proposition}

\begin{proof}
From Remark 3.4 in \cite{Bauer2020}, we can see that the universal distance is $\infty$ if the number of path connected components of $\RRf$ and $\RRg$ differ. If they do not differ, then they are finite due to stability of the universal distance. Corollary 4.3 in \cite{Bauer2020} states that the universal distance is equivalent to our universal distance, completing the proof.
\end{proof}

The Reeb graph edit distance is not defined with multiple connected components in mind \cite{DiFabio2016}. Thus, the Reeb graph edit distance is the only distance we have defined that is just a pseudometric on Reeb graphs, not an extended pseudometric. This is due to the fact that there are no elementary deformations which are able to insert isolated vertices or delete edges between vertices without deleting the vertices as well.

\truncatedOptional{
\begin{proposition}[{\cite[Proposition~2.19]{Chambers2021}}]
For a fixed $m\in[0,1)$, the truncated interleaving distance is path component sensitive.
\end{proposition}
}
\section{Examples}
\label{sec:examples}

Here we discuss several examples which present the differences between the bottleneck, interleaving, functional distortion, and Reeb graph edit distance. In each, we present the original scalar field, corresponding Reeb graphs and full extended persistence diagram, as well as figures showing how each of the Reeb graph metrics can be computed. 
Each of the scalar field examples are compact 2-manifolds, all without boundary except for Ex.~\ref{example:globalMaxChange}. As stated before, the universal distance is only defined for PL Reeb graphs. It is known, however, that each compact 2-manifold is triangulable. Furthermore, we can also construct PL functions which approximate the original function well enough such that the Reeb graphs are identical. Thus, we will show the universal distance for each of the examples defined below.
For the interleaving distance diagrams, we will alternate between several slightly different representations that we believe are the most enlightening for the particular example we are discussing. However, each description of the interleaving distance will use notation that is inline with the cosheaf definition.
Since the truncated interleaving distance acts similarly to the original interleaving distance due to strong equivalence (Cor.~\ref{cor:strongEquivTruncated}, we leave out this distance for all examples except for Ex.~\ref{example:globalMaxChange}.

Fig.~\ref{fig:exampleSummary} provides a visual summary of the various examples we have show. We label the value corresponding to $2d_I$ to show that none of the Reeb graph metrics pass this bound. We also show the $L^{\infty}$ distance where applicable. 

\begin{figure}
    \centering
    \includegraphics[width=0.9\textwidth]{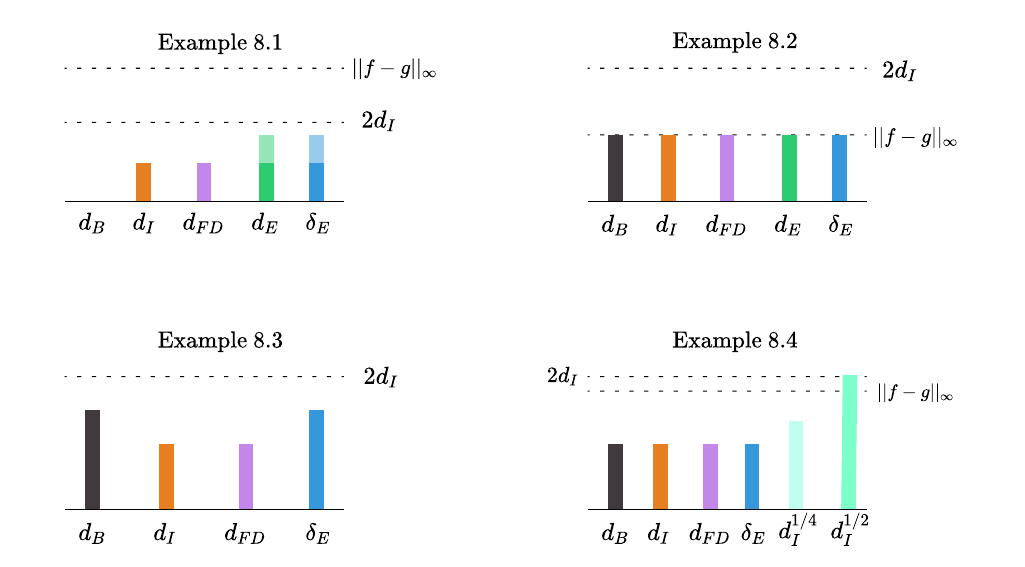}
    \caption{Visual summary of the distance values attained for each example. We use transparency of the bars in Example 1 to illustrate the range of possible distance values of $d_E$ and $\delta_E$. For Example 4, we show two different choices of $m$ for the truncated interleaving distance to illustrate the affect that $m$ plays on the distance values.}
    \label{fig:exampleSummary}
\end{figure}

\subsection{Example: Scalar fields with graph isomorphic Reeb graphs and identical persistence diagrams}
\label{example:graphIsoEqPers}

Let $(\X,f)$, $(\Y,g)$ be constructible scalar fields as shown in  Fig.~\ref{fig:exampleOne}(a). Both $\X$ and $\Y$ are compact 2-manifolds without boundary. The labels $\{a_1,\ldots,a_8\}$ represent identical function values for both scalar fields. We denote the vertices of $\RRf$ and $\RRg$ corresponding to $\{a_1,\ldots,a_8\}$ as $\{v_1,\ldots v_8\}$ and $\{u_1,\ldots,u_8\}$, respectively.

\begin{figure}
    \centering
    \includegraphics[width=\textwidth]{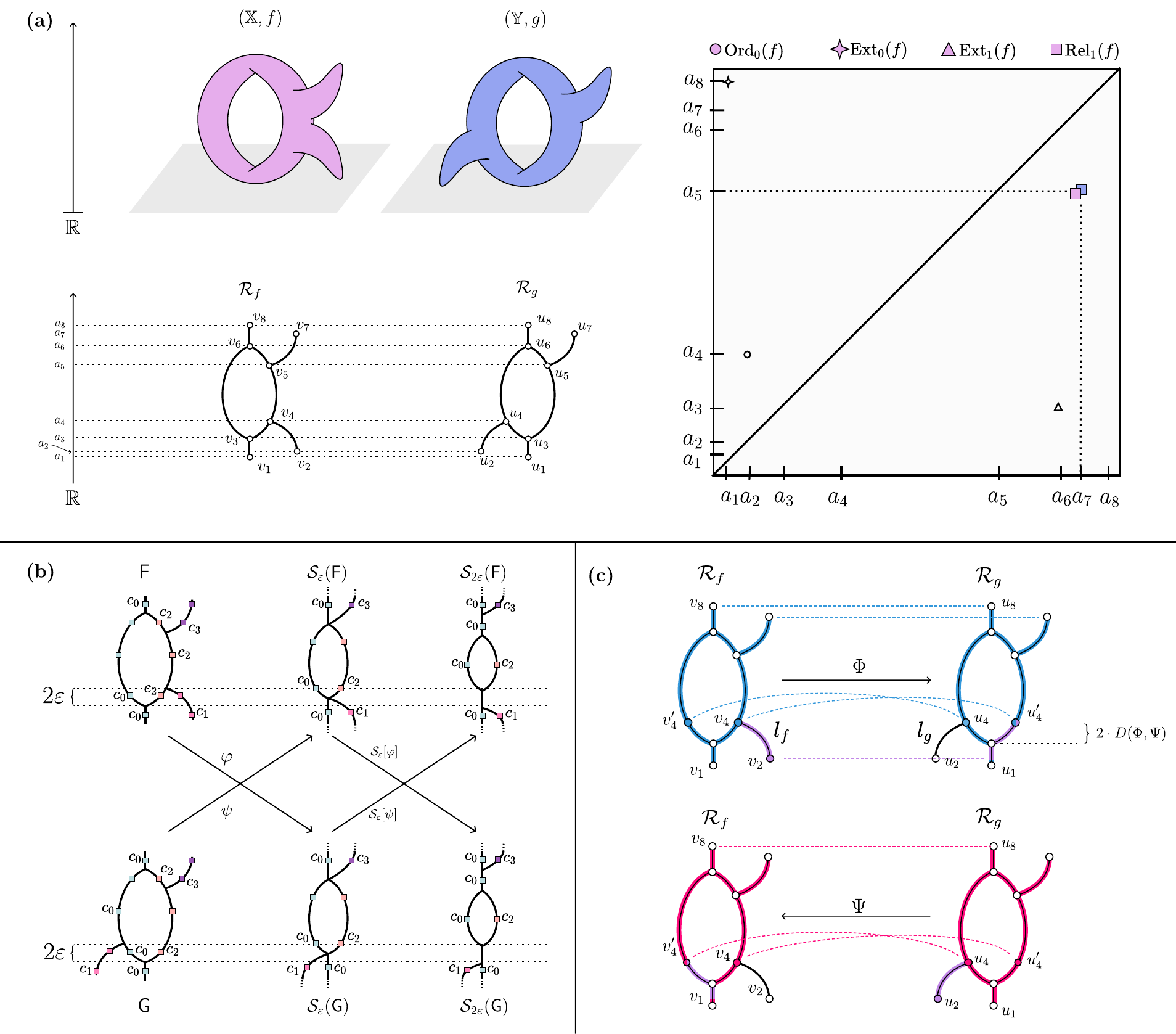}
    \caption{Summary figure for Ex.~\ref{example:graphIsoEqPers}. \textbf{(a)} Two scalar fields $(\X,f)$ and $(\X,g)$ with their corresponding Reeb graphs $\RRf$,$\RRg$ below them. To the right is the persistence diagrams of the scalar fields overlayed on top of each other. Smaller points without color indicate that these features are plotted in both diagrams. \textbf{(b)} Diagram displaying the interleaving distance between the two Reeb graphs. \textbf{(c)} Diagram displaying the optimal mappings $\Phi$,$\Psi$ for the functional distortion distance.}
    \label{fig:exampleOne}
\end{figure}

\paragraph{Bottleneck Distance} The persistence diagrams are identical, making the bottleneck distances equal to 0. That is
\[d_b(\RRf,\RRg) = d_B(\RRf,\RRg) = 0.\]

\paragraph{Interleaving Distance} Before determining the interleaving distance between these two scalar fields, recall that from Fig.~\ref{fig:graphIsoNonRGIso} we know that their Reeb graphs are isomorphic as \emph{graphs}, but not as \emph{Reeb graphs}.  Fig.~\ref{fig:nonZeroInterleaving} shows why these are not isomorphic as cosheafs as well. By  Property~\ref{p:isomorphismInvariance}, this implies that the two Reeb graphs do not have a $0$-interleaving between them. 

The squares overlayed on the Reeb graphs of Fig.~\ref{fig:exampleOne}(a) represent the elements of the cosheafs at various fibers. For there to be a $\e$-interleaving, we need to construct a map from $\F$ to $\SS_{\e}(\G)$ such that these components respect inclusion. In this case, our issue is that the component labeled $c_1$ eventually merges into $c_2$ in $\F$, but merges into $c_0$ in $\G$, while the rest of the components are equivalent in both. Switching the labelings of $c_0$ and $c_2$ in $\G$ results in a similar situation for the upper leaf. 

We must choose $\e$ so that the maps $\phi:\F \to \SS_{\e}(\G)$,$\psi:\G \to \SS_{\e}(\F)$ respect inclusions (and are therefore well-defined morphisms between the cosheafs). Our only choice is choose $\e$ large enough so that the root of the component $c_1$ merges with the rest of the Reeb graph below the 1-cycle. Choosing $\e = \frac{1}{2}|a_4 - a_3|$ achieves this. We can see from the diagram that this allows us to choose maps $\phi$ and $\psi$ which respect inclusion and allow for $\SS_{\e}[\psi] \circ \phi = \sigma_{\F}^{2\e}$ and $\SS_{\e}[\phi] \circ \psi = \sigma_{\G}^{2\e}$. Thus, since this is the optimal choice of $\e$, we have that \[d_I(\RRf,\RRg) = \frac12(a_4-a_3).\]

\paragraph{Functional Distortion Distance} Let $\ell_f$ denote the downward leaf of $\RRf$ whose minimum is $v_2$ and which merges with the rest of $\RRf$ at $v_4$, and let $\ell_g$ be defined analogously.  To define continuous maps between $\RRf$ and $\RRg$, first note that if we remove $\ell_f$ and $\ell_g$ from $\RRf$ and $\RRg$, the Reeb graphs become isomorphic. Thus, we can define $\Phi$ as being the Reeb graph isomorphism from $\RRf \setminus \ell_f$ to $\RRg \setminus \ell_g$, and $\Psi$ as its inverse. By continuity, there is no way for us to map the bottom leaves directly to each other without disturbing the identity map that we have just defined. Thus, our best plan is to assign $\ell_f$ to be horizontally ``flattened'' as in  Fig.~\ref{fig:exampleOne}(c) and assign $\ell_g$ to analogously.

The largest point distortion then comes from the pairs $(v'_4,u'_4),(v_4,u'_4)$ (or similarly $(v'_4,u'_4),(v'_4,u_4)$). Thus, we have
\[D(\Phi,\Psi) = \lambda((v_2,u'_4),(v_4,u'_4)) = \frac12|d_{\rf}(v'_4,v_4) - d_{\rg}(u'_4,u'_4)| = \frac12(a_4-a_3).\]

Since the functional distortion distance is bounded below by the interleaving distance by \cref{thm:strongEquiv}, we have that \[d_{FD}(\RRf,\RRg) = \frac12(a_4-a_3).\]

Note that we could have defined $\Phi$ and $\Psi$ such that they contracted the leaves into a single point, which would result in the same exact map distortion $D(\Phi,\Psi)$. However, this would cause $||\rf-\rg\circ\Phi||_{\infty} = ||\rf\circ\Psi - \rg||_{\infty} = (a_4-a_2)$, which would imply that $\Phi$ and $\Psi$ are not optimally chosen.

\begin{figure}
    \centering
    \includegraphics[width=\textwidth]{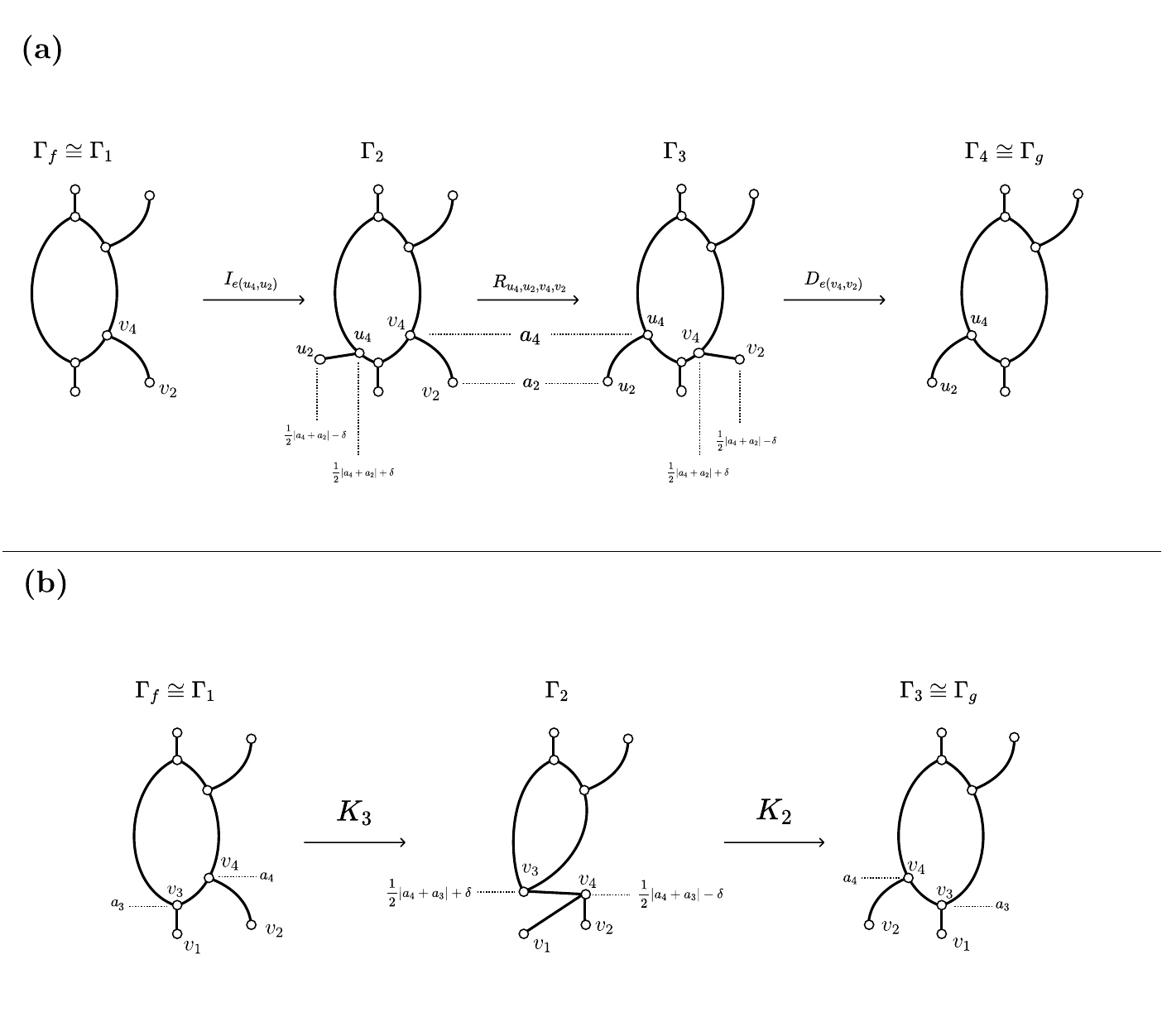}
    \caption{\textbf{(a)} Reeb graph edit distance for Ex.~\ref{example:graphIsoEqPers}/Fig.~\ref{fig:exampleOne} through an insertion of a small leaf, then relabeling vertices to simultaneously increase this leaf and then flatten the leaf we intend to remove, finally followed by a deletion of the shrunken leaf. \textbf{(b)} An alternate, suboptimal sequence carrying $\Gamma_f$ to $\Gamma_g$ through a K-3 type operation followed by a K-2 operation to move the target leaf from the right side of the stem to the left side of the stem.}
    \label{fig:exampleOneRGED}
\end{figure}

\paragraph{Reeb graph edit distance.} To compute the Reeb graph edit distance $d_E$, we investigate two different edit sequences carrying $(\Gamma_f,\ell_f)$ to $(\Gamma_g,\ell_g)$. While the first sequence we show yields a lower edit cost in this case, we find it appropriate to investigate both sequences since their cost is dependent on the values of different vertices. Specifically, if the length of $e(v_2,v_4) \in E(\Gamma_f)$ and $e(u_2,u_4) \in E(\Gamma_g)$ was to increase such that $\frac12|a_4 - a_2| > |a_3 - a_2|$, the actual edit distance would be governed by the latter sequence rather than the former. 

The most intuitively direct way of carrying $(\Gamma_f,\ell_f)$ to $(\Gamma_g,\ell_f)$ would be to delete the leaf $e(v_4,v_2) \in E(\Gamma_f)$ and then birth a new leaf to correspond to $e(u_4,u_2) \in E(\Gamma_g)$. As stated in Rem.~\ref{rem:spreadRelabel}, it is cheaper to first birth a flat edge, perform a relabel deformation to change both edges simultaneously, and then delete the newly flattened edge.  Fig.~\ref{fig:exampleOneRGED}(a) shows this sequence.

Let $(\Gamma_1,\ell_1) \cong (\Gamma_f, \ell_f)$. We insert the edge $e(u_4,u_2)$ such that $\ell_2(u_4) = \frac12|a_4 + a_2| + \delta$,  $\ell_2(u_2) = \frac12|a_4 + a_2| - \delta$, with $0 < \delta << 1$, making the length of this leaf $2\delta$. This $\delta$ is added strictly for purposes of aligning with the definition of elementary birth deformation which states that the introduced vertices cannot have the same function value. A relabel operation is then applied to $(\Gamma_2,\ell_2)$ so that $l_3(u_4) = a_4$, $l_3(u_2) = a_2$, $l_3(v_4) = \frac12|a_4-a_2| + \delta$, and $l_3(v_2) = \frac12|a_4-a_2| + \delta$. A death deformation is then applied to $e(v_4,v_2)$. We compute the cost of this edit sequence $S_1$ to be

\[c(S_1) = \bigg[\frac12(2\delta)\bigg] + \bigg[\frac12|a_4-a_2| - \delta\bigg] +  \bigg[\frac12(2\delta)\bigg] =  \frac12|a_4-a_2| + \delta.\] As $\delta \to 0$, we have $c(S_1) \to \frac12|a_4-a_2|$.

Fig.~\ref{fig:exampleOneRGED}(b) shows a sequence $S$ carrying $(\Gamma_f,\ell_f)$ to $(\Gamma_g,\ell_g)$ through a $K_3$-type deformation followed by a $K_2$-type deformation. Simply put, we move the leaf $e(v_2,v_4)$ to the other side of the Reeb graph. We use the $K_3$-type deformation to set $\ell_f(v_3) = \frac{1}{2}|a_4 + a_3| + \delta$ and  $\ell_f(v_4) = \frac{1}{2}|a_4 + a_3| - \delta$, where $0 < \delta << \frac{1}{2}|a_4 - a_3|$. Composing the two $K$-type deformations means that the cost of this edit sequence is 
\begin{align*}
    c(S_2) = & \big[\max\{|\ell_f(v_4)-\ell_2(v_4)|,|\ell_f(v_3)-\ell_2(v_3)|\}\big] \\       & +  \big[\max\{|\ell_2(v_4)-\ell_g(v_4)|,|\ell_2(v_3)-\ell_g(v_3)|\}\big] \\
           = & |a_4-a_3|+2\delta.
\end{align*}

\noindent As we have $\delta \to 0$, we have $c(S_2) \to |a_4-a_3|$. Thus, since the Reeb graph edit distance is defined as the infimum cost among all sequences and is bounded below by the functional distortion distance, we have \[\frac{1}{2}|a_4-a_3| \leq d_E(\RRf,\RRg) \leq \frac12|a_4-a_2|\].

\begin{figure}
    \centering
    \includegraphics[width=\textwidth]{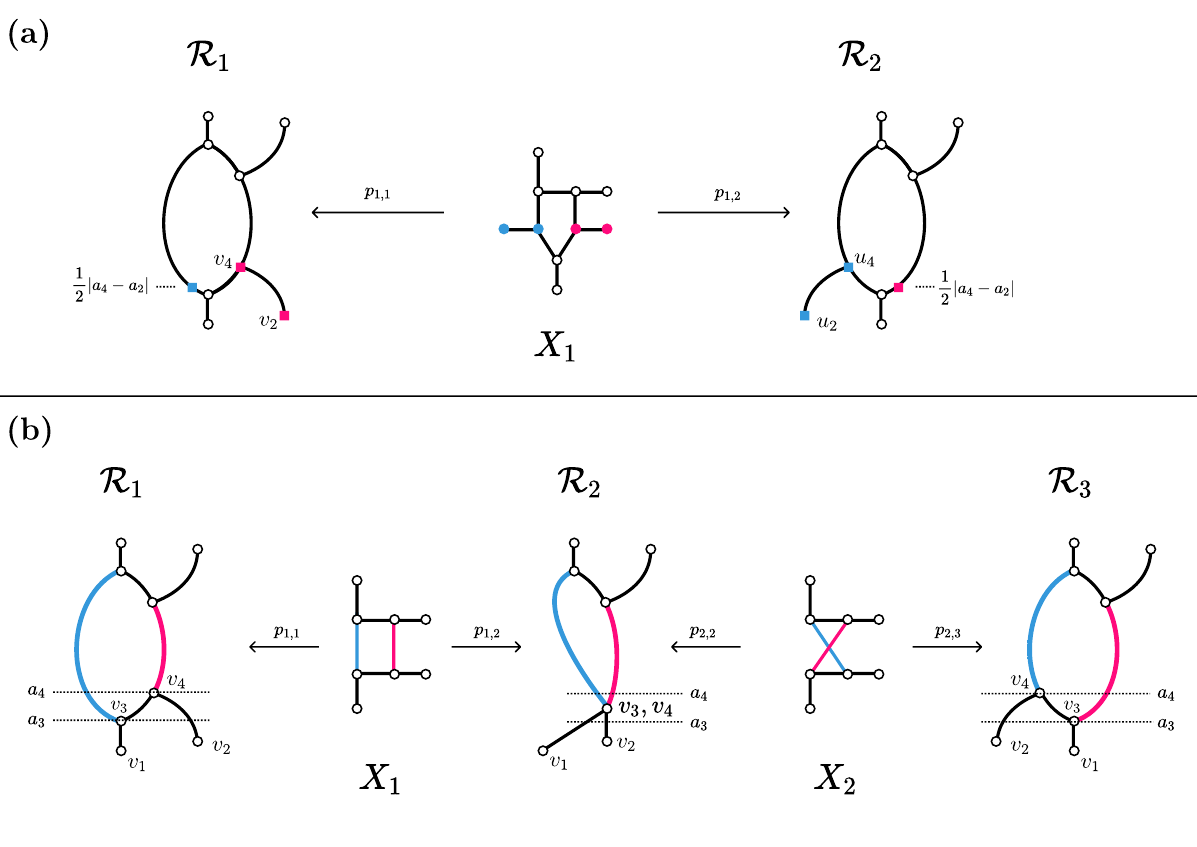}
    \caption{\textbf{(a)} Depiction of the universal distance by inserting a new leaf on the leaft side and deleting the old leaf on the right side. \textbf{(b)} A suboptimal sequence for the universal distance is given by mapping $X_1$ and $\X_2$ both into the common space $\RR_2$ which allows us to switch the leaf from the right side to the left side of the Reeb graph.}
    \label{fig:exampleOneuniversal}
\end{figure}

\paragraph{Universal Distance} To compute the universal distance, we investigate two different zigzag diagrams which are analogous to the first and second edit sequences from Fig.~\ref{fig:exampleOneRGED}, respectively. 

The zigzag diagram $z_1$ in  Fig.~\ref{fig:exampleOneuniversal}(a) consists of just a single space $X_1$ with different Reeb quotient maps going to $\RR_1 \cong \RRf$ and $\RR_2\cong \RRg$. The left quotient map contracts the edge $e(x_1,x_2)$ to a single point in $\RR_1$, and maps $e(x_3,x_4)$ to $e(v_4,v_2) \in \RR_1$. The right quotient map is defined similarly. Contracted points in $\RR_1$ and $\RR_2$ are such that their function value is half the distance between the root and tip of that leaf. Specifically, the function value is $\frac12|a_4-a_2|$. 

If we were to construct universal deformations moving $\RR_1$ to $\RR_2$ in this way, we would have a relabel operation to flatten the leaf, an insert operation to add the new leaf, followed by another relabel deformation to alter the newly introduced leaf, and finally a death deformation to remove the previously flattened leaf. 

To compute the cost of this sequence, first note that $X_1$ is the limit of this diagram. We have that $(\rf_1 \circ p_{1,1})(x_1) - (\rf_2 \circ p_{1,2})(x_1) = \frac12|a_4-a_2|$ and similarly for $x_2,x_3,x_4 \in X_1$. Thus, the cost of this sequence $c_{z_1}$ is $\frac12|a_4-a_2|$.

Fig.~\ref{fig:exampleOneuniversal}(b) shows an zigzag diagram $z_2$ carrying $\RRf$ to $\RRg$. We have a relabel operation, followed by two slide operations, and then a final relabel operation. As we perform the two slide operations, since no relabel operation is needed between them, the Reeb graphs $\RR_2$ and $\RR_3$ are identical. This implies that we could in fact shrink this zigzag diagram slightly by removing both $X_2$ and $\RR_3$. However, removing these two pieces in the diagram will not change the final cost of the sequence since only relabel operations can affect the cost. 

Given that this diagram  is the most optimal diagram carrying $\RRf$ to $\RRg$, we compute the cost by constructing the iterated pullback of $X_1,X_2$, and $X_3$. Now, let $x_1 \in X_1$ be the point along edge $e(v_4,v_5)$ such that $(\rf_1\circ p_{1,1})(x_1) = a_4 + \delta_0$, with $0 < \delta_0 << 1$. Let $x_2\in X_2$ be such that $(f_2\circ p_{1,2})(x_1) = (\rf_2\circ p_{2,2})(x_2)$, and let $x_3 \in X_3$ be such that $(\rf_3\circ p_{2,3})(x_2) = (\rf_3\circ p_{3,3})(x_3)$. This implies that $(x_1,x_2,x_3) \in X_1 \times_{\RR_2} X_2 \times_{\RR_3} X_3$. We can see then that $(\rf_4 \circ p_{3,4})(x_3) = a_3 + \delta_1$, where $\delta_1$ approaches $0$ as $\delta_0$ approaches $0$. Thus, as $\delta_0 \to 0$, we have $c_{z_2} \to |a_4-a_3|$. Finally, as in the Reeb graph edit distance case, this implies that \[\frac12|a_4-a_3|\leq\delta_E(\RRf,\RRg) \leq \frac12|a_4-a_2|.\]

\subsection{Example: Stretched Tori results in equality of metrics.}
\label{example:stretchedTori}

\begin{figure}
    \centering
    \includegraphics[width=\textwidth]{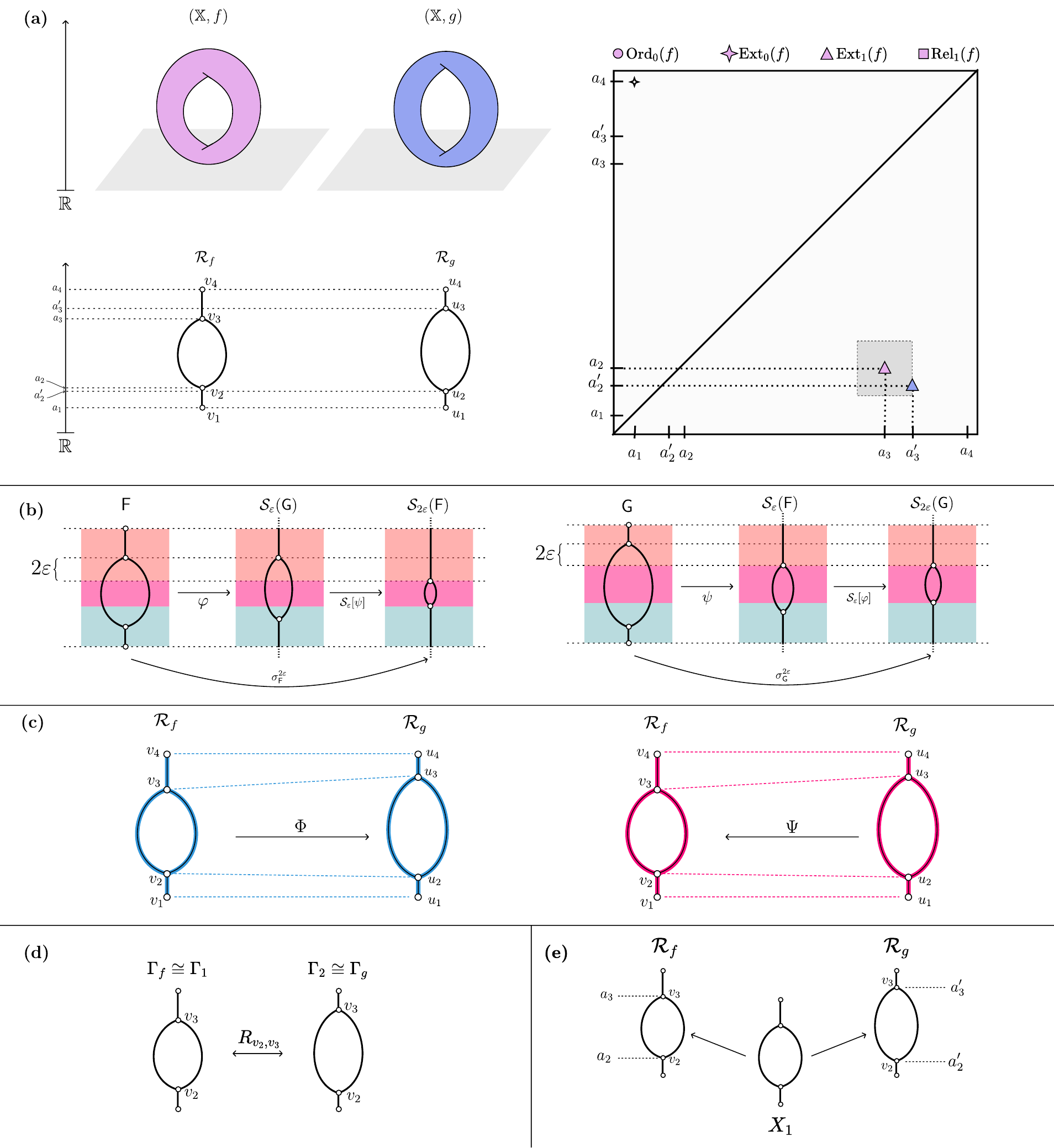}
    \caption{Summary figure for Ex.~\ref{example:stretchedTori} \textbf{(a)} Two scalar fields $(\X,f)$ and $(\X,g)$ with their corresponding Reeb graphs $\RRf$,$\RRg$ below them. To the right is the persistence diagrams of the scalar fields overlayed on top of each other. Smaller points without color indicate that these features are plotted in both diagrams. The colored points indicate features which are unique to that diagram. \textbf{(b)} The interleaving distance between the two Reeb graphs. We choose to lay these mappings horizontally instead of interleaved as in Ex.~\ref{example:graphIsoEqPers}(b) to more adequately express the mappings between the pre-cosheafs. \textbf{(c)} Optimal mappings $\Phi$,$\Psi$ for the functional distortion distance. \textbf{(d)} Optimal sequence of edit operations for the Reeb graph edit distance. \textbf{(e)} Shortest and optimal zigzag diagram for the universal distance.}
    \label{fig:exampleTwo}
\end{figure}

Let $(\X,f)$ and $(\X,g)$ be two scalar fields defined as in  Fig.~\ref{fig:exampleTwo}. Let $\{v_1,\ldots,v_4\}$ and $\{u_1,\ldots,u_4\}$ denote the vertices of $\RRf$ and $\RRg$, respectively, defined in increasing function value order. Note that while these scalar fields are defined on the same $\X$ (a torus), we assign $f$ and $g$ such that $a_2 = \rf(v_2) > \rg(u_2) = a'_2$ and $a_3 = \rf(v_3) < \rg(v_3) = a'_3$. 

\paragraph{Bottleneck Distance} We can see that the persistence diagrams for the tori are quite similar. The smallest square centered at $(a_3,a_2)$ which encompasses $(a'_3,a'_2)$ has side length equal to $a'_3 - a_3$, implying that the bottleneck distance between these diagrams is $d_{B^1} = a'_3 - a_3$. Since the 0 dimensional extended diagrams are identical, we have that \[\dB{}(\RRf,\RRg) = a'_3-a_3.\]

\paragraph{Interleaving Distance} We claim that $\e = a'_3 - a_3$. To show this, first note that there will be no issue with mapping the correct components so that they respect inclusion since the spaces are identical besides the stretching of the loop. Our concern is that, in the interleaving, we may have a point $a\in \R$ such that $|\G(a^{\delta})| = |\mathcal{S}_{2\e}(\G)(a^{\delta})| = 2$ and $|\mathcal{S}_{\e}(\F)(a^{\delta})| = 1$, where $a^{\delta} = (a-\delta,a+\delta)$ with $\delta > 0$. This would mean the map $\psi \circ \mathcal{S}_{\e}[\phi] \neq \sigma^{2\e}_{\G}$ since the image of $\psi$ would be one component, causing the image of $\psi \circ \mathcal{S}_{\e}[\phi]$ to be one component as well, while the image of $\sigma^{2\e}_{\G}$ would be two components.

Smoothing this loop by $\e$ will cause the maxima node to shift down by $\e$ and the bottom node to be shifted up by $\e$. Thus, the values of the nodes of $\mathcal{S}_{\e}(\F)$ would have values $a_3 - \e$ and $a_2+\e$. Similarly, the nodes for $\mathcal{S}_{2\e}(\G)$ would have values $a'_3-2\e$ and $a'_2+2\e$. 

If $\e = a'_3 - a_3$, we have $2a_3 - a'_3$, $a_2+a'_3-a_3$, $2a_3-a'_3$, and $a'_2 + 2a'_3 - 2a_3$ for the respective nodes. Thus, if $x\in \R$ satisfies $2a_3 - a'_3 > x > a_2+a'_3-a_3$, then $|\mathcal{S}_{\e}(\F)(x^{\delta})| = 2$. Similarly, if $x \in \R$  satisfies $2a_3-a'_3  > x > a'_2+2a'_3-2a_3$, we have $|\mathcal{S}_{2\e}(\G)(x^{\delta})| = |\G(x^{\delta})| = 2$. Since $a'_3 - a_3 > a_2 - a'_2$, we have $a'_2+2a'_3-2a_3 > a_2 + a'_3 - a_3$.

From this, we guarantee $\e = a'_3 - a_3$ will provide us an $\e$-interleaving by observing that $(2a_3-a'_3,a'_2+2a'_3-2a_3) \subset (2a_3 - a'_3, a_2+a'_3-a_3)$. We leave the other direction to the interested reader. Thus, \[d_I(\RRf,\RRg) = a'_3 - a_3.\]

\paragraph{Functional Distortion Distance} We can define optimal continuous maps $\Phi = \Psi^{-1}$ between $\RRf$ and $\RRg$ by stretching $\RRf$ so that $v_2 \mapsto u_2$ and $v_3 \mapsto u_3$, as depicted in  Fig.~\ref{fig:exampleTwo}(c). This implies that $D(\Phi,\Psi) = \frac{1}{2}(a'_3 - a_3)$, since the max point distortion is $\max\{|\rf(v_2)-\rg(u_2)|,|\rf(v_3)-\rg(u_3)|\}$. However, note that $||\rf - \rg \circ \Psi||_{\infty} = ||\rf \circ \Psi - \rg||_{\infty} = |\rf(v_3)-\rg(u_3)| = a'_3 - a_3$. Thus, the functional distortion distance between $\RRf$ and $\RRg$ is \[d_{FD}(\RRf,\RRg) = a'_3 - a_3.\]

\paragraph{Reeb Graph Edit Distance} Let $(\Gamma_f,\ell_f) = (\Gamma_1,\ell_1)$ and $(\Gamma_g,\ell_g) = (\Gamma_2,\ell_2)$ be the combinatorial Reeb graphs of $\RRf$ and $\RRg$, respectively.  Fig.~\ref{fig:exampleOne}(d) shows the simple, optimal sequence $S$ of edit operations to carry $\Gamma_f$ to $\Gamma_g$ consisting of one relabel operation. The Reeb graph edit distance takes the max function shift in these two relabelings, which is the shift from $a_3$ to $a'_3$ for node $v_3$. Thus, \[d_E = a'_3-a_3.\]

\paragraph{Universal Distance} Similar to the Reeb graph edit distance, the universal distance consists of one relabel operation. The maximum function value change is then from $a_3$ to $a'_3$ for node $v_3$. Thus, \[\delta_E = a'_3 - a_3\].

\subsection{Example: Genus-2 surface and simply connected domain with leaves}
\label{example:compound}

Let $(\X,f)$ and $(\Y,g)$ be the scalar fields shown in  Fig.~\ref{fig:exampleThree}(a). Let $\{v_1,\ldots,v_6\}$ and $\{u_1,\ldots,u_6\}$ denote the vertices of $\RRf$ and $\RRg$, respectively. As compared to previous examples, the domains of these two scalar fields are not homeomorphic to one another.

\begin{figure}
    \centering
    \includegraphics[width=\textwidth]{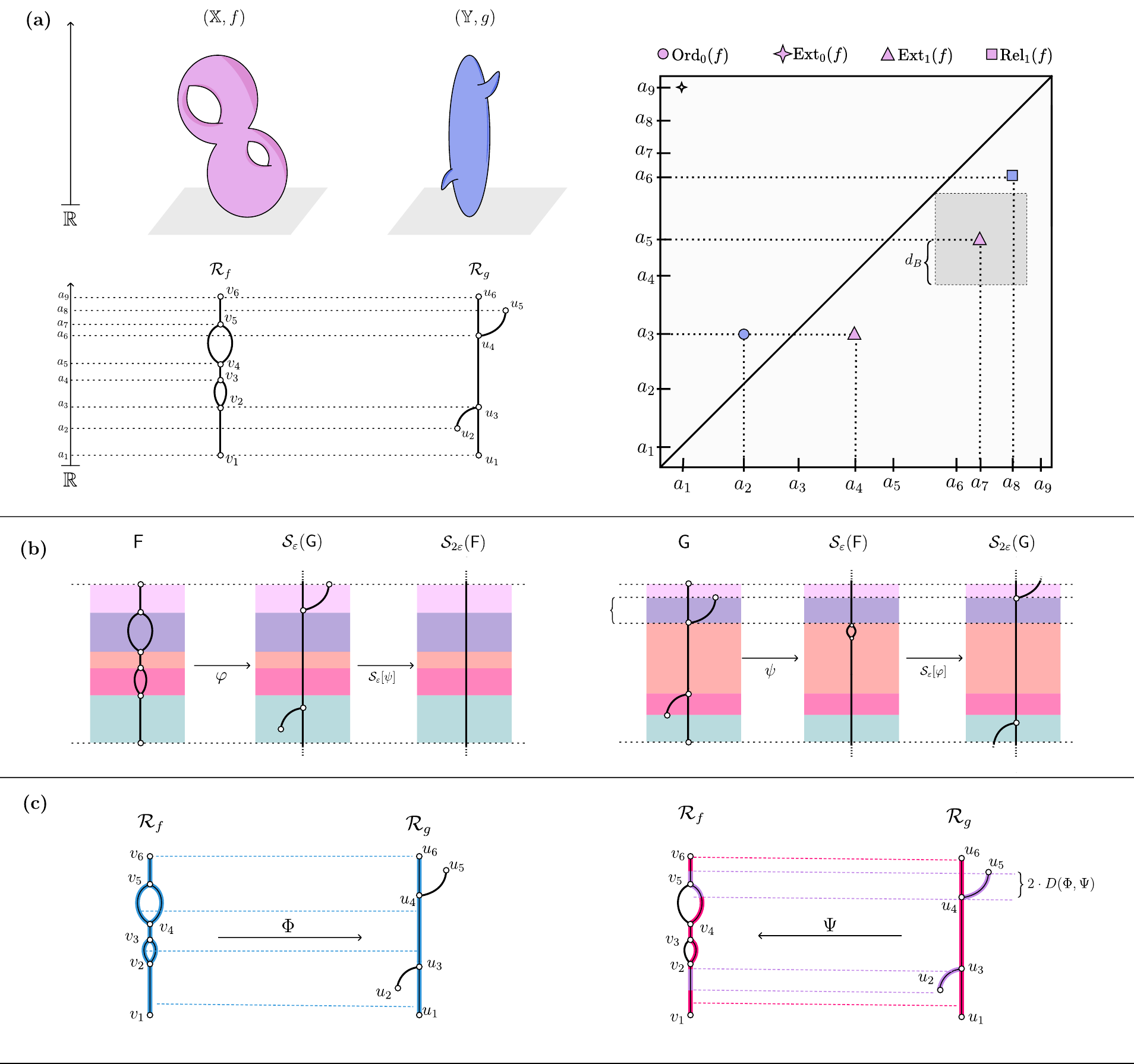}
    \caption{Summary figure for Ex.~\ref{example:compound}. \textbf{(a)} Two scalar fields $(\X,f)$ and $(\X,g)$ with their corresponding Reeb graphs $\RRf$,$\RRg$ below them. To the right is the persistence diagrams of the scalar fields overlayed on top of each other. Smaller points without color indicate that these features are plotted in both diagrams. The colored points indicate features with are unique to that diagram. \textbf{(b)} Diagram displaying the interleaving distance between the two Reeb graphs. We choose to lay these mappings horizontally instead of interleaved as in Ex.~\ref{example:graphIsoEqPers}\textbf{(b)} to more adequately express the mappings between the pre-cosheafs. \textbf{(c)} Diagram displaying the optimal mappings $\Phi$,$\Psi$ for the functional distortion distance.}
    \label{fig:exampleThree}
\end{figure}

\begin{figure}
    \centering
    \includegraphics[width=\textwidth]{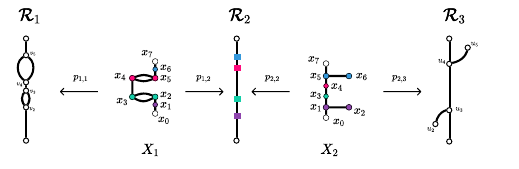}
    \caption{Figure depicting the universal distance for Ex.~\ref{example:compound}. \textbf{(a)} Zigzag diagram carrying $\RRf$ to $\RRg$. \textbf{(b)} Universal deformations which would be involved to carry $X_1$ to $X_2$. Since these deformations have no relabels, we are able to consolidate this entire sequence into just $X_1 \to \RR_2 \to X_2$.}
    \label{fig:exampleThreePt2}
\end{figure}

\paragraph{Bottleneck Distance} We can see in the persistence diagrams that the best matchings for each point is either to its duplicate in the other diagram or to the diagonal. This makes the largest distance in this matching to be $||(a_5,a_7)-(\frac{a_5+a_7}{2},\frac{a_5+a_7}{2})||_{\infty}$, implying that \[\dB{}(\RRf,\RRg) = \frac12(a_7-a_5).\]

\paragraph{Interleaving Distance} Fig.~\ref{fig:exampleThree}(b) depicts an $\e$-interleaving between the two Reeb graphs. The color coding illustrates how we can partition $\F$ and $\G$ into pieces where the number of connected components changes. We choose $\e$ to be $\frac{1}{2}(a_8-a_6)$ so that interval $I$ corresponds to only one component in $\mathcal{S}_{2\e}(\G)$ since it will always correspond to only one component in $\mathcal{S}_{\e}(\F)$. This choice of $\e$ will also cause each hole to be completely closed in $\mathcal{S}_{2\e}(\F)$, meaning the equality $\mathcal{S}_{\e}[\psi] \circ \phi = \sigma^{2\e}_{\F}$ is guaranteed. Thus, the interleaving distance is \[d_I(\RRf,\RRg) = \frac{1}{2}(a_8-a_6)\].

\paragraph{Functional Distortion Distance} In Fig.~\ref{fig:exampleThree}(c), we choose $\Phi$ to map every point straight across to $\RRg$ and choose $\Psi$ to do the same. Focusing on $\Phi$, we can see that the two points in the center of the loop between $v_4$ and $v_5$ are mapped to the same point on $\RRg$. Let $x_1$ and $x_2$ denote the points such that $\rf(x_1) = \rf(x_2) = \frac12(\rf(v_5)+\rf(v_4)) = \frac12(a_7+a_5)$. Thus, the height of any path from $x_1$ to $x_2$ is $|\rf(v_5) - \frac12(\rf(v_5)+\rf(v_4))| = |\frac12(\rf(v_5)+\rf(v_4)) - \rf(v_4)| = \frac12(a_7-a_5)$, meaning the point distortion $\lambda((x_1,\Phi(x_1)),(x_2,\Phi(x_2))) = \frac12(a_7-a_5)$.

From $\Psi$, we can see that $d_{\rg}(u_5,u'_5) = a_8 - a_6$ and $\Psi(u_5) = \Psi(u'_5)$. Thus, $\lambda((\Psi(u_5),u_5),(\Psi(u'_5),u_5)) = a_8 - a_6 > \frac{1}{2}(a_7-a_5)$. Thus, $D(\Phi,\Psi) = \frac{1}{2}(a_8-a_6)$. Since $||\rf-\rg\circ \Phi||_{\infty} = ||\rf\circ\Psi - \rg||_{\infty} = 0$ due to our functions not distorting the function values at all, we have that \[d_{FD} = \frac{1}{2}(a_8-a_6)\].

\paragraph{Reeb Graph Edit Distance} Since these Reeb graphs are constructed from two non-homeomorphic spaces, the Reeb graph edit distance is not defined.

\paragraph{Universal Distance}
Fig.~\ref{fig:exampleThreePt2}(a) shows the optimal zigzag diagram carrying $\RRf \cong \RR_1$ to $\RRg \cong \RR_2$. There is no single space $X$ which is able to map the cycles of $\RRf$ to the leaves of $\RRg$, meaning we have to ``delete" the 1-cycles and then ``insert" the two leaves. To do this, we have two spaces $X_1$ and $X_2$ which map to a common $\RR_2$ where the cycles are removed and the leaves are yet to be inserted.

To delete the cycles and insert the leaves in the most optimal way, we must note that the edges which connect to the endpoints of each cycle which are not part of the cycle themseleves -- edges $e(v_5,v_6),e(v_3,v_4)$ and $e(v_1,v_2)$ -- will essentially ``cover'' the cycles once they are removed. That is, the edge $e(v_5,v_6)$ merges with $e(v_4,v_3)$ at $v_4$, and edge $e(v_1,v_2)$ merges with $e(v_4,v_3)$ at $v_3$. Consider first the cycle between vertices $v_5$ and $v_4$. We must choose a placement of both $v_4$ and $v_5$ such that the largest distance that either of them traverses will be minimized. The optimal way to choose this is relabeling both $v_5$ and $v_4$ to have a function value of the midppoint between the vertices -- $\frac12(f(v_5)+f(v_4))$. This is the same case for the other 1-cycle and both leaves that we insert.

The universal distance is then \[\delta_E = \frac12|f(v_5) - f(v_4)| = \frac12|a_7 - a_5|,\] which is half the size of the largest 1-cycle in $\RRf$.

Fig.~\ref{fig:exampleFourEdit}(b) depicts the universal deformations that are involved in carrying $X_1$ to $X_2$. Since there are no relabels in this sequence, we are able to consolidate all of this information into a much smaller zigzag diagram.

\paragraph{Observations} This is an example of two Reeb graphs whose FDD and interleaving distance is strictly less than the bottleneck distance. This also shows how each metric is dependent on the ``largest" features in each Reeb graph rather than being dependent on multiple features. In fact, if $\RRf$ were replaced with a single edge from $v_1$ to $v_6$, the interleaving distance and FDD would both be unchanged, the bottleneck distance would be $\frac{1}{2}(a_8-a_6)$, and the universal distance would be $(a_8-a_6)$. Similarly, if $\RRg$ were replaced with a single edge, the bottleneck distance and universal distance would remain unchanged, and both the interleaving distance and FDD would be $\frac{1}{4}(a_7-a_5)$.

\subsection{Example: Stretching local maximum past the global maximum}
\label{example:globalMaxChange}

Let $(\X,f)$ and $(\X,g)$ be the scalar fields shown in  Fig.~\ref{fig:exampleFour}(a). Let $\{v_1,\ldots,v_8\}$ and $\{u_1,\ldots,u_8\}$ denote the vertices of $\RRf$ and $\RRg$, respectively. Both scalar fields are defined on the same domain. The only change from $(\X,f)$ to $(\X,g)$ is that the local maximum located at $v_7$ is assigned to $a_9$ instead of $a_7$ as its function value.  The peak surrounding $v_7$ is scaled to match.  This change ultimately makes $u_8$ the new global maximum of $(\X,g)$.

\begin{figure}
    \centering
    \includegraphics[width=\textwidth]{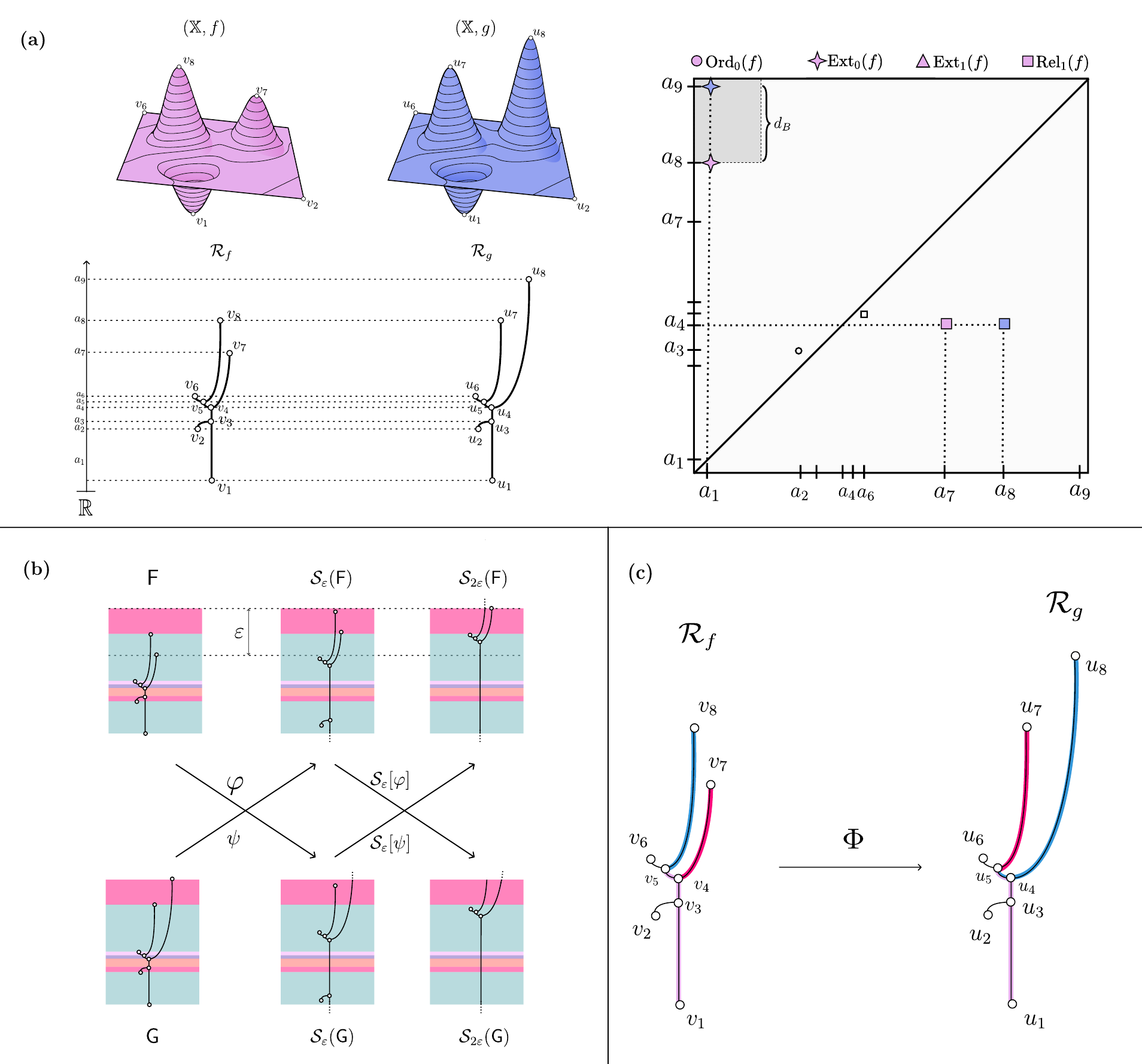}
    \caption{Summary figure for Ex.~\ref{example:globalMaxChange}.\textbf{(a)} Two scalar fields $(\X,f)$ and $(\X,g)$ with their corresponding Reeb graphs $\RRf$,$\RRg$ below them. To the right is the persistence diagrams of the scalar fields overlayed on top of each other. Smaller points without color indicate that these features are plotted in both diagrams. The colored points indicate features with are unique to that diagram. \textbf{(b)} Diagram displaying the interleaving distance between the two Reeb graphs. \textbf{(c)} Diagram displaying the optimal mapping $\Phi$ for the functional distortion distance. The map $\Psi$ is implied. \label{fig:exampleFour}}
\end{figure}

\begin{figure}
    \centering
    \includegraphics[width=\textwidth]{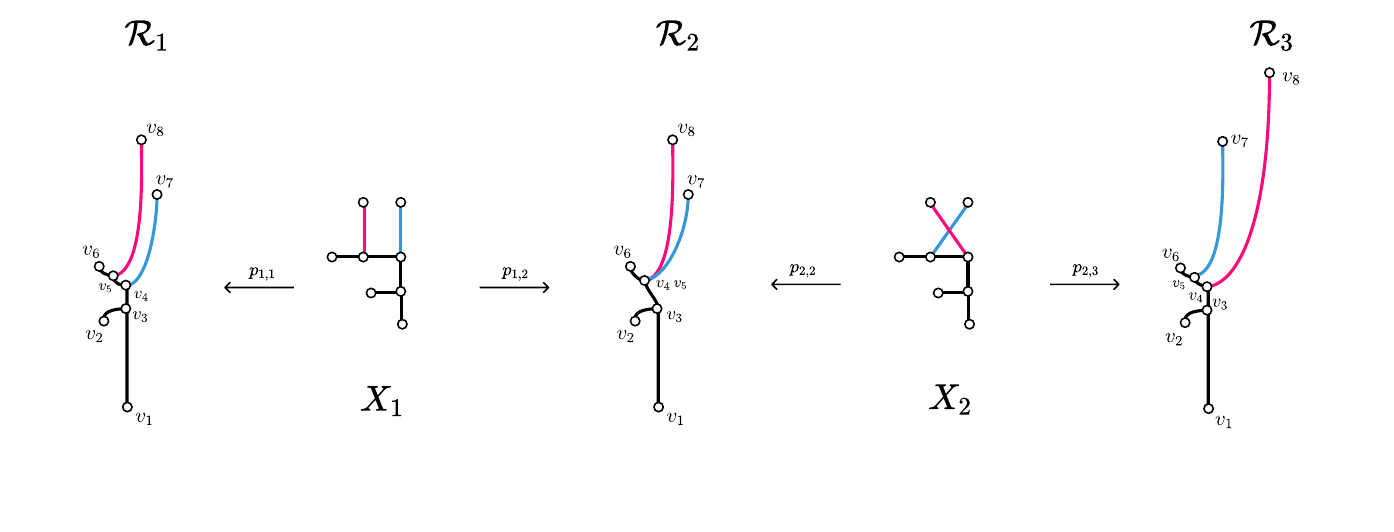}
    \caption{Zigzag diagram depicting the optimal edit sequence carrying $\RR_1 \cong \RRf$ to $\RR_3 \cong \RRg$.}
    \label{fig:exampleFourEdit}
\end{figure}

\paragraph{Bottleneck Distance} The coordinate in the persistence diagram which represents the global maximum and global minimum pairing is $(a_1,a_8)$. Since the global maximum has changed in $(\X,g)$, the pair is now represented as $(a_1,a_9)$. The value $a_9-a_8$ is the largest difference in the best matching for these persistence diagrams, making the bottleneck distance \[\dB{}(\RRf,\RRg) = a_9-a_8.\] 

This is equivalent to the situation where we take the global maximum of $(\X,f)$ and raise it to have function value $a_9$. In essence, this bottleneck distance has measured that the global maximum has increased in value, but \emph{not} necessarily that the global maximum has changed in position.

\paragraph{Interleaving Distance} We have, at this point, gained some intuition that the interleaving distance should be at most determined by the growth in this single maxima. Let $x_f$ represent the leaf of $\RRf$ beginning at $v_5$ and ending at $v_8$, let $y_f$ be the leaf beginning at $v_4$ and ending at $v_7$, $x_g$ be the leaf of $\RRg$ beginning at $u_4$ and ending at $u_8$, and let $y_g$ be the leaf beginning at $u_5$ and ending at $u_7$. Since $x_g$ ends at a higher function value than $x_f$, we know that there cannot exist a $0$-interleaving because $\G(a^{\delta})$, where $a_7 < a < a_9$, will have two components and $\mathcal{S}_0(\F)(a^{\delta}) = \F(a^{\delta})$ will only have one. Because of this, $\e \geq |a_9-a_8|$, since any smaller causes the same situation. Verifying, however, that this is the true interleaving distance requires us to check the continuity around the roots of these two leaves. Since the distance between the root nodes of $x_f$ and $y_f$ is less than $\e$, we are allowed to map the leaf $x_f$ to $x_g$ and $y_f$ to $y_g$ without breaking continuity. Thus, \[d_I(\RRf,\RRg) = a_9-a_8\].

\paragraph{Truncated Interleaving Distance} Just as in the standard interleaving case, to find a $e$-interleaving for truncated interleaving, we need to apply the truncated smoothing operator $\mathcal{S}^{m\e}_{\e}$ such that the resulting graph $\mathcal{S}^{m\e}_{e}(\mathsf{F})$ is stretched enough so that $v_8$ has a new function value of $a_9$. Note that smoothing by $\e = a_9-a_8$ is not enough to construct an interleaving since we will then truncate by $m\e$ -- making the highest function value of $\mathcal{S}^{m\e}_{e}(\mathsf{F})$ to be $a_9 - m(a_9-a_8)$. Thus, we must choose a value of $\e$ for a fixed $m$ such that $\e - m\e \geq a_9-a_8$, implying that $\e \geq \frac{a_9-a_8}{1-m}$.

From Cor.~\ref{cor:strongEquivTruncated}, we know that $d_I(\RRf,\RRg) = a_9-a_8 \leq d^m_I(\RRg,\RRg) \leq \frac{1}{1-m}(a_9-a_8)$. Thus, \[d_I^m(\RRg,\RRg) = \frac{1}{1-m}(a_9-a_8)\].

\paragraph{Functional Distortion Distance} The continuous maps between these Reeb graphs involve mapping the leaf corresponding to the global maximum of $\RRf$ to the leaf corresponding to the global maximum of $\RRg$.  Fig.~\ref{fig:exampleFour}(c) depicts the optimal mapping. Note that we only show the mapping $\Phi$ since the other map $\Psi$ can be easily inferred from $\Phi$. 

In this mapping, the single-edge path from $v_5$ to $v_8$ is mapped to the path from $u_5$ to $u_8$. Additionally, the path from $v_5$ to $u_1$ is mapped to $u_5$ to $u_1$. This means that the path from $u_5$ to $u_4$ is covered twice. Without overlapping these paths (or overlapping two other paths in a similar configuration), we would not have a continuous mapping.

From this mapping, we can see that the functional distortion distance is the change in function value due to mapping $v_8$ to $u_8$, i.e. \[d_{FD}(\RRf,\RRg) = |a_9-a_8|.\] 

One thing to keep in mind is that this scalar field does in fact have a boundary, as compared to every other example that we generated. this can complicate the continuity of our maps $\Phi,\Psi$ if we do not take into account which pieces of the Reeb graph are also boundaries. Since vertex $v_2$ is the minima of $(\X,f)$ which lies on its boundary, we can conclude that the entire leaf containing $v_2$ is a non-open set. Thus, if our map $\Psi$ was constructed in such a way that this was the pre-image of some open set in $\R_g$, it would not in fact be continuous. We avoid this situation by making sure that boundary nodes of $\R_f$ are mapped to boundary nodes of $\R_g$, and vice versa.

\paragraph{Reeb Graph Edit Distance} 
As previously noted in Sec.~\ref{sec:rged}, the older Reeb graph edit distance is defined only on 2-manifolds without boundary.
Since this example does in fact have a boundary, we do not consider the Reeb graph edit distance further, although it is interesting to note that in some cases a sequence of edits is possible in this setting as well. 

\paragraph{Universal Distance} We can trivially deform $\RRf$ to $\RRg$ using a single relabel operation which relabels $v_7$ to have function value $a_9$. However, this leads to a suboptimal value. The best course of action is to switch the roots of the two maxima $v_8$ and $v_7$, then apply a relabel to stretch $v_7$ and $v_8$ to have function values $a_8$ and $a_9$, respectively. Thus, \[\delta_E(\RRf,\RRg) = |a_9-a_8|,\] since the relabeling needed to switch the adjacencies of the leaves is much less. We illustrate this edit sequence in Fig.~\ref{fig:exampleFourEdit}.

\paragraph{Observations} This is our first example where the Reeb graphs are \textbf{contour trees} -- the loop free variant of the Reeb graph, which prevents the existence of points in $\text{Ext}_1$. Furthermore, while each of the distances ended up being equal in value, there are still several intricacies that are involved in computing the Reeb graph metrics. We can view each of these computations as trying to define a matching between the leaves of the contour trees and computing the cost of doing so. For example, in the universal distance case, we assign the leaf $e(v_4,v_7)$ to $e(u_5,u_7)$. However, since the $e(v_4,v_7)$ is not adjacent to $e(v_4,v_6)$ but $e(u_5,u_7)$ \emph{is} adjacent to the corresponding leaf $e(u_5,u_6)$, the adjacencies of the leaves must be changed. In order to do this, our connecting spaces switch the connectivity of the two maxima $v_8$ and $v_7$. This is exactly the same process as in the interleaving and functional distortion distance; the cost of the corresponding ``switch'' operation is less than the cost of the other operations and is therefore is not accounted for. In the bottleneck distance case, no matter how complicated these adjacencies become, they will never be taken into account. 

We note here that the $L^{\infty}$ distance between the scalar fields is $|a_9-a_7| > |a_9-a_8|$. However, it is possible that we have the exact same contour trees where the $L^{\infty}$ distance is also $|a_9-a_8|$. Just in the same way as how the universal distance is computed, changing the total ordering in the saddles that connect to $v_7$ and $v_8$ in $(\sX,\sf)$ would mean that $v_7$ and $v_8$ would switch which saddles they are connected to. Suppose $\e = |a_5-a_4|+\delta$, for some $\delta > 0$. If $v_5$ decreases by $\e$ and $v_4$ increase by $\e$, then the scalar field itself will not have positional changes of the maxima, but the connectivity of the contour tree would change so that $v_7$ would now be connected to $v_5$. Then, increasing $v_7$ to be at $a_8$ and increasing $v_8$ to be at $a_9$ would construct a scalar field whose $L^{\infty}$ distance is only $|a_9-a_8|$. This illustrates the complexity of how various perturbations can cause complex changes in the Reeb graph, yet the Reeb graph metrics still appropriately keep the distance under the $L^{\infty}$ distance.

\section{Discussion}
\label{sec:discussion}

The landscape of Reeb graph metrics is complex and still growing. Research continues to be done in constructing new theoretical Reeb graph metrics as well as trying to compute these distances (or similar distances) on Reeb graphs and other graph-based topological descriptors.

In particular,  Bauer et al. recently released an extended abstract in parallel with our paper discussing several other bounds relating the Reeb graph metrics to one another. They also provide a new distance known as the \textbf{functional contortion distance}, $d_{FC}$. This metric considers all possible maps $\phi: R_f \rightarrow R_g$ and $\psi: R_g \rightarrow R_f$ between two Reeb graphs $R_f$ and $R_g$.  Then, we say $\phi$ and $\psi$ are an $\e$-contortion if for every $x \in R_f$, $\psi(\phi(x))$ is connected to $x$ within an $\e$-neighborhood in $R_f$, and similarly for $y \in R_g$ and its image $\phi(\psi(y))$ in $R_g$.  The \emph{functional contortion distance} is then the infimum over all possible $\e$-contortions.
This distance is an interesting variant on both the functional distortion and the interleaving distances, and the $\e$-neighborhood connectivity is a key property in their proofs of stronger bounds.  

In addition to constructing this novel distance, the authors also spend ample time proving more bounds between the interleaving, functional distortion, universal, and their functional contortion distance. Specifically, the prove the following:

\begin{theorem}[{\cite[Theorem~11]{Bauer2021}}]
\label{thm:functionalcontortion}
The interleaving distance $\di$, the functional distortion distance $\dfd$, and the functional contortion $d_{FC}$ are all strongly equivalent to one the universal distance $\delta_{E}$. Specifically,
\begin{align*}
    \di \leq \du \leq 5\di \hspace{25PX} \dfd \leq \du \leq 3\dfd \hspace{25PX} d_{FC} \leq \du 3 d_{FC} 
\end{align*}
\end{theorem}

The authors of this work in fact conjecture that all inequalities in Thm.~\ref{thm:functionalcontortion} are tight, but a proof of this is not included in the extended abstract. In their presentation of this work, the authors mentioned that they had developed further examples to demonstrate the lower bounds that are not in the extended abstract, but noted that full details will will appear in an upcoming journal version of the paper. 

To the best of our knowledge, there are no examples illustrating that the interleaving distance and functional distortion distance are not equivalent on the space of Reeb graphs. We arrive at the following conjecture.

\begin{conjecture}\label{conjecture:fdd-interleaving-equiv}
The functional distortion and interleaving distance are equal on the space of constructible Reeb graphs.
\end{conjecture}

Despite the possibility of the bounds in Thm.~\ref{thm:functionalcontortion} being tight, we arrive at the conjecture that the functional distortion and interleaving distance are strongly equivalent to the universal distance with smaller bounds.

\begin{conjecture}\label{conjecture:fdd-interleaving-edit-strong-equiv}
The functional distortion distance and interleaving distance are both strongly equivalent to the universal distance defined on the space of PL Reeb graphs. More specifically, we conjecture
\[\dfd \leq \du \leq 2\dfd \hspace{50PX} \di \leq \du \leq 2\di\].
\end{conjecture}

Of course, if Conjecture \ref{conjecture:fdd-interleaving-equiv} is true, then Conjecture \ref{conjecture:fdd-interleaving-edit-strong-equiv} can be reduced to just the single statement \[\dfd = \di \leq \du \leq 2\di = 2\dfd.\]

We now finish our discussion as a series of open questions and challenges that we hope will summarize the current research landscape.

\textbf{Each metric is computationally complex.} To our discontent, the interleaving distance has been shown to be graph isomorphism complete \cite{deSilva2016,Bjerkevik2017}, and there currently exists no polynomial time algorithm for computing the functional distortion distance (nor Gromov Hausdorff distance).

There also exists no polynomial time algorithm for the Reeb graph edit distance \cite{Sridharamurthy2018,Yan2021}. In general, the labeled graph edit distance is known to be NP-hard \cite{Zhiping2009} and also APX-hard \cite{Chich-Long1994}. While the Reeb graph edit distance is slightly different than the general labeled Reeb graph edit distance, this leaves little hope of efficient, exact algorithms for the most general settings.

For the interleaving distance, a glimmer of hope arises with work investigating fixed parameter tractable algorithms \cite{FarahbakhshTouli2019,Stefanou2020}, and comparisons are possible if the Reeb graph has simple enough structure, such as the labeled merge tree \cite{Gasparovic2019,Stefanou2020}.

\textbf{Are there distances defined on the simpler case contour tree or merge tree?}
The \textbf{merge tree} is constructed similarly to the Reeb graph except that the equivalence relation is defined on \emph{sublevel set} rather than levelsets. While the distances we have constructed still apply merge trees, many have devoted research specifically for distances on merge trees: \cite{Morozov2013} constructed the interleaving distance between merge trees and showed that the interleaving distance is bounded below by the 0-dimensional ordinary bottleneck distance; \cite{Gasparovic2019} introduced the intrinsic interleaving distance on labeled merge trees and was subsequently used in practice for uncertainty visualization in \cite{Yan2019a}; \cite{Sridharamurthy2018} introduced an edit distance on merge trees which has experimentally shown promise and is also computationally tractable; \cite{Beketayev2014} constructed a distance on merge trees based on a previously defined notion of a \textbf{branch decomposition tree} with an accompanying computation. The simplified structure of the merge tree has allowed researchers to construct distances which are actually computationally feasible.

The \textbf{contour tree} is a Reeb graph defined on a simply connected domain (see example \ref{example:globalMaxChange}). Each Reeb graph metric is still a well-defined distance for contour trees, but this distances have still proven difficult to compute. Other distances have been constructed specifically for the contour tree, such as \cite{Buchin2017} which utilized the well-studied Fr\'{e}chet distance to do so. 

As with \cref{thm:functionalcontortion}, Bauer et. al provide a proof stating that the functional contortion distance and universal distance are equal on the space of contour trees. However, to the best of our knowledge, the only examples in the literature which show the difference between the universal, functional distortion, and interleaving distances have been for domains which are not simply connected. This leads us to the following conjecture: 

\begin{conjecture}
The functional distortion distance, interleaving distance, and universal distance are equivalent on the space of PL Reeb graphs where the domain $\X$ is simply connected. That is
\[d_B \leq \di = \dfd = \du.\]
\end{conjecture}

Furthermore, Bauer et al. have also provided a proof for equality between the interleaving, functional distortion, functional contortion, and universal distance on the space of merge trees \cite{Bauer2021}.

\textbf{What are the possible applications of these metrics?} 
Reeb graphs and other topological signatures have already been utilized in various areas of data analysis \cite{Yan2021}. The merge tree edit distance was utilized in studying time-dependent scalar fields, 3D cryo electron micoscopy data, shape data, and other synthetic data sets. In \cite{Saikia2017,SaikiaAlgorithm2017}, researchers used a heuristic approach for defining distance between merge trees for feature tracking of time-dependent scalar fields. 

Depending on the application domain, these Reeb graph metrics may be suitable. These distances do seem to be well-suited for the analysis of time-dependent scalar fields since the domains remain consistent and the changes are incremental. We can infer that these distances will perform well in areas that the bottleneck distance has been utilized since, from a data analysis perspective, they share many of the same properties. A case where they may not be suitable is when the scalar fields have vary in scale, since these Reeb graph metrics are sensitive to scales/shifts of the data.

\textbf{Are these Reeb graph metrics well-equipped to measure similarity on different domains?} The construction of each of these distances have allowed us to compare scalar fields with completely different domains. This in general is an extremely desirable property since comparing functions defined on different domains is intrinsically difficult. The property of stability, for example, is only defined for the spaces with identical domains because defining the $L^{\infty}$ distance is only well-defined on identical domains. We can then think of each of these distances as having similar utility to the Gromov-Hausdorff distance, which is a distance metric constructed for spaces defined on arbitrary domains.

However, as we have seen in examples such as Ex.~\ref{example:compound}, the distances do not necessarily have a way to differ between features like holes and leaves. Furthermore, since each distance is constructed to only look at the largest feature, the features that are smaller are not encoded in the distance at all. This leads to situations where a single scalar field with leaves can be equidistant from a single edge Reeb graph and a Reeb graph with multiple loops.

\textbf{Sensitivity to multiple features may be desirable.} From the analysis and examples we have constructed, we can see that there are several areas where each of these distances would lack in an application setting. One of the most notable properties is that each distance is insensitive to the presence of multiple features. This is completely due to each metric involving some sort of ``worst-case" function. The interleaving distance finds the smallest $\e$ to remove $\emph{all}$ features needed, functional distortion distance finds the largest distortion among all distortions for a given set of maps, and the universal distance find the largest displacement of a single vertex. Consider the \textbf{degree-q Wasserstein distance}, which is a variation of the bottleneck distance that is indeed sensitive to multiple features:

\begin{definition}
Let $D_1$ and $D_2$ be two persistence diagrams.The \textbf{degree-$q$ Wasserstein distance} between $D_1$ and $D_2$, for any positive real number $q$, is defined as \[W_q(D_1,D_2) = \Bigg[\inf_{\eta: D_1 \to D_2} \sum_{x \in X}||x - \eta(x)||^q_{\infty} \Bigg]^{1/q}\]
\end{definition}

As we can see, the distance is sensitive to each pairing in the bijection $\eta$ as opposed to just the pairing that creates the largest $L^{\infty}$ distance. However, we understand that adjusting the distances similar to the Wasserstein distance is not a catch-all from a data analysis standpoint. It is easy to construct examples such that the Wasserstein distance to a single object from two distinct objects is equal, despite one object possibly having multiple small features (i.e. multiple small holes in the surface) and the other object having just one large feature (i.e. one large hole in the surface). The Wasserstein distance does, however, require additional constraints on the function $f$ of the scalar field to guarantee a stability result \cite{DeyWang2021}.

Through the examples we have constructed, we can see where the bottleneck distance differs from other Reeb graph metrics in many ways -- specifically in its inability to discern between different Reeb graphs in some cases as well as its treatment of global maxima/minima. Unfortunately, the computational difficulty of these well-constructed Reeb graph distances will prove a challenge for those intending to utilize these in a data analysis context. Between each other, the Reeb graph metrics have some very unintuitive similarities which are difficult to discover without these concrete examples.

\backmatter

\bmhead{Conflict of Interest}

On behalf of all authors, the corresponding author states that there is not conflict of interest.

\backmatter

\bmhead{Acknowledgments}

We thank our anonymous reviewers for their effort in reviewing this manuscript and offering detailed feedback.  
We thank Tim Ophelders for his permission to include Fig.~\ref{fig:graphandSmoothing}. In addition, we thank Ulrich Bauer for the inspiration for Figure~\ref{fig:distanceLandscape} from his talk at SoCG 2020 as well as for clarifications regarding the edit distance, and Håvard Bjerkevik for a helpful discussion on interlevel set persistence.

This work is supported in part by the U.S. Department of Energy, Office of Science, Office of Advanced Scientific Computing Research, under Award Number(s) DE-SC-0019039,
and by the National Science Foundation under grants  
CCF-1907591, CCF-2106578, CCF-2142713
CCF-1907612, CCF-2106672, and DBI-1759807.

\bibliography{sn-bibliography}%

\begin{appendices}

\section{Interlevel set Persistence}
\label{sec:appx:interlevelset}

In this section, we describe another formulation of persistent homology which is also closely related to the Reeb graph.
In particular, this construction has implications for the relationship between the Reeb graph distances and the ungraded bottleneck distance.
These ideas come from taking a categorical viewpoint of persistence; we direct the interested reader to \cite{Riehl2017} for an excellent introduction to the category theory basics. 
In this section, we largely follow notation from \cite{Bjerkevik2021}.

Let $\Int$ be the poset category of connected, open intervals $I = (a,b)$ with a morphism $I \to J$ iff $I \subseteq J$. 
Let $\IntCl$ be the same setup with closed intervals $[a,b]$ instead. 
Note that either of these constructions can also be thought of as a subcategory of $(\R^{op} \times \R)$ where  morphisms are given by $\leq$ in $\R$, and $\geq$ in $\R^{op}$. 
In the case of $\Int$, we have the pairs $(a,b) \in \R^{op} \times \R$ where $a < b$, and for $\IntCl$ we allow $a \leq b$. 
Certain special subsets of $\IntCl$ are called blocks. 
\begin{definition}
\label{def:blocks}
A block is a subset of $\IntCl$ of one of the following forms, where $a,b \in \R \cup \{\pm \infty\}$
\begin{itemize}
    \item $[a,b]_{BL} = 
    \{ (c,d) \in \IntCl \mid c \leq b, d \geq a \}$ 
    \item $[a,b)_{BL} = 
    \{ (c,d) \in \IntCl \mid a \leq d < b \}$ 
    \item $(a,b]_{BL} = 
    \{ (c,d) \in \IntCl \mid a < c \leq b \}$ 
    \item $(a,b)_{BL} = 
    \{ (c,d) \in \IntCl \mid a < c, d < b \}$ 
\end{itemize}
\end{definition}
Since each interval $[a,b] \in \IntCl$ can be viewed as a point in the plane, we can picture the set of possible blocks as in Fig.~\ref{fig:BjerkevikBlocks}. 
Note that the notation is meant to align with the interval seen on the projection to the diagonal; e.g.~block type $[a,b)_{BL}$ intersects the diagonal on interval $[a,b)$, etc.

\begin{figure}%
    \centering
    \includegraphics[width=\textwidth]{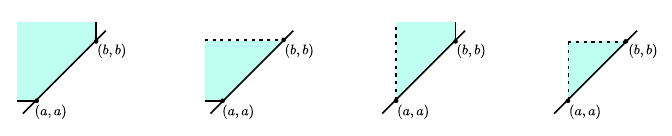}
    \caption{Blocks of an interlevel set persistence module, as given in \ref{def:blocks}.
    From left to right, these are $[a,b]_{BL}$, $[a,b)_{BL}$, $(a,b]_{BL}$, and $(a,b)_{BL}$. }
    \label{fig:BjerkevikBlocks}
\end{figure}

Let $\Vec$ be the category of finite dimensional vector spaces over a field $k$.
Given a Reeb graph $(X,f)$, the $n$-dimensional interlevel set persistence module   is a functor $C^f_n:\IntCl \to \Vec$ given by $I \mapsto H_n(f\inv(I))$.
In the case of constructible data as we are assuming for our Reeb graphs, then $C^f_n$ is always block decomposable, meaning that it can be written as the direct sum of modules supported on blocks. 
In the case of Reeb graphs, the interlevel set information is contained in the $0$-dimensional portion, $C_0^f$.
Further, the ``types'' of points as labeled in extended persistence terminology align with the different kinds of endpoints of the four types of interval representations: $(-,-)_{BL}$ for $\Ext_1$, $[-,-]_{BL}$ for $\Ext_0$, $[-,-)_{BL}$ for $\Ord_0$, and $(-,-]_{BL}$ for $\Rel_1$. 
Consider the example of Fig.~\ref{fig:interlevelset_example}. 
In this case, the two Reeb graphs each have the same, single point in the $\Ext_0$ diagram; and one point in the $\Rel_1$ and $\Ext_1$ diagrams respectively which is nearly identical.
The result is that the interlevel set blocks for the first graph are $\{[a_1,a_5]_{BL}, (a_2,a_4)_{BL}\}$, and for the second they are $\{[a_1,a_5]_{BL}, (a_2,a_3]_{BL}\}$.

\begin{figure}
    \centering
    \includegraphics[width = \textwidth]{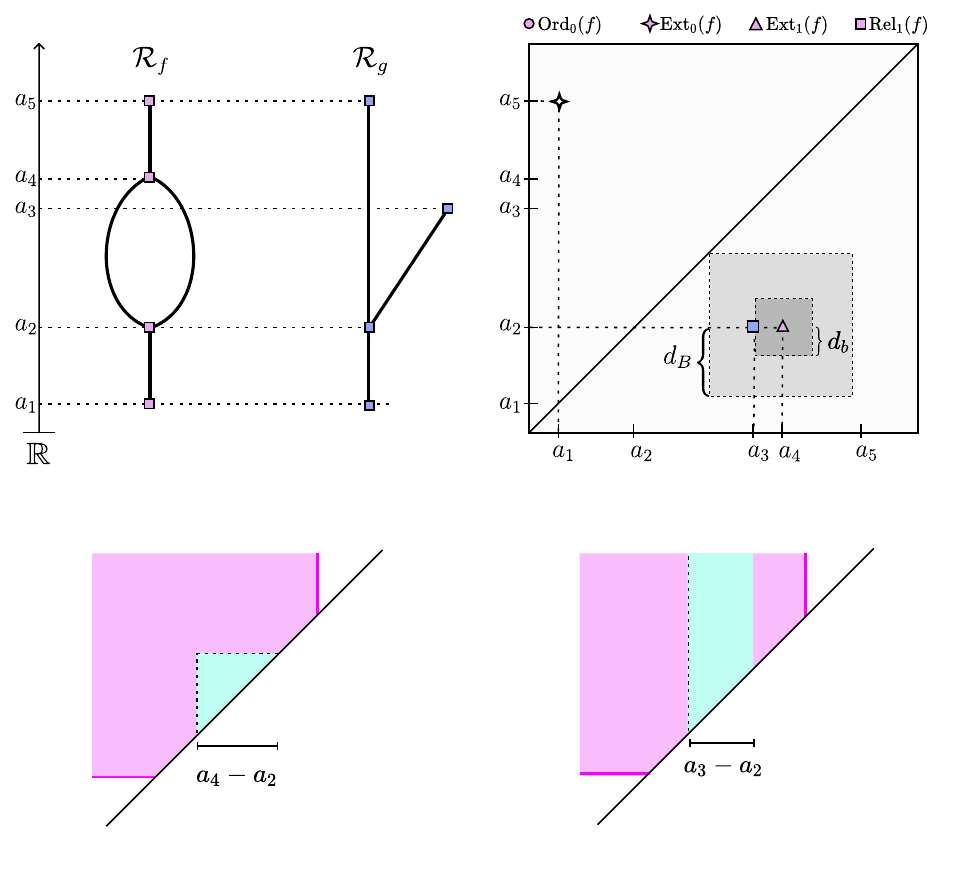}
    \caption{\textbf{(upper-left)} Two Reeb graphs accompanied with their persistence diagrams \textbf{(upper-right)}. Corresponding blocks \textbf{(bottom)} are shown. Blue represents points of the 2-dimensional vector space while purple represents points of the 1-dimensional vector space.}
    \label{fig:interlevelset_example}
\end{figure}

Prior work~\cite{Bjerkevik2016a, Botnan2018} studies these interlevelset persistence from a theoretical point of view, including understanding analogues of interleavings and bottleneck distances.  
They pass to persistence diagrams by looking at the projection of the blocks to the diagonal, and then computing bottleneck distance on the resulting sets of intervals.
In particular, \cite[Thm.~3.8]{Bjerkevik2021}, states the following, with definitions of constructions included and with some mild adjustments to fit this paper's notation.
\begin{theorem}[{\cite[Thm.~3.8]{Bjerkevik2021}}]
\label{thm:BjerkevikFullAppendix}
Let $\mathcal{R}_f$ and $\mathcal{R}_g$ be Reeb graphs. 
We define the following notation for $\mathcal{R}_f$, and have corresponding notation  defined for $\mathcal{R}_g$ replacing $g$ for $f$ wherever it appears.
\begin{itemize}
    \item 
Let $C^f:\Int \to \Set$ be given by $(a,b) \mapsto \pi_0f^{-1}(a,b)$ with morphisms induced by inclusion.
    \item 
Let $C^f_{\Vec}: \Int \to \Vec$ be given by $(a,b) \mapsto H_0(f^{-1}(a,b))$ with morphisms induced by inclusion.
\item Let $\mathcal{B}(C_{\Vec}^f)$ be the block decomposition of $C_{\Vec}^f$.
\item 
Let $\mathcal{B}_{\mathrm{diag}}(C_{\Vec}^f)$ be the collection of intervals in $\R$ representing the blocks of $\mathcal{B}(C_{\Vec}^f)$. 
\end{itemize}
Then 
\begin{equation*}
    \widetilde{d_B}(
    \mathcal{B}_{\mathrm{diag}}(C_{\Vec}^f),
    \mathcal{B}_{\mathrm{diag}}(C_{\Vec}^g)
    \leq 
    2 d_I(C^f, C^g).
\end{equation*}

\end{theorem}

This notation and theorem must come with a caveat, as we have added the notation $\widetilde{d_B}$ to disambiguate this definition with that of Defn.~\ref{def:GradedBottleneck}.
Indeed, in projecting each block to the diagonal, we are left encoding the block type by the endpoints of the interval. 
However, interleavings comparing interval modules of the form with different endpoints, e.g.
\begin{equation*}
    F,G: (\R,\leq) \to \Vec
\end{equation*}
given by 
\begin{equation*}
    F(c) = 
    \begin{cases}
    k & c \in [a,b]\\
    0 & \text{else}
    \end{cases}
\end{equation*}
and 
\begin{equation*}
    G(c) = 
    \begin{cases}
    k & c \in [a,b)\\
    0 & \text{else}
    \end{cases}
\end{equation*}
have interleaving (and thus bottleneck) distance 0.
Meanwhile, when compared as modules defined as the corresponding levelset blocks shown in Fig.~\ref{fig:BjerkevikBlocks}, the distance is not 0. 
The result is that 
$    
\widetilde{d_B}(\mathcal{B}_{diag}(C_{\Vec}^R), \mathcal{B}_{diag}(C_{\Vec}^{R'})
$ 
from Thm.~\ref{thm:BjerkevikFullAppendix} is actually the ungraded bottleneck distance as written, since each block type corresponds to a different type of point in the extended diagram, but the projection to $\R$-valued intervals forgets this information. 

To see this in the example of Fig.~\ref{fig:interlevelset_example},
we have interleaving distance determined by the height of the tail, so
\begin{equation*}
    \tfrac12(a_3-a_2) = d_I(\RR_f,\RR_g).
\end{equation*}
Graded bottleneck distance is determined by the point in the $\Ext_1^f$ diagram, so 
\begin{equation*}
    d_B(\RR_f,\RR_g) = \tfrac{1}{2}(a_4-a_2).
\end{equation*}
However, in projecting the blocks to the diagonal, we have $\mathcal{B}_{\mathrm{diag}}(C_\Vec^f) = \{ [a_1,a_5], (a_2,a_4) \}$ and 
$\mathcal{B}_{\mathrm{diag}}(C_\Vec^f) 
= \{ [a_1,a_5], (a_2,a_3] \}$
and thus  
\begin{equation*}
    \widetilde{d_B}(
    \mathcal{B}_{\mathrm{diag}}(C_{\Vec}^f),
    \mathcal{B}_{\mathrm{diag}}(C_{\Vec}^g)
    =  d_b(\RR_f,\RR_g) = (a_4-a_3) \approx 0. 
\end{equation*}
Of course, we can see in this example that the theorem is satisfied as 
\begin{equation*}
    a_4-a_3 = \widetilde{d_B}(
    \mathcal{B}_{\mathrm{diag}}(C_{\Vec}^f),
    \mathcal{B}_{\mathrm{diag}}(C_{\Vec}^g)
    \leq 2 d_I(\mathcal{R}_f, \mathcal{R}_g) = a_4-a_2.
\end{equation*}
However, as we also notice that in this example, we happen to have 
\begin{equation*}
    \ExDgm(\tilde{f}),\ExDgm(\tilde{g})) = 
    \tfrac{1}{2} (a_4-a_2) \leq  2 d_I(\mathcal{R}_f, \mathcal{R}_g) = a_3-a_2,
\end{equation*}
so it is our goal to replace the ungraded bottleneck distance in the left side of the inequality of Thm.~\ref{thm:BjerkevikFullAppendix} with its graded counterpart.

With gracious thanks to the anonymous reviewer for pointing this out, a similar argument to that presented in \cite{Bjerkevik2021} as proof of \ref{thm:BjerkevikFullAppendix} with a minor modification results in the strengthened result given as Thm.~\ref{thm:strengthenedBottleneck_and_interleaving}. 
Namely, Prop 7.3 of \cite{Botnan2018} shows that the interleavings between block decomposeable modules $\Int \to \Vec$ can be split into interleavings between only blocks of the same type. 
Thus, an interleaving between $C_{\Vec}^R$ and $C_{\Vec}^{R'}$ induces an matching between points in the extended diagram which match type.
However, the distance between the $\R$ intervals can differ from the interleaving distance between the associated blocks by up to a factor of two.
This is caused directly by the fact that open intervals $(a,b)_{BL}$ are $(b-a)/2$ interleaved with the trivial module; while open intervals $(a,b)$ in $\R$ require a $(b-a)$ interleaving with the trivial module. 
The result is that 
\begin{equation*}
d_{B}(\ExDgm(\tilde{f}),\ExDgm(\tilde{g})) \leq 2d_I(C_{\Vec}^{R}, C_{\Vec}^{R'}). 
\end{equation*}
Since an interleaving on $C^R$ induces an interleaving on $C_{\Vec}^R$ by sending the elements of $\pi_0$ to the generators of $H_0$, we have 
\begin{equation*}
    d_I(C_{\Vec}^{R}, C_{\Vec}^{R'}) \leq d_I(C^{R}, C^{R'}).
\end{equation*}
Putting this together gives 
\begin{equation*}
    d_{B}(\ExDgm(\tilde{f}),\ExDgm(\tilde{g})) \leq 2d_I(R_f,R_g)
\end{equation*}
which is the strengthened version of \cite[Thm.~3.8]{Bjerkevik2021} which we have given as Thm.~\ref{thm:strengthenedBottleneck_and_interleaving}.

\section{Multiple Connected Components}\label{sec:appx:multipleConnected}
In the preceeding work, we have assumed for that our Reeb graphs are connected. Furthermore, path componenet sensitivity states that Reeb graph metrics attain a value of $+\infty$ if the Reeb graphs have different numbers of connected components. However, computing the distance between two Reeb graphs with $n>1$ path-connected componenets is, in general, still possible. Yet, this point is often not expanded upon in the literature. Here we adopt the convention that the distance between these Reeb graphs is minimum distance over all possible matchings between the path-connected components of the Reeb graphs.

Suppose we have two Reeb graphs $\RR_f,\RR_g$, both with $n>1$ path-connected components and suppose we want to compute the value of a Reeb graph metric $d$ between $\RR_f$ and $\RR_g$. Let $\pi_0(\RR_f)$ denote the set of path-connected components of $\RR_f$ and let $\RR_f^i$ denote the Reeb graph on the $i^{th}$ path-connected component. Then we define the distance $d$ between $\RR_f,\RR_g$ to be
\[d(\RR_f,\RR_g) = \min_{\sigma:\pi_0(\RR_f) \to \pi_0(\RR_g)}\max_i\big\{d(\RR_f^i,\RR_g^{\sigma(i)}\big\},\]
where we minimize over all possible bijections $\sigma: \pi_0(\RR_f) \to \pi_0(\RR_g)$ between the set of path-connected components.

When computing the bottleneck distance between Reeb graphs with multiple connected components, it is not necessary to adopt this convention since the persistence diagram of a Reeb graph is identical to the union of the persistence diagrams defined on the individual path-connected components. Note that, however, this will be bounded above by the bottleneck distance that \emph{does} use this convention. That is,

\[d_B\Big(\bigcup_i\ExDgm^i(f),\bigcup_i\ExDgm^i(g)\Big) \leq \min_{\sigma:\pi_0(\RR_f) \to \pi_0(\RR_g)}\max_i\big\{d_B(\RR_f^i,\RR_g^{\sigma(i)}\big\},\]
where $\ExDgm^i(f)$ refers to the full extended persistence diagram defined on the $i^{th}$ path-connected componenet of $\RR_f$.

It is important to note that our goal here is to make the bottleneck distance measure features as similarly as the Reeb graph metrics due. In the end, our discrimintavity bounds still hold -- regardless of how close we construct the bottleneck distance to act similarly to the Reeb graph metric.

\section{Truncated Interleaving Distance Properties}\label{sec:appx:truncatedProperties}
The strong equivalence between the interleaving distance and truncated interleaving distance automatically leads us to several properties that the truncated interleaving distance exhibits. We state these properties without proof.

\begin{proposition}
Let $\RR_f$ and $\RR_g$ be two constructible Reeb graphs defined on the same space $\X$. Then For any fixed $m \in [0,1)$, we have
\[d^m_I(\RR_f,\RR_g) \leq \frac{1}{1-m}||f-g||_{\infty}.\]
\end{proposition}

Since this is not quite stability, we are not guaranteed that the universal edit distance bounds the truncated interleaving distance for any fixed $m$. Ex.~\ref{example:globalMaxChange} shows an example of this.

\begin{corollary}
For any fixed $m\in[0,1)$, the truncated interleaving distance is isomorphism indiscernible.
\end{corollary}

\begin{corollary}\label{cor:truncated-interleaving-discrim}
For any fixed $m\in[0,1)$, the truncated interleaving distance is more discriminative than the graded bottleneck distance. That is
\[d_{B}(\RR_f,\RR_g) \leq \frac{9}{1-m}\cdot d^{m}_{I}(\RR_f,\RR_g).\]
\end{corollary}

\begin{corollary}\label{cor:truncated-interleaving-discrim-ungraded}
For any fixed $m\in[0,1)$, the truncated interleaving distance is more discriminative than the ungraded bottleneck distance. That is
\[\db{}(\RR_f,\RR_g) \leq \frac{2}{1-m}\cdot d^{m}_{I}(\RR_f,\RR_g).\]
\end{corollary}

\begin{corollary}
For any fixed $m\in[0,1)$, the truncated interleaving distance is more discriminative than the graded bottleneck distance when the Reeb graphs are defined on simply connected domains. That is
\[d_{B}(\RR_f,\RR_g) \leq \frac{3}{1-m}\cdot d^{m}_{I}(\RR_f,\RR_g).\]
\end{corollary}

\begin{proposition}[{\cite[Proposition~2.19]{Chambers2021}}]
For any fixed $m\in[0,1)$, the truncated interleaving distance is path component sensitive.
\end{proposition}

\end{appendices}

\end{document}